\documentclass[letterpaper, reqno, 10pt]{amsart}

\usepackage[all]{xy}
\usepackage[english]{babel}
\usepackage[utf8]{inputenc}
\usepackage{amssymb,amsmath,amsthm}
\usepackage{amsfonts}
\usepackage{graphicx}
\usepackage{psfrag}
\usepackage[dvipsnames]{xcolor}

\usepackage{xcolor}
\definecolor{winered}{rgb}{0.6,0,0}
\definecolor{lessblue}{rgb}{0,0,0.7}
\usepackage[pdftex,colorlinks=true,linkcolor=winered,citecolor=lessblue,urlcolor=lessblue,breaklinks=true,bookmarksopen=true]{hyperref}

\usepackage{csquotes}
\usepackage[alphabetic]{amsrefs}
\usepackage{enumerate}
\usepackage{cancel}
\usepackage[normalem]{ulem}
\usepackage{lipsum}

\allowdisplaybreaks

\newcommand{\M}{{\mathcal M}}
\newcommand{\id}{\mathrm{id}}
\newcommand{\Ric}{\mathrm{Ric}}
\newcommand{\Scal}{\mathrm{Scal}}

\newcommand{\e}{\varepsilon}
\newcommand{\ep}{\varepsilon}

\renewcommand{\S}{\mathcal S}
\newcommand{\s}{\sigma}

\newcommand{\n}{\nabla}
\renewcommand{\Im}{\mathrm{Im}}
\renewcommand{\Re}{\mathrm{Re}}
\renewcommand{\L}{\mathcal{L}}
\newcommand{\T}{\mathcal{T}}

\newcommand{\bT}{{}^\be T}
\newcommand{\bS}{{{}^\be S^*}}
\newcommand{\boT}{\overline{\bT^*}}
\newcommand{\WFb}{{\mathrm{WF}'_{\mathrm b}}}
\renewcommand{\ell}{\mathrm{Ell}}

\newcommand{\mB}{\mathfrak B}
\newcommand{\mR}{\mathcal R}
\newcommand{\as}{\mathrm a}
\renewcommand{\b}{\beta}

\newcommand{\g}{{\kappa}}

\renewcommand{\j}{{\mathrm j}}
\renewcommand{\l}{{\mathrm l}}
\newcommand{\tl}{{\tilde \l}}
\renewcommand{\o}{\omega}
\renewcommand{\O}{\mathcal O}
\newcommand{\p}{\mathrm p}
\renewcommand{\q}{{\mathrm q}}
\renewcommand{\r}{{\mathrm r}}
\renewcommand{\k}{\mathrm k}
\newcommand{\bs}{\mathrm b}
\newcommand{\f}{\mathrm f}
\renewcommand{\t}{\mathrm t}
\newcommand{\supp}{\mathrm{supp}}
\newcommand{\tr}{\mathrm{tr}}
\renewcommand{\a}{\alpha}
\newcommand{\de}{\delta}
\newcommand{\be}{{\mathrm{b}}}
\renewcommand{\d}{\partial}
\renewcommand{\c}{{\mathfrak c}}
\newcommand{\md}{{\mathrm d}}

\newcommand{\abs}[1]{\left\lvert#1\right\rvert}
\newcommand{\norm}[1]{\left\lVert#1\right\rVert}
\newcommand{\ldr}[1]{\left\langle #1 \right\rangle}
\renewcommand{\div}{\mathrm{div}}

\renewcommand{\H}{\mathrm H}
\newcommand{\sfH}{{\mathsf H}}

\newcommand{\A}{\mathcal A}
\newcommand{\Rc}{{\mathcal R}}
\newcommand{\U}{\mathcal U}

\newcommand{\F}{\mathcal F}

\newcommand{\Vol}{\mathrm{Vol}}
\newcommand{\Char}{\mathrm{Char}}

\newcommand{\R}[0]{\mathbb{R}}
\newcommand{\C}[0]{\mathbb{C}}
\newcommand{\N}[0]{\mathbb{N}}

\newcommand{\D}[0]{\mathcal{D}}
\newcommand{\De}{{\mathrm{D}}}

\newcommand{\Op}{\mathrm{Op}}

\newcommand{\esssupp}{{\mathrm{ess\text{-}supp}}}

\newcommand{\tBox}{\tilde \Box}
\renewcommand{\P}{\mathrm{P}}
\newcommand{\tP}{\tilde \P}
\newcommand{\Aop}{{\mathcal{A}}}
\newcommand{\Bop}{{\mathrm{B}}}
\newcommand{\Jop}{{\mathrm{J}}}
\newcommand{\Fop}{{\mathrm{F}}}
\newcommand{\Kop}{{\mathrm{K}}}
\newcommand{\Sop}{{\mathrm{S}}}
\newcommand{\Lop}{{\mathrm{L}}}
\newcommand{\tLop}{{\tilde \Lop}}
\newcommand{\Qop}{{\mathrm{Q}}}
\newcommand{\Rop}{{\mathrm{R}}}

\theoremstyle{plain}
\newtheorem{thm}{Theorem}[section]
\newtheorem{prop}[thm]{Proposition}
\newtheorem{lemma}[thm]{Lemma}
\newtheorem{cor}[thm]{Corollary}
\theoremstyle{definition}
\newtheorem{definition}[thm]{Definition}

\newtheorem{remark}[thm]{Remark}

\newtheorem{assumption}[thm]{Assumption}

\author{Peter Hintz}
\address{Department of Mathematics, ETH Z\"urich, R\"amistrasse 101, 8092 Z\"urich, Switzerland}
\email{peter.hintz@math.ethz.ch}
\email{phintz@psu.edu}

\author{Oliver Petersen}
\address{Department of Mathematics, Stockholm University, Albanovägen 28, 10691 Stockholm, Sweden}
\email{oliver.petersen@math.su.se}

\author{Andr\'as Vasy}
\address{Department of Mathematics, Stanford University, CA 94305-2125, USA}
\email{andras@math.stanford.edu}

\title[Kerr--de~Sitter stability: full subextremal range]{Conditional non-linear stability of Kerr--de~Sitter spacetimes in the full subextremal range}


\thanks{P.H. is grateful to the Swiss National Science Foundation under grant number TMCG-2\textunderscore{}223167. O.P. gratefully acknowledges support from the Swedish Research Council under grant number 2021-04269. A.V. gratefully acknowledges support from the National Science Foundation under
grant number DMS-2247004.}

\begin{document}
\hbadness=100000
\vbadness=100000

\begin{abstract}
We show the stability of Kerr--de~Sitter black holes, in the full subextremal range, as solutions of the vacuum Einstein
equation with a positive cosmological constant under the assumption that mode stability
holds for these spacetimes. The method is similar to the (unconditional) proof in
the slowly rotating case by Hintz and Vasy. The key novelties are the implementation of
constraint damping in the full subextremal range as well as the verification of a subprincipal symbol
condition at the trapped set.
\end{abstract}

\maketitle

\tableofcontents
\begin{sloppypar}

\section{Introduction}

This paper proves the non-linear stability of Kerr--de~Sitter black hole spacetimes as solutions of Einstein's vacuum equation with a positive cosmological constant $\Lambda>0$,
\[
  \Ric(g)-\Lambda g = 0,
\]
assuming mode stability for the linearized, at Kerr--de~Sitter, equation.
In this equation, $g$ is a Lorentzian metric (of signature $(-,+,+,+)$) on a $4$-dimensional background manifold $M$. Kerr--de~Sitter metrics, which are explicitly described in Section~\ref{subsec: KdS spacetimes}, are solutions of this equation and describe rotating black holes. They come in a family
$g_{b}$, $b=(m,a)$, parameterized by their mass $m$ and specific angular momentum $a$;
subextremality is a condition that $b=(m,a)$ lie in a certain open
set defined by the conditions that the event and cosmological horizons be non-degenerate (i.e., not extremal). Our work
concentrates on a region somewhat bigger than that between these two
horizons, i.e., than what might be called the black hole exterior, or
more precisely the domain of outer communications; this is the
manifold $M$.

When suitably interpreted, namely in a gauge fixed version,
this equation is a tensorial quasilinear wave equation, and hence it is
natural to consider initial data on a spacelike initial
hypersurface $\Sigma_0$. Here we recall that the geometric initial
data are the first and the second fundamental forms of $\Sigma_0$, and
these need to satisfy compatibility conditions, called the constraint equations. The stability statement studied here is that when the initial data
for such a spacetime $g_{b_0}$ are perturbed, i.e., changed slightly in a
suitable norm, then there is a global solution of Einstein's equation
which stays close to the spacetime we started out with. Moreover, the
solution decays, at an exponential rate independent of the initial
data (but depending on $b_0$), to a potentially different (but nearby) member $g_{b}$ of the Kerr--de
Sitter family. Such a statement was proved in the slowly rotating
case, i.e., in a neighborhood of $a=0$ in the subextremal range, by Hintz and Vasy in
\cite{HV2018}. In the present paper we prove that under an additional natural
condition, namely assuming mode stability for the linearized Einstein
equations (which we explain in more detail below), the stability statement in fact extends to
the full subextremal range. (In \cite{HV2018}, the mode
stability itself was proved in the slowly rotating case, i.e., this additional assumption was proved there.) We first state our main theorem informally:

\begin{thm}[Conditional stability of the Kerr--de~Sitter family in the full
  subextremal range; informal version, see Theorem~\ref{thm: main} for
  a precise formulation]
\label{ThmIntroBaby}
  Fix subextremal Kerr--de~Sitter parameters $b_0$. Suppose $(h,k)$ are smooth initial data on $\Sigma_0$, satisfying
  the constraint equations, which are close to the data
  $(h_{b_0},k_{b_0})$ of a Kerr--de~Sitter spacetime in a high
  regularity norm and suppose that mode stability of
  $D_{g_{b_0}}\Ric-\Lambda$ holds; see Assumption~\ref{ass:
  mode stability}. Then there exist a solution $g$ of Einstein's equation attaining these initial data at $\Sigma_0$ and black hole parameters $b$ which are close to $b_0$ such that
  \[
    g - g_b = \O(e^{-\alpha t_*})
  \]
  for a constant $\alpha>0$ independent of the initial data; that is, $g$ decays exponentially fast to the Kerr--de~Sitter metric $g_b$. Moreover, $g$ and $b$ are quantitatively controlled by $(h,k)$.
\end{thm}

Before explaining the mode stability condition, we emphasize that
there are many obstacles to global solvability and asymptotic analysis
of quasilinear evolution equations such as Einstein's equation, after
a suitable gauge is fixed. (There is also the question of the choice of a gauge and
its implementation, which we discuss below.) For
stability questions, the key problem is the analysis of the linearized
equations, linearized at metrics close to (in the norms used) the
metric being perturbed, in decaying function spaces. Versions of a Newton iteration, such as a
Nash--Moser scheme, then imply the desired non-linear statement.

For the
linear equations in an asymptotic geometry like that of Kerr--de~Sitter
spaces, two key issues arise: control of derivatives, i.e., high
frequency behavior, and control of decay. The general scheme initiated
in \cite{V2013} means that these are both controlled if {\em the null-bicharacteristics}, meaning in this case the lifts of null-geodesics to phase space,
{\em are
well-behaved} (roughly, the only trapping is normally hyperbolic
trapping) and if a suitable {\em subprincipal symbol condition} is satisfied
at the trapped set. In this case, solutions of the linearized equation have an asymptotic
expansion, modulo exponential decay, into finitely many mode
solutions. The former condition was shown in the full subextremal regime in
\cites{PV2023,PV2024}, and {\em the latter condition} (which is automatic in
the setting of scalar wave equations) {\em is shown in the present
paper}. Thus, the only obstacle to working in decaying function spaces
is the potential presence of finitely many non-decaying mode solutions.
We hence say that {\em mode stability holds}
if there are no non-decaying mode solutions of the
linearized equation {\em except those corresponding to the natural geometric
  freedom}, namely infinitesimal diffeomorphisms and the linearization of the
Kerr--de~Sitter family in the parameters. These mode solutions should be thought of as asymptotic generalized
eigenfunctions of a(n asymptotic) Killing vector field $T$ which also
satisfy the linearized Einstein equation; see Assumption~\ref{ass:
  mode stability}  for the precise statement. Under the mode
stability assumption, the geometric modes can be incorporated
into the solution scheme (effectively removed by gauge choices plus
incorporated into changing the parameters $b=(m,a)$) and thus allow a
proof of the main theorem. Note
that if there are non-geometric mode solutions that do not decay in time, one cannot expect
solvability of the non-linear equation, as even simple ODE examples
show.

As already mentioned above, Einstein's equation is only well-behaved
upon the imposition of additional gauge conditions. A typical gauge
used is the harmonic, or wave, or DeTurck gauge (\cite{D1981}, see
also \cite{GL91} for additional context, and \cite{F-B1952} for the
original use in the Lorentzian context),
$\Upsilon_{g_{b_0}}(g)=0$; geometrically this means that the
identity map from $(M,g)$ to $(M,g_{b_0})$ is a wave map, i.e., a
critical point of the energy functional, and locally can also be expressed as
the vanishing of the difference of a contracted version of Christoffel symbols,
$g^{ij}({}^{g}\Gamma_{ij}^k-{}^{g_{b_0}}\Gamma_{ij}^k)=0$; see Section~\ref{sec:reduction}. With a suitable
implementation, namely considering
\begin{equation}
\label{EqIGaugeFixed}
\Ric(g)-\Lambda g-\delta_g^*\Upsilon_{g_{b_0}}(g)=0,
\end{equation}
with $\delta_g^*$ the symmetric gradient,
the Einstein equation is directly turned into a non-linear PDE of
the desired type, namely, a non-linear hyperbolic PDE. While this is locally well-behaved, a fundamental
problem is that the global behavior of this equation is not fully
geometric, essentially because the mode solutions of its linearization are not necessarily
those of the linearized version of $\Ric(g)-\Lambda g$. This is so
because the constraint propagation operator, $\delta_g
\T_g\circ\delta_g^*$ (with $\delta_g$ the divergence and $\T_g$ the trace-reversal), which is a one-form wave operator, may itself have
non-decaying mode solutions. (Note here that $\delta_g\T_g(\Ric(g)-\Lambda g)=0$ for any metric $g$, and thus $\delta_g\T_g\circ\delta_g^*$ annihilates $\Upsilon_{g_{b_0}}(g)$. If $\delta_g\T_g\circ\delta_g^*$ has non-decaying mode solutions, it may be that for a solution of~\eqref{EqIGaugeFixed} one has neither $\Ric(g)-\Lambda g=0$ nor $\Upsilon_{g_{b_0}}(g)=0$, even asymptotically.) There are however {\em choices}
here, both regarding {\em the
precise gauge} and also {\em the implementation of the gauge}. Regarding the
former, we use a {\em finite dimensional gauge modification} much as in the
slowly rotating case \cite{HV2018}. Regarding the latter, we use
{\em constraint damping}, as in \cite{HV2018}, but unlike \cite{HV2018} in a
manner that is {\em well-adapted to the specific $g_{b_0}$ around which
we are working}. Constraint damping modifies $\delta_g^*$ by a 0th
order term to $\tilde\delta^*=\delta_{g_{b_0}}+\A_{g_{b_0}}$,
with the goal of making $\delta_g
\T_g\circ\tilde\delta^*$ have no non-decaying modes and thus assuring that
solutions decay. One of the main technical achievements of this paper
is to prove that {\em this constraint damping indeed works in the full
subextremal range}. We remark that constraint damping also plays a key
role in numerical general relativity \cites{BFHR1999,Gund05,Pre2005}, and thus in the detection of
gravitational waves by LIGO and Virgo \cite{LIG2016}.

While we do not address the mode stability issue here, we  mention
that for scalar waves in the ``physically relevant'' subextremal
range, namely under the assumption that $\Lambda m^2$ is small, mode
stability has been shown by Hintz \cite{Hin2024M} by combining results on de
Sitter space (so no black hole) and the known mode stability of Kerr
spacetimes, due to Whiting \cite{W1989} and Shlapentokh-Rothman \cite{Shl2015}. It is expected that
similar methods would yield mode stability even in the tensorial case,
relying in part on the work of Andersson,
H\"afner and Whiting \cite{AHW24} in the $\Lambda=0$ case.

Stability for solutions of Einstein's equation has a long history both
in the $\Lambda>0$ and in the $\Lambda=0$ case. The former is
analytically better behaved, while maintaining the interesting
geometric features like black holes, trapping, etc, while the latter
is an analytically challenging degenerate limit. However,
algebraically the degenerate limit is simpler (as is often the case in
mathematics), and hence in the black hole setting mode stability is known in that case in the full
subextremal range, as referenced above, while this is not so for $\Lambda>0$.

The first general stability result was Friedrich's \cite{Fr86} on the
stability of global de~Sitter space, thus a global symmetric space,
for $\Lambda>0$. The analogous $\Lambda=0$ case, namely Minkowski space, took a few years longer and
was accomplished in the monumental work of Christodoulou and
Klainerman \cite{CK93}, later simplified and extended by Lindblad and
Rodnianski \cite{LR10}
and by Bieri and Zipser \cite{BZ09}, with a more precise description of the
solution, as an asymptotic expansion, given by Hintz and Vasy
\cite{HV20}, see also
Hintz \cite{Hin23E} which would in principle improve on the latter work. Parallel with this, the de~Sitter stability result was
localized by Ringstr\"om \cite{Rin2008}, which also brought the methods used for
$\Lambda>0$ and $\Lambda=0$ closer. The timing for the {\em slowly rotating}
black hole stability is similar: the $\Lambda>0$ case was proved first
by Hintz and Vasy \cite{HV2018}, followed recently by a series of rather involved papers in the
$\Lambda=0$ case by Klainerman,
Szeftel, Giorgi, and Shen \cites{KS20P,KS22C,KS22E,KS23K,GKS24,Sh2023}
as well as, under additional restrictions, by Dafermos, Holzegel, Rodnianski, and
Taylor \cites{DHR19T,DHR19S,DHRT21}. More recently, for $\Lambda>0$,
Fournodavlos and Schlue \cite{FS24} showed
the stability of the {\em cosmological or expanding region} given that of the black hole
exterior (domain of outer communications), i.e., the region between
the event and cosmological horizons, which was shortly followed by a
more localized work (in the spirit of Ringstr\"om \cite{Rin2008}) of
Hintz and Vasy \cite{HV24} which is directly compatible
with the gauges used in the earlier exterior stability work \cite{HV2018}
and indeed applies in the full subextremal range (and in a certain
sense beyond!), thus by the present work is only conditional on mode
stability in this range. In the $\Lambda>0$ case the present work completes the
picture in the full subextremal case apart from the conditional nature
in terms of mode stability. The $\Lambda=0$ case is still a
significant way from achieving a similar resolution in the full
subextremal range, though in that
case, as already mentioned, mode stability is known and moreover
stability for the linearized equation was very recently proved by
H\"afner, Hintz and Vasy \cite{HHV25}.

In the rest of the introduction, we recall the definition and
structure of subextremal Kerr--de~Sitter spacetimes in
Section~\ref{subsec: KdS spacetimes} and state the precise results in Section~\ref{subsec:results}.

\subsection{Kerr--de~Sitter spacetimes} \label{subsec: KdS spacetimes}

\noindent
Before defining subextremal Kerr--de~Sitter metrics, we start by
discussing the key ingredient in this, namely, the polynomial
\begin{equation} \label{eq: mu}
	\mu(r) := \left(r^2 + a^2\right)\left(1 -\frac{\Lambda r^2}3\right) - 2mr.
\end{equation}
The subextremality condition is that $\mu$ has four distinct real roots 
\[
	r_- < r_C < r_e < r_c,
\]
which is equivalent to the discriminant condition
\begin{equation} \label{eq: discriminant}
	- \left( 1 + \frac{\Lambda a^2}3 \right)^4 \left( \frac a m \right)^2 + 12 \left( 1 - \frac{\Lambda a^2}3 \right) \Lambda a^2 + \left(1 - \frac{\Lambda a^2}3 \right)^3 - 9 \Lambda m^2 > 0.
      \end{equation}
      This assures that the zeros are non-degenerate, i.e., $\mu'$
      does not vanish at these. The labeling corresponds to the
      terminology that $r=r_c$ is the cosmological horizon (which
      would be the only horizon in the limiting case $m=a=0$, which is
      de~Sitter space), $r=r_e$ is the event horizon, $r=r_C$ is the
      Cauchy horizon, with these latter two being typically present also in the
      limiting case $\Lambda=0$, which is Kerr spacetime. (There is no Cauchy horizon for $a=0$.) Our focus
      here is on a neighborhood of $[r_e,r_c]$.
      
The Kerr--de~Sitter spacetime, extended over the future event horizon and future cosmological horizon, is given by
\begin{equation} \label{eq: manifold}
	M 
		:= \R_{t_*} \times (r_e - 2\e_M, r_c + 2\e_M) _r \times \mathbb{S}^2_{\phi_*, \theta},
\end{equation}
for some small $\e_M > 0$ and the smooth metric
\begin{equation}
\begin{split}
	g
		=& - \frac{\mu(r)}{b^2\varrho^2}\left(\md t_* - a \sin^2(\theta) \md \phi_* \right)^2 + \varrho^2 \frac{1 - f(r)^2}{\mu(r)} \md r^2 \\
		& - \frac 2 b f(r) (\md t_* - a \sin^2(\theta) \md \phi_*)\md r \\*
		& + \frac{c(\theta)\sin^2(\theta)}{b^2\varrho^2}\left(a\md t_* - \left(r^2 + a^2\right) \md \phi_*\right)^2 + \varrho^2 \frac{\md \theta^2}{c(\theta)},
\end{split}
\label{eq: metric}
\end{equation}
where
\[
	b 
		:= 1 + \frac{\Lambda a^2}3, 
	\quad c(\theta) := 1 + \frac{\Lambda a^2}3 \cos^2(\theta),
	\quad \varrho^2 := r^2 + a^2 \cos^2 (\theta),
\]
and, as above,
\[
	\mu(r) 
		:= \left(r^2 + a^2\right)\left(1 -\frac{\Lambda r^2}3\right) - 2mr,
\]
and where $f\colon (r_e - 2\e_M, r_c + 2\e_M)\to\R$ is such that
\[
	f(r_e) = -1, \quad f(r_c) = 1.
\]
Since $\mu(r_e) = \mu(r_c) = 0$ and $\mu(r) > 0$ for $r \in (r_e, r_c)$, there is at least one 
\[
	r_0 \in (r_e, r_c)
\]
satisfying
\begin{equation} \label{eq: choice of r 0}
	\frac{\md}{\md r} \frac{\mu(r)}{r^2 + a^2}\Big|_{r = r_0}
		= 0,
\end{equation}
which we fix throughout the paper.
Following \cite{PV2023}*{Rmk.~1.1}, we may also assume that $f$ satisfies
\begin{equation}
\label{EqIf}
\begin{split}
	f'(r) 
		&> 0, \quad r \in (r_e - 2\e_M, r_c + 2\e_M), \\
	f(r_0)
		&= 0,
\end{split}
\end{equation}
and that $\md t_*$ is (past) timelike everywhere
on $M$.
 This form of
the metric is obtained from the more familiar Boyer--Lindquist version,
using the standard notation for $t$ and $\phi$ for that, by letting
\begin{equation}
	\begin{split} \label{eq: extended coordinates}
	t_*
		&:= t - \Phi(r), \\
	\phi_*
		&:= \phi - \Psi(r),
	\end{split}
\end{equation}
where $\Phi$ and $\Psi$ satisfy
\begin{align*}
		\Phi'(r)
			&= b \frac{r^2 + a^2}{\mu(r)} f(r), \\
		\Psi'(r)
			&= b \frac a{\mu(r)} f(r);
\end{align*}
see \cite{PV2023}*{Sec.~1.1}.
Note that the metric~\eqref{eq: metric} indeeds extends smoothly to all of $M$, i.e., in particular to the north and south poles of $\mathbb{S}^2$, corresponding to the points $\theta = 0, \pi$ where the coordinates are not defined.
Again by \cite{PV2023}*{Rmk.~1.1}, it is also possible to choose $f: (r_e - 2\e_M, r_c + 2\e_M) \to \R$ to be real analytic, in which case the metric~\eqref{eq: metric} extends to a real analytic metric on all of $M$.
In the real analytic case, the quasinormal modes are also real analytic, as shown in \cite{PV2021}.

\subsection{Main result}\label{subsec:results}

Fix $\Lambda > 0$.
Let $B \subset \R^4$ denote the $4$-dimensional space of parameters $(\vec a, m)$ determining the angular momentum (i.e., the rotation $a = \abs{\vec a}$ and the rotation axis) and mass of a subextremal Kerr--de~Sitter black hole.
We fix the Killing vector field
\begin{equation} \label{eq: T}
	T
		:= \d_{t_*} + \frac a{r_0^2 + a^2}\d_{\phi_*}
\end{equation}
throughout the paper - this will be our choice of the \emph{stationarity} of the Kerr--de~Sitter spacetime, c.f.~\cites{PV2023, PV2024}.

\begin{assumption}[The mode stability assumption] \label{ass: mode stability} 
Given any $b \in B$, let $(M, g_b)$ be the subextremal Kerr--de~Sitter spacetime extended beyond the horizons as in~\eqref{eq: manifold}, with parameters $b$.
\begin{enumerate}
\item Let $\sigma \in \C$ with $\Im(\s) \geq 0$ and $\sigma \neq 0$, and suppose that 
\[
	u
		\in C^\infty \left(M; S^2 T^*M \right)
\]
such that
\begin{align}
	\left( \De \Ric_g - \Lambda \right) u
		&= 0, \nonumber \\
	\left( \L_T + i \s \right) u
		&= 0. \label{eq: QNM assumption}
\end{align}
Then there exists a $1$-form $\o \in C^\infty \left(M; T^*M \right)$ such that
\begin{align}
	u
		&= \de_g^* \o. \nonumber
\end{align}
\item Let $k \in \N_0$ and suppose that $u \in C^\infty \left(M; S^2 T^*M \right)$ is such that
\begin{align}
	\left( \De \Ric_g - \Lambda \right) u
		&= 0, \nonumber \\
	\left( \L_T \right)^k u
		&= 0. \label{eq: QNM assumption 2}
\end{align}
Then there exist a $b' \in T_b B$ and a $1$-form $\o \in C^\infty \left(M; S^2 T^*M \right)$ such that
\begin{align*}
	u
		&= g'(b') + \de_g^* \o,
\end{align*}
where $g'(b')$ is the perturbation in the Kerr--de~Sitter family corresponding to $b'$.
\end{enumerate}
\end{assumption}

\begin{remark}
Note that~\eqref{eq: QNM assumption} is equivalent to the existence of a smooth $u_0$, with $\L_T u_0 = 0$, such that
\[
	u
		= e^{-i\s t_*} u_0,
\]
and~\eqref{eq: QNM assumption 2} is equivalent to the existence of smooth $u_j$, with $\L_T u_j = 0$ for $j = 0, \hdots, {k-1}$, such that
\[
	u
		= \sum_{j = 0}^{k-1} t_*^j u_j.
\]
\end{remark}

\begin{thm}[Conditional non-linear stability; full subextremal range]\label{thm: main}
Fix $b_0 \in B$ and let $(M, g_{b_0})$ be a subextremal Kerr--de~Sitter spacetime, extended beyond the horizons as in~\eqref{eq: manifold}, for which Assumption~\ref{ass: mode stability} holds. 
Let $t_0 \in \R$ and denote
\[
	\Sigma_0
		:= \{t_* = t_0\} \subset M.
\]
There is an $s_0 \in \R$ and an $\a > 0$ such that for all $s \geq s_0$, the following holds:
Let 
\[
	h, k \in C^\infty\left(\Sigma_0; S^2 T^*\Sigma_0 \right)
\]
be initial data satisfying the constraint equations for Einstein's vacuum equations
\begin{equation} \label{eq: constraint equations}
\begin{split}
	\Scal_h - \abs{k}_h^2 + \tr_h(k)^2 
		&= 2 \Lambda, \\
	\div_h k - \md \left( \tr_h k \right) 
		&= 0,
\end{split} 
\end{equation}
and suppose that $(h, k)$ are close to Kerr--de~Sitter spacetime initial data $(h_{b_0}, k_{b_0})$ in
\[
	H^{s+1}(\Sigma_0; S^2 T^*\Sigma_0) \oplus H^s(\Sigma_0; S^2 T^*\Sigma_0).
\]
Then there exist a subextremal Kerr--de~Sitter spacetime $g_b$, for a $b \in B$ close to $b_0$, and a smooth Lorentzian metric $g$ such that
\[
	\Ric_g - \Lambda g = 0, 
\]
$(h, k)$ are the first and second fundamental forms of $g$ on $\Sigma_0$, and
\[
	g - g_b
		\in \O_\infty \left( e^{-\a t_*} \right),
\]
along with all stationary derivatives with respect to $T$, defined in~\eqref{eq: T}.

Moreover, $g$ satisfies the gauge condition
\[
	\Upsilon_{g_{b_0}} \left( g \right) - \Upsilon_{g_{b_0}} \left( g_{b_0, b} \right) - \theta
		= 0,
\]
where 
\begin{equation} \label{eq: interpolating metric}
	g_{b_0, b} 
		:= \chi(t_*) g_b + \left( 1 - \chi(t_*) \right) g_{b_0},
\end{equation}
for a $\chi \in C^\infty(\R)$ such that $\chi(t_*) = 0$ for $t_* \leq t_0 + 1$ and $\chi(t_*) = 1$ for $t_* \geq t_0 + 2$, where $\theta \in \Theta_{b_0}$ and $\Theta_{b_0} \subset C_c^\infty(M; S^2 T^*M)$ is a particular finite-dimensional space associated with $b_0$, and where $\Upsilon$ is defined in~\eqref{eq: Upsilon def} below.
\end{thm}

\begin{remark}
Note here that the proven decay rate $\a > 0$ is independent of the choice of initial data. 
\end{remark}

We complete the introduction by discussing the structure of the
remainder of the paper. First, in Section~\ref{sec:reduction} we
prove that given the background results, the main theorem follows once
two key ingredients are proved: suitable structure of the subprincipal
operator at the trapped set (ST), and constraint damping (CD), i.e.\
that for the constraint propagation equation solutions, and in particular
mode solutions, are decaying. In Section~\ref{sec: spectral gap} we
prove the desired properties of the subprincipal operator at the
trapped set. In Section~\ref{sec: CD} we finally prove the most
demanding ingredient of our proof, namely that constraint damping
works in the full subextremal range. This requires {\em large}
constraint damping, as in \cite{HV2018}, and unlike in Kerr settings;
for convenient references to the existing literature we deal with this
(equivalently) as a semiclassical problem.

\section{Proof of conditional non-linear stability; reduction to ST and CD}\label{sec:reduction}

The purpose of this section is to reduce the proof of Theorem~\ref{thm: main} to Theorem~\ref{thm: spectral gap} and Theorem~\ref{thm: stable constraint propagation} below.
Fix a $\Lambda > 0$.
We are interested in solving Einstein's vacuum equation with a cosmological constant $\Lambda$, given by
\begin{equation} \label{eq: Einstein}
	\Ric_g - \Lambda g
		= 0.
\end{equation}
We use the sign convention $(-,+,+,+)$ for the metric.
Since this equation is not hyperbolic, we need to fix a gauge in order to solve it.
We first fix a subextremal Kerr--de~Sitter metric $g_{b_0}$. 
Instead of directly solving~\eqref{eq: Einstein}, we instead look for a solution $(b, g, \theta)$ of the \emph{gauge fixed Einstein equation}
\begin{equation} \label{eq: gauge fixed Einstein}
	\Ric_g - \Lambda g + \left(\de_{g_{b_0}}^* + \A_{g_{b_0}} \right) \left( \Upsilon_{g_{b_0}}(g) - \Upsilon_{g_{b_0}}\left(g_{b_0, b} \right) - \theta \right)
		= 0,
\end{equation}
where
\begin{itemize}
	\item $g_{b_0, b}$ was defined in~\eqref{eq: interpolating metric},
	\item $\Upsilon_{g_{b_0}}(g)$ is the \emph{gauge one-form}
\begin{equation} \label{eq: Upsilon def}
	\Upsilon_{g_{b_0}}(g)
		:= \ g \left( g_{b_0} \right)^{-1} \de_g \T_g g_{b_0},
\end{equation}
	where we denote the `trace reversal operator' by
\begin{equation} \label{eq: trace reversal}
	\T_g a
		:= a - \frac12 \tr_g(a) g
\end{equation}
for any $a \in S^2T^*M$,
	\item $\A_{g_{b_0}}$ is a homomorphism from $T^*M$ to $S^2 T^*M$,
	\item $\theta$ is a one-form that will be chosen from a fixed finite-dimensional space $\Theta_{b_0}$, which we construct below.
\end{itemize}
The role of $\theta$ is to make sure that the linearized version of~\eqref{eq: gauge fixed Einstein} does not admit any non-decaying mode solutions apart from linear perturbations within the family of Kerr--de~Sitter spacetimes.
Note that the gauge fixed Einstein equation~\eqref{eq: gauge fixed Einstein} is a quasilinear wave equation in $g$ for any fixed $b$ and $\theta$.
Our setup generalizes that of \cite{HV2018} to the full subextremal range of Kerr--de~Sitter black holes.

Assuming that a metric $g$ satisfies~\eqref{eq: gauge fixed Einstein}, one would like to show that it also satisfies~\eqref{eq: Einstein}.
We follow the celebrated method of Choquet-Bruhat in her proof of local existence and geometric uniqueness for the Cauchy problem for the Einstein equations in \cite{F-B1952}, here in the language of de~Turck \cite{D1981}. 
Recall first that the Einstein tensor of any semi-Riemannian metric $g$ is divergence-free (this standard fact follows by contracting the second Bianchi identity appropriately), 
\begin{equation}\label{eq: EE divergence free}
	\de_g \T_g \Ric_g
		= \de_g \left( \Ric_g - \frac12 \Scal_g g \right)
		= 0.
\end{equation}
Therefore, applying $2 \de_g \T_g$ to~\eqref{eq: gauge fixed Einstein}, we conclude that
\begin{align}
	0
		&= 2 \de_g \T_g \left( \de_{g_{b_0}}^* + \A_{g_{b_0}} \right) \left( \Upsilon_{g_{b_0}}(g) - \Upsilon_{g_{b_0}}(g_{b_0, b}) - \theta \right) \nonumber \\
		&= \left( \Box_g + 2\de_g \T_g \A_{g_{b_0}} + 2 \de_g \T_g \left( \de_{g_{b_0}}^* - \de_g^*  \right) - \Lambda \right) \left( \Upsilon_{g_{b_0}}(g) - \Upsilon_{g_{b_0}}(g_{b_0, b}) - \theta \right), \label{eq: gauge vanishing equation}
\end{align}
where we in the last step have used the standard equality
\begin{equation} \label{eq: wave equation identity}
	2 \de_g \T_g \de_g^* 
		= \Box_g - \Ric_g
\end{equation}
for any Lorentzian metric $g$.
Both the terms $2\de_g \T_g \A_{g_{b_0}}$ and $2 \de_g \T_g ( \de_{g_{b_0}}^* - \de_g^* )$ in~\eqref{eq: gauge vanishing equation} are of first order, and in particular of lower order than $\Box_g$.
The function $\theta$ will be supported away from $\Sigma_0$ and $g_{b_0, b}$ coincides with $g_{b_0}$ near $\S_0$ by construction.
Hence, by a standard argument (c.f.~\cite{HV2018}*{Prop.~3.10}), we may choose $g$ so that
\[
	\left( \Upsilon_{g_{b_0}}(g) - \Upsilon_{g_{b_0}}(g_{b_0, b}) - \theta \right)|_{\Sigma_0}
	 	= \Upsilon_{g_{b_0}}(g)|_{\Sigma_0}
	 	= 0.
\]
In that case, the constraint equations will imply that
\[
	\n \left( \Upsilon_{g_{b_0}}(g) - \Upsilon_{g_{b_0}}(g_{b_0, b}) - \theta \right) |_{\Sigma_0}
		= \n \Upsilon_{g_{b_0}}(g) |_{\Sigma_0}
	 	= 0.
\]
Applying the standard uniqueness statement for linear wave equations to~\eqref{eq: gauge vanishing equation} therefore implies that 
\[
	\Upsilon_{g_{b_0}}(g) - \Upsilon_{g_{b_0}}(g_{b_0, b}) - \theta
	 	= 0
\]
on the domain of dependence of $\Sigma_0$.
This fact, together with~\eqref{eq: gauge fixed Einstein}, then finally implies that Einstein's vacuum equation with a cosmological constant $\Lambda$, i.e.,~\eqref{eq: Einstein}, is satisfied.

\subsection{Choice of constraint damping}
As already noticed, the first order term $2\de_g \T_g \A_{g_{b_0}}$ in~\eqref{eq: gauge fixed Einstein} does not change the character of the equation. 
We can therefore choose $\A_{g_{b_0}}$ quite freely; it will be used to \emph{damp} solutions to~\eqref{eq: gauge fixed Einstein}. 
The classical de~Turck gauge, c.f.~\cite{D1981}, is the case when $\A_{g_{b_0}} = 0$.
However, even when proving the non-linear stability of the static part of the de~Sitter spacetime (formally given by the Kerr--de~Sitter metric for $a = m = 0$), choosing $\A_{g_{b_0}} = 0$ leads to the existence of \emph{growing non-gauge mode solutions} for the linearized version of~\eqref{eq: gauge fixed Einstein}, see \cite{HV2018}*{App.~C} for explicit computations.
We prove here that one can choose $\A_{g_{b_0}}$ such that all growing solutions to the linearized version of~\eqref{eq: gauge fixed Einstein} are gauge solutions, for any subextremal Kerr--de~Sitter spacetime $g_{b_0}$.
We refer to this as \emph{constraint damping}; an idea going back to at least \cite{BFHR1999}. 
This played a crucial role in the proof of non-linear stability for slowly rotating Kerr--de~Sitter spacetimes by Hintz and Vasy in \cite{HV2018}.
It is also related to the choice of `gauge source functions' discussed by Friedrich in \cites{F1996}.

Concretely, we choose the homomorphism $\A_{g_{b_0}}$ to be
\begin{equation} \label{eq: Ag choice}
	\A_{g_{b_0}} \o
		:= - h^{-1} \T_{g_{b_0}}^{-1} \left( \c \otimes_s \o - e G_{b_0}(\c, \o) g_{b_0} \right)
\end{equation}
for any $\o \in T^*M$, where $e, h > 0$ are sufficiently small, where $\c$ is a \emph{timelike} one-form and where $G_{g_0}$ is the dual metric of $g_{b_0}$.
This form of $\A_{g_{b_0}}$ was also used in \cite{HV2018}, though the notation was slightly different, but one of the key novelties in this paper is that we make a significantly different choice of $\c$ compared to \cite{HV2018}.
The factor $e > 0$ in~\eqref{eq: Ag choice} is a small number avoiding certain degeneracies when $e = 0$.
It is conceivable that one can prove similar results when $e = 0$, but this seems to require more involved microlocal analysis. 
We choose the dual of the timelike one-form $\c$ to be 
\begin{equation} \label{eq: c sharp first}
	\c^{\sharp}
		:= \d_{t_*} + \frac a{r^2 + a^2} \d_{\phi_*} + \g f(r) \d_r + \g^2 \sin(\theta) \d_\theta,
\end{equation}
for a particular choice of the function $f$ in Section~\ref{subsec: KdS spacetimes}, and where $\g$ is a small parameter ensuring in particular that $\c^\sharp$ is timelike.
Our implementation of constraint damping in the full subextremal range of Kerr--de~Sitter spacetime is one of the key novelties of this paper.

\begin{remark}
The proof of our constraint damping result, Theorem~\ref{thm: stable constraint propagation}, will essentially be based on semiclassical propagation of singularities along $\c^\sharp$ (at zero frequency). 
At the critical points of $\c^\sharp$, the operator will have to satisfy a certain extra criterion, which is in general quite complicated to check.
By choosing $\c^\sharp$ as in~\eqref{eq: c sharp first}, we make sure that $\c^\sharp$ only vanishes on the north and south pole of the $\mathbb{S}^2$, i.e., where $\theta = 0, \pi$. 
At these points the Kerr--de~Sitter metric simplifies significantly, enabling us to check that the criterion is fulfilled (see Lemma~\ref{le: k conformal computation}).
\end{remark}

\subsection{The asymptotic expansion of linear perturbations}
The proof of conditional non-linear stability requires a precise understanding of the linearized problem.
For this, we start by discussing the asymptotic expansion of solutions to the linearization of the gauge fixed Einstein equation~\eqref{eq: gauge fixed Einstein} around a subextremal Kerr--de~Sitter metric $g_{b_0}$.
The linearization of the Ricci curvature at a metric where $\Ric_g = \Lambda g$ is given by (c.f.~\cite{B2007})
\[
	\De \Ric_g ( u ) 
		= \frac12 \Box_g u - \de_g^* \de_g \T_g u + \Rc_g u,
\]
where
\[
	( \Rc_g u)_{\a\b}
		:= - {R_\a}^{\gamma \de}{}_\b u_{\gamma \de} - \Lambda u_{\a \b}.
\]
A straightforward computation shows that $\De \Upsilon_{g}(u) = \de_g \left( \T_g u \right)$.
Since $\Upsilon_{g_{b_0}}(g_{b_0}) = 0$, we conclude that the linearization of~\eqref{eq: gauge fixed Einstein} in $g$, at $b = b_0$ and $\theta = 0$, is given by
\begin{equation} \label{eq: gauge fixed lin}
\begin{split}
	\Lop_{b_0} u
		:= & 2 \De_{g_{b_0}} \left( \Ric_g - \Lambda g + \left( \de_{g_{b_0}}^* + \A_{g_{b_0}} \right) \left( \Upsilon_{g_{b_0}}(g) - \Upsilon_{g_{b_0}}\left(g_{b_0, b} \right) - \theta \right) \right) u \\
		= & \ 2 \left( \De_{g_0} \Ric - \Lambda + \left( \de_{g_{b_0}}^* + \A_{g_{b_0}} \right) \de_{g_{b_0}} \T_{g_{b_0}} \right) u \\
		= & \ \left( \Box_{g_{b_0}} + 2 \A_{g_{b_0}} \de_{g_{b_0}} \T_{g_{b_0}} + 2 \Rc_{g_{b_0}} \right) u,
\end{split}
\end{equation}
where we recall that $\A_{g_{b_0}}, \T_{g_{b_0}}$ and $\Rc_{g_{b_0}}$ are homomorphisms, i.e., differential operators of order zero.
We call $\Lop_{b_0}$ the \emph{linearized gauge fixed Einstein operator}.
The first result towards Theorem~\ref{thm: main} is that this operator has a discrete set of quasinormal modes:

\begin{thm}[Quasinormal modes] \label{thm: QNM}
Let $(M, g_{b_0})$ be a subextremal Kerr--de~Sitter spacetime extended beyond the horizons as in~\eqref{eq: manifold}.
Let $\Lop_{b_0}$ be the linearized gauge fixed Einstein operator defined in~\eqref{eq: gauge fixed lin}.
Then there are sequences
\begin{align*}
	&\s_1, \s_2, \hdots \in \C, \\
	&k_1, k_2, \hdots \in \N,
\end{align*}
such that the following holds:
If $u \in C^\infty(M)$ is a non-trivial solution to
\begin{equation} \label{eq: generalized mode equation}
\begin{split}
	\Lop_{b_0} u
		&= 0, \\
	\left( \L_T + i \s \right)^k u
		&= 0,
\end{split}
\end{equation}
for a $\s \in \C$ and a $k \in \N$, then $\s = \s_j$ and $k \leq k_j$ for some $j \in \N$.
Moreover, for any fixed $\s \in \C$ with $\Im(\sigma) \geq 0$, the vector space of $u \in C^\infty(M)$ solving~\eqref{eq: generalized mode equation}, for some $k \in \N$, is finite dimensional.
Furthermore, only finitely many $\s_j$ have non-negative imaginary part. 
\end{thm}

The very last part of this statement is particularly important, since it says that there is only a finite dimensional set of generalized quasinormal modes that do not decay.
We will prove Theorem~\ref{thm: QNM} together with the second result towards Theorem~\ref{thm: main}, which says that solutions to the linearized gauge fixed Einstein equation have asymptotic expansions in quasinormal modes:

\begin{thm}[The asymptotic expansion] \label{thm: expansion linearized}
Let $(M, g_{b_0})$ and $\Lop_{b_0}$ be as in Theorem~\ref{thm: QNM}.
There exist $s_0 \in \R$, $\a > 0$, and $\de \in (0, 1)$ such that for any $t_0 \in \R$, any $s \geq s_0$ and any $f \in e^{- \a t_*} \bar H^{s - 1 + \de}(M; S^2 T^*M)$ with $\supp(f) \subset \{t_* > t_0\}$, the unique solution to
\[
	\Lop_{b_0} u 
		= f
\]
with $\supp(u) \subset \{t_* > t_0\}$ has an asymptotic expansion
\begin{align*}
	u - \sum_{j = 1}^N u_j 
		\in e^{- \a t_*} \bar H^s(M; S^2 T^*M),
\end{align*}
where $\left( \L_T + i \s_j \right)^{k_j} u_j = 0$ for $\s_j$ and $k_j$ as in Theorem~\ref{thm: QNM}.
\end{thm}

From this theorem, we conclude that solutions to the linearized gauge fixed Einstein equation~\eqref{eq: gauge fixed lin} decay, apart from at most a finite-dimensional set of generalized quasinormal mode solutions of the form
\[
	u_j
		= e^{- i \s_j t_*} \sum_{k = 0}^{k_j-1} t_*^k v_k
		\in \ker\left(\Lop_{b_0} \right)
\]
where $\L_T v_k = 0$.

\begin{remark}
As explained in Section~\ref{subsec: partial compactification} below, there is a natural partial compactification where the coordinate
\[
	\tau
		:= e^{-t_*}
\]
is a boundary defining function of timelike infinity.
With respect to this compactification, the relevant Sobolev spaces are just given by the b-Sobolev spaces 
\[
	e^{- \a t_*} \bar H^s(M; S^2 T^*M)
		= \tau^{\a} \bar H_{\mathrm b}^s\left(M; S^2 \left( \bT^*M \right) \right).
\]
\end{remark}

A key novelty in this paper is an ingredient in the proof of Theorem~\ref{thm: expansion linearized}, namely a bound for the subprincipal part of $\Lop_{b_0}$ at the trapping, proven below in Theorem~\ref{thm: spectral gap}.

\begin{proof}[Proof of Theorem~\ref{thm: QNM} and Theorem~\ref{thm: expansion linearized}, assuming Theorem~\ref{thm: spectral gap}.]
As already noted, the linearized gauge fixed Einstein equation~\eqref{eq: gauge fixed lin} is a linear wave equation.
We would like to apply \cite{PV2023}*{Thm.~1.6}, or indeed its generalization \cite{PV2024}*{Thm.~1.13}.
However, in order to not overcomplicate the presentation, \cite{PV2024}*{Thm.~1.13} was formulated for scalar wave equations with vanishing subprincipal part. 
However, the linearized gauge fixed Einstein equation~\eqref{eq: gauge fixed lin} has a non-trivial subprincipal part, both since it is a tensorial equation and due to the first order term $2 \A_{g_{b_0}} \de_{g_{b_0}} \T_{g_{b_0}}$.
The subprincipal part of the wave operator only affects the radial point estimates in \cite{V2013} by shifting the required regularity by a finite amount and it does not affect the standard propagation of singularities and elliptic estimates.
Therefore, the analysis of Vasy~\cite{V2013}, and the extension to the full subextremal range of Kerr--de~Sitter spacetimes by Petersen--Vasy in \cites{PV2023, PV2024}, applies to conclude the proof of Theorem~\ref{thm: QNM}.

Proving Theorem~\ref{thm: expansion linearized}, the subprincipal part needs to satisfy a suitable bound at the trapped set in order for the microlocal trapping estimates by Wunsch--Zworski in \cite{WZ2011} and by Dyatlov in \cites{Dya2015, Dya2015b, Dya2016} to apply (see also \cite{H2021} and \cite{J2025} for further improvements on this). 
Proving that such a bound indeed satisfied in the full subextremal range of Kerr--de~Sitter spacetimes is one of the novelties in this paper, see Theorem~\ref{thm: spectral gap} below. 
By Theorem~\ref{thm: spectral gap}, the analysis of Vasy~\cite{V2013}, and the extension to the full subextremal range of Kerr--de~Sitter spacetimes by Petersen--Vasy in \cites{PV2023, PV2024}, applies to conclude the proof of Theorem~\ref{thm: expansion linearized}.
\end{proof}

As a corollary, we will prove the asymptotic expansion of solutions for the initial value problem.
For this, we first introduce the functions spaces that will be used throughout the rest of the paper.
\begin{definition}
Fixing a (small) $\e_M > 0$, we define
\[
	\Omega
		:= [t_0, \infty) \times [r_e - \e_M, r_c + \e_M]_r \times \mathbb{S}^2_{\phi_*, \theta}
		\subset M,
\]
with the induced metric $g$.
For $s, \a \in \R$, and a vector bundle $E \to M$, we define the function spaces
\[
	\bar H^{s, \a}(\Omega; E)
		:= e^{-\a t_*}\bar H^{s}(\Omega; E)
\]
and
\[
	D^{s, \a}(\Omega; E)
		:= \bar H^{s, \a}(\Omega; E) \oplus \bar H^{s + 1} \left(\Sigma_0, E|_{\Sigma_0} \right) \oplus  \bar H^s \left(\Sigma_0, E|_{\Sigma_0} \right),
\]
with the norm
\[
	\norm{(f, u_0, u_1)}_{D^{s, \a}}
		:= \norm{f}_{\bar H^{s, \a}_{\mathrm b}} + \norm{u_0}_{\bar H^{s + 1}} + \norm{u_1}_{\bar H^s}.
\]
\end{definition}

Note here that the boundary of $\Omega$ is a manifold with corners, consisting of three boundary faces; the `initial' spacelike hypersurface $\Sigma_0$ and the two `final' spacelike hypersurfaces $\Sigma_e := \{r = r_e - \e_M\}$ and $\Sigma_c := \{r = r_c + \e_M\}$.
The Sobolev spaces $\bar H^s$ are standard translation invariant ones along $T$.
The `bar' notation in $\bar H^{s, \a}$ refers to extendible distributions the three boundary hypersurfaces $\Sigma_0, \Sigma_e$ and $\Sigma_c$.

In order to formulate initial value problems, we will throughout the paper use the notation
\[
	\gamma_0(u)
		:= \left( u, \L_Tu \right)|_{\Sigma_0}.
\]

\begin{cor}[The asymptotic expansion - initial data version] \label{cor: expansion linearized - initial data}
Let $(M, g_{b_0})$ and $\Lop_{b_0}$ be as in Theorem~\ref{thm: QNM}.
There are $s_0 \in \R$ and $\a > 0$ such that for any $t_0 \in \R$, any $s > s_0$ and any $\left( f, u_0, u_1 \right)	\in D^{s, \a}(\Omega; S^2 T^*\Omega)$, the solution to
\[
	(\Lop_{b_0} u, \gamma_0(u)) \\
		= \ \left( f, u_0, u_1 \right),
\]
has an asymptotic expansion
\begin{align*}
	u - \sum_{j = 1}^N u_j 
		\in \bar H^{s, \a}(\Omega; S^2 T^*\Omega),
\end{align*}
where $\left( \L_T + i \s_j \right)^{k_j} u_j = 0$ for $\s_j$ and $k_j$ as in Theorem~\ref{thm: QNM}.
\end{cor}
\begin{proof}
Fix $\chi \in C^\infty(\R)$ such that 
\[
	\chi(t_*) 
		= 
		\begin{cases}
			0, & t_* \leq t_0 + 1, \\
			1, & t_* \geq t_0 + 2.
		\end{cases}
\]
Note that
\[
	u
		= (1-\chi)u + \chi u,
\]
where $(1 - \chi) u$ vanishes for $t_* \geq t_0 + 2$ and
\[
	\Lop_{b_0} \chi u
		= [\Lop_{b_0}, \chi] u + \chi f.
\]
Now, $[\Lop_{b_0}, \chi] u$ vanishes for $t_* \geq t_0 + 2$ and $\chi f$ is supported in $\{t_* \geq t_0 + 1\}$, implying that 
\[
	\Lop_{b_0} \chi u
		\in \bar H^{s, \a}(\Omega; S^2 T^*\Omega).
\]
Therefore the statement follows by Theorem~\ref{thm: expansion linearized}.
\end{proof}

\subsection{Quasinormal modes for the gauge fixed equation}
The next step is to deal with the quasinormal modes appearing in Theorem~\ref{thm: expansion linearized}.
This is where the constraint damping comes in.

\begin{thm}[Non-decaying mode solutions are gauge solutions] \label{thm: growing modes}
Let $(M, g_{b_0})$ be a subextremal Kerr--de~Sitter spacetime extended beyond the horizons as in~\eqref{eq: manifold} for which Assumption~\ref{ass: mode stability} holds. 
Let $\Lop_{b_0}$ be the linearized gauge fixed Einstein operator defined in~\eqref{eq: gauge fixed lin}.
Assume that $u \in C^\infty(M; S^2 T^*M)$ satisfies
\begin{equation} \label{eq: mode of L g}
\begin{split}
	\Lop_{b_0} u
		&= 0, \\
	\left( \L_T + i \s \right)^k u
		&= 0,
\end{split}
\end{equation}
for a $\s \in \C$ with $\Im(\s) \geq 0$ and $k \in \N_0$, where $T$ as in~\eqref{eq: T}.
Then 
\[
	\de_{g_{b_0}} \T_{g_{b_0}} u
		= 0
\]
and the following holds:
\begin{enumerate}
\item If $\s \neq 0$, then there exists an $\o \in C^\infty \left(M; T^*M \right)$ such that
\begin{align*}
	u
		&= \de_{g_{b_0}}^* \o.
\end{align*}
\item \label{item: sigma 0 theorem} If $\s = 0$, then there exist $b' \in T_{b_0} B$ and $\o \in C^\infty \left(M; T^*M \right)$ such that
\begin{align*}
	u
		&= g'(b') + \de_{g_{b_0}}^* \o.
\end{align*}
\end{enumerate}
\end{thm}
We prove this theorem assuming constraint damping (Theorem~\ref{thm: stable constraint propagation}):

\begin{proof}[Proof of Theorem~\ref{thm: growing modes}, assuming Theorem~\ref{thm: stable constraint propagation}]
Following \cite{HV2018}, the idea is to relate the generalized quasinormal modes of $\Lop_{b_0}$ to the generalized quasinormal modes of the linearized Einstein equation.
After reducing the problem to standard (i.e., non-generalized) quasinormal modes, Assumption~\ref{ass: mode stability} would imply the assertion.
Firstly, linearizing the condition~\eqref{eq: EE divergence free}, we conclude that
\[
	\de_{g_{b_0}} \T_{g_{b_0}} \De_{g_{b_0}} \left( \Ric_g - \Lambda g \right)u
		= 0,
\]
for any $u \in C^\infty(S^{2}T^*M)$.
Secondly, recall from~\eqref{eq: gauge fixed lin} that
\[
	\Lop_{b_0}
		= 2 \De_{g_{b_0}} \left( \Ric_g - \Lambda g + \left( \de_{g_{b_0}}^* + \A_{g_{b_0}} \right) \left( \Upsilon_{g_{b_0}} (g) - \Upsilon_{g_{b_0}} (g_{b_0, b}) - \theta \right) \right)
\]
at $b = b_0$ and $\theta = 0$.
Combining these with~\eqref{eq: mode of L g} implies that
\begin{align*}
	\de_{g_{b_0}} \T_{g_{b_0}} \left( \de_{g_{b_0}}^* + \A_{g_{b_0}} \right) \de_{g_{b_0}} \T_{g_{b_0}} u
		&= 0, \\
	\left( \L_T + i \s \right)^k u
		&= 0.
\end{align*}
Since $\L_T$ commutes with $\de_{g_{b_0}} \T_{g_{b_0}}$, we also see that
\begin{align*}
	\left( \L_T + i \s \right)^k \de_{g_{b_0}} \T_{g_{b_0}} u
		&= 0,
\end{align*}
which means that $\de_{g_{b_0}} \T_{g_{b_0}} u$ is a generalized quasinormal mode of frequency $\s$ with respect to
\begin{equation} \label{eq: constraint op}
	\de_{g_{b_0}} \T_{g_{b_0}} \left( \de_{g_{b_0}}^* + \A_{g_{b_0}} \right).
\end{equation}
By Theorem~\ref{thm: stable constraint propagation}, i.e., the constraint damping result,~\eqref{eq: constraint op} \emph{does not admit any non-decaying generalized quasinormal modes}.
Since we have assumed that $\Im(\s) \geq 0$, we therefore conclude that
\[
	\de_{g_{b_0}} \T_{g_{b_0}} u
		= 0,
\]
which, together with~\eqref{eq: mode of L g}, implies that
\begin{equation}
\begin{split} \label{eq: k generalized mode}
	\left( \De \Ric_{g_{b_0}} - \Lambda \right) u
		&= 0, \\
	\left( \L_T + i \s \right)^k u
		&= 0.
\end{split}
\end{equation}
In case $\s = 0$, the assertion is implied by Assumption~\ref{ass: mode stability}. 
In the general case when $\Im(\s) \geq 0$, since $\L_T + i \s$ commutes with $\De \Ric_{g_{b_0}} - \Lambda$, we deduce that
\begin{align*}
	\left( \De \Ric_{g_{b_0}} - \Lambda \right) \left( \left( \L_T + i \s \right)^{k-1} \right) u
		&= 0, \\
	\left( \L_T + i \s \right) \left( \L_T + i \s \right)^{k-1} u 
		&= 0,
\end{align*}
which means that $\left( \L_T + i \s \right)^{k-1} u$ is a standard quasinormal mode.
Hence, if $\Im(\s) \geq 0$ with $\s \neq 0$, Assumption~\ref{ass: mode stability} implies that 
\[
	\left( \L_T + i \s \right)^{k-1} u
		= \de_{g_{b_0}}^* \o,
\]
for an $\o \in C^\infty(T^*M)$.
Defining
\[
	\tilde u
		:= u - \frac1{(k-1)!}\de_{g_{b_0}}^* \left(t_*^{k-1}\o \right),
\]
it follows that
\begin{equation} \label{eq: k - 1 generalized mode}
\begin{split}
	\left( \De \Ric_{g_{b_0}} - \Lambda \right) \tilde u
		&= \left( \De \Ric_{g_{b_0}} - \Lambda \right) u - \frac1{(k-1)!} \left( \De \Ric_{g_{b_0}} - \Lambda \right)\de_{g_{b_0}}^* \left(t_*^{k-1} \o \right) \\
		&= 0, \\
	\left( \L_T + i \s \right)^{k-1} \tilde u
		&= \left( \L_T + i \s \right)^{k-1} u - \frac1{(k-1)!} \left( \L_T + i \s \right)^{k-1} \de_{g_{b_0}}^* \left(t_*^{k-1} \o \right) \\
		&= \de_{g_{b_0}}^* \o - \frac1{(k-1)!} \de_{g_{b_0}}^* \left( \L_T + i \s \right)^{k-1} \left(t_*^{k-1} \o \right) \\
		&= 0.
\end{split}
\end{equation}
We have thus reduced from a generalized quasinormal mode of order $k$ in~\eqref{eq: k generalized mode} to a generalized mode of order $k-1$ in ~\eqref{eq: k - 1 generalized mode}.
Iterating this argument and finally applying Assumption~\ref{ass: mode stability} proves the assertion also in this case.
This completes the proof.
\end{proof}

\subsection{Linear stability}

We now turn to the proof of linear stability of any subextremal Kerr--de~Sitter spacetime.
With the non-linear iteration procedure in mind, presented in the next subsection, we will prove a result which is more general than strictly needed for linear stability.
As in Theorem~\ref{thm: main}, for any $b_0, b \in B$, let
\[
	g_{b_0, b} 
		:= \chi(t_*) g_b + \left( 1 - \chi(t_*) \right) g_{b_0},
\]
for a fixed $\chi \in C^\infty(\R)$ such that 
\[
	\chi(t_*) 
		= 
		\begin{cases}
			0, & t_* \leq t_0 + 1, \\
			1, & t_* \geq t_0 + 2.
		\end{cases}
\]
The first part of the non-linear differential operator we will consider is the evolution operator
\begin{align*}
	\P_1(\tilde g, b, \theta) 
		:= & \ \left( \Ric - \Lambda \right) (g_{b_0, b} + \tilde g) \\*
		& \ + \left(\de_{g_{b_0}}^* + \A_{g_{b_0}} \right) \left( \Upsilon_{g_{b_0}}(g_{b_0, b} + \tilde g) - \Upsilon_{g_{b_0}}(g_{b_0, b}) - \theta \right).
\end{align*}
This is the same operator as in~\eqref{eq: gauge fixed Einstein}, just writing $g = g_{b_0, b} + \tilde g$.
The linearization is given by
\begin{equation} \label{eq: D P1}
\begin{split}
	\De_{(\tilde g, b, \theta)} \P_1(\tilde g', b', \theta')
		:= & \ \De_{(g_{b_0, b} + \tilde g)} \left( \Ric - \Lambda \right) ( \chi g_{b_0}'(b') + \tilde g' ) \\*
		& \ + \left(\de_{g_{b_0}}^* + \A_{g_{b_0}} \right) \De_{(g_{b_0, b} + \tilde g)} \Upsilon_{g_{b_0}}( \chi g_{b_0}'(b') + \tilde g' ) \\
		& \ - \left(\de_{g_{b_0}}^* + \A_{g_{b_0}} \right) \left( \De_{g_{b_0, b}}\Upsilon_{g_{b_0}} (\chi g_{b_0}'(b')) + \theta' \right).
\end{split}
\end{equation}
This specializes at a subextremal Kerr--de~Sitter spacetime, i.e., where $(\tilde g, b, \theta) = (0, b_0, 0)$, to
\begin{align*}
	\De_{(0, b_0, 0)} \P_1(\tilde g', b', \theta')
		= & \ \De_{g_{b_0}} \left( \Ric - \Lambda \right) ( \chi g_{b_0}'(b') + \tilde g' ) \\*
		& \ + \left(\de_{g_{b_0}}^* + \A_{g_{b_0}} \right) \De_{g_{b_0}} \Upsilon_{g_{b_0}}( \tilde g' - \theta' ) \\
		= & \ \frac12 \Lop_{b_0} \tilde g' + \left( \De_{g_{b_0}} \Ric - \Lambda \right) \chi g_{b_0}'(b') + \left(\de_{g_{b_0}}^* + \A_{g_{b_0}} \right) \theta'.
\end{align*}
The second operator will describe the initial condition:
\begin{align*}
	\P_2(\tilde g, b) 
		:= & \ \gamma_0(\tilde g + g_{b_0, b}) - \iota_{b_0}(h, k) \\
		= & \ \gamma_0(\tilde g) - \iota_{b_0}(h, k),
\end{align*}
where $(h,k)$ is a fixed solution to the constraint equations~\eqref{eq: constraint equations} and are prescribed, c.f.~Theorem~\ref{thm: main}, and $\iota_b$ is a certain map taking the geometric initial data $(h, k)$ to wave equation initial data $(\tilde g_0, \tilde g_1)$ at $\Sigma_0$ introduced in \cite{HV2018}*{Prop.~3.10}.
The linearization is given by
\begin{equation} \label{eq: D P2}
\begin{split}
	\De_{(\tilde g, b)}\P_2(\tilde g', b') 
		:= & \ \gamma_0(\tilde g' + \chi g_{b_0}'(b')) \\
		= & \ \gamma_0(\tilde g'),
\end{split}
\end{equation}
which at a subextremal Kerr--de~Sitter spacetime becomes
\begin{align*}
	\De_{(0, b_0)}\P_2(\tilde g', b')
		= & \ \gamma_0(\tilde g').
\end{align*}

Let us now define the finite dimensional space 
\[
	\Theta_{b_0} 
		\subset C_c^\infty(M; T^*M),
\]
from which $\theta$ will be chosen.
For a fixed $b_0 \in B$, let first $\s_1, \hdots, \s_N \in \C$ be all quasinormal mode frequencies in Theorem~\ref{thm: QNM} with $\Im(\s) \geq 0$, and relabel so that $\s_1 = 0$ and $\s_j \neq 0$ for $j \neq 1$.
Since, for any fixed frequency $\s_j$, the space of generalized solutions has finite dimension $l_j$, we may choose a basis
\[
	u_{j1}, \hdots, u_{jl_j}.
\]
Theorem~\ref{thm: growing modes} implies that, for each $j = 2, \hdots, N$, there are generalized quasinormal mode one-forms $\o_{j l}$ such that 
\[
	u_{jl}
		= \de_{g_{b_0}}^*\o_{jl},
\]
and 
\begin{equation} \label{eq: omega gauged}
	\de_{b_0} \T_{g_0}\de_{g_{b_0}}^*\o_{jl} 
		= 0,
\end{equation}
for $j = 2, \hdots, N$ and $l = 1, \hdots, l_j$.
Moreover, for each $l = 1, \hdots, l_1$, there exist $b_l' \in T_{b_0}B$ and a generalized quasinormal mode one-form $\o_l$ such that
\[
	u_{1, l}
		= g'_{b_0}(b_l') + \de_{g_{b_0}}^* \o_l
\]
and
\begin{equation} \label{eq: g' gauged}
	\de_{g_{b_0}} \T_{g_{b_0}} \left( g'_{b_0}(b_l') + \de_{g_{b_0}}^* \o_l \right)
		= 0.
\end{equation}
Using this, we define
\begin{align*}
	\theta_{jl}
		:= & \ [\de_{g_{b_0}} \T_g \de_g^*, \chi] \o_{jl} 
		\in C_c^\infty(M), \\
	\theta_l
		:= & \ [\de_{g_{b_0}} \T_g \de_g^*, \chi] \o_l + [\de_{g_{b_0}}, \chi] \T_{g_{b_0}} g_{b_0}'(b_l')
		\in C_c^\infty(M).
\end{align*}
This allows us to define
\begin{align*}
	\Theta_{b_0}
		:= & \ \mathrm{span} \left( \theta_l \colon l = 1, \hdots, l_1 \right) \oplus \mathrm{span} \left( \theta_{j,l} \colon j = 2, \hdots, N, l = 1, \hdots, l_j \right) \\
		& \ \subset C_c^\infty(M; T^*M).
\end{align*}

\begin{remark}
Note that all elements in $\Theta_{b_0}$ are supported away from $\Sigma_0$, as needed for the Choquet-Bruhat-type argument presented of this section.
\end{remark}

In view of the Nash--Moser iteration for the non-linear stability proof, we want to formulate our results as proving the existence of (and later on also tame estimates for) a \emph{right inverse} $\Rop_{(\tilde g, b, \theta)}$ to the \emph{linearization of}
\[
	\P{(\tilde g, b, \theta)}
		:= \left( \P_1{(\tilde g, b, \theta)}, \P_2{(\tilde g, b)} \right).
\]
The first step, which is the main goal in this subsection, is to construct $\Rop_{(\tilde g, b, \theta)}$ at the Kerr--de~Sitter initial data, i.e., where $(\tilde g, b, \theta) = (0, b_0, 0)$.
For this, let first 
\[
	(f, u_0, u_1) \in D^{s, \a}(\Omega; S^2 T*\Omega).
\]
Let $u$ denote the unique solution to the initial value problem
\begin{align*}
	\left( \Lop_{b_0}, \gamma_0 \right) u 
		&= (f, u_0, u_1).
\end{align*}
By Corollary~\ref{thm: expansion linearized}, there are unique numbers $c_{jl} \in \C$ for $j = 2, \hdots, N$ and $l = 1, \hdots, l_j$, and $c_l \in \C$ for $l = 1, \hdots, l_1$ such that
\begin{equation} \label{eq: solution asymptotics}
	u - \de^*_{g_{b_0}} \left( \sum_{j = 2}^N \sum_{l = 1}^{l_j} c_{jl} \o_{jl} \right) - \sum_{l = 1}^{l_1} c_l \left( g'_{b_0}(b_l') + \de_{g_{b_0}}^* \o_l \right)
		\in \bar H^{s, \a}(\Omega; S^2 T^*\Omega).
\end{equation}
We define
\[
	\Rop_{(0, b_0, 0)}(f, u_0, u_1)
		:= (\tilde g', b', \theta'),
\]
where
\begin{equation} \label{eq: R at KdS}
\begin{split}
	\tilde g'
		:= & \ u - \de^*_{g_{b_0}} \left( \sum_{j = 2}^N \sum_{l = 1}^{l_j} c_{jl} \chi \o_{jl} \right) - \sum_{l = 1}^{l_1} c_l \left( \chi g'_{b_0}(b_l') + \de_{g_{b_0}}^* \chi \o_l \right) \in \bar H^{s, \a}(\Omega; S^2 T^*\Omega), \\
	b'
		:= & \ \sum_{l = 1}^{l_1} c_l b_l' \in T_{b_0}B, \\
	\theta'
		:= & \ \sum_{j = 2}^N \sum_{l = 1}^{l_j} c_{jl} \theta_{jl} + \sum_{l = 1}^{l_1} c_l \theta_l
		\in \Theta_{b_0}.
\end{split}
\end{equation}
Note that the asymptotics for $\tilde g'$ are the same as those for $u$ the expression in~\eqref{eq: solution asymptotics}, so $\tilde g'$ does indeed decay of order $e^{-\a t_*}$.

\begin{thm}[A right inverse at the Kerr--de~Sitter spacetime] \label{thm: linearized surjectivity}
Let $(M, g_{b_0})$ be a subextremal Kerr--de~Sitter spacetime extended beyond the horizons as in~\eqref{eq: manifold} for which Assumption~\ref{ass: mode stability} holds.
There exist $s_0 \in \R$ and $\a > 0$ such that if $s \geq s_0$, then 
\begin{align*}
	\Rop_{(0, b_0, 0)} \colon D^{s, \a}(\Omega; S^2 T^*\Omega)
		&\to \bar H^{s, \a}(\Omega, S^2T^*\Omega) \times T_{b_0} B \times \Theta_{b_0}
\end{align*}
is a continuous map and such that
\[
	\De_{(0, b_0, 0)}\P \circ \Rop_{(0, b_0, 0)}(f, u_0, u_1)
			= (f, u_0, u_1)
\]
for all 
\[
	(f, u_0, u_1) \in D^{\infty, \a}(\Omega; S^2 T^*\Omega).
\]
\end{thm}

\begin{proof}
Let us write
\[
	\Rop_{(0, b_0, 0)}(f, u_0, u_1)
		:= (\tilde g', b', \theta'),
\]
where $(\tilde g', b', \theta')$ are defined in~\eqref{eq: R at KdS}.
It is by construction clear that $\Rop_{(0, b_0, 0)}(f, u_0, u_1)$ is continuous.
Note first that $\tilde g'$ induces the same initial data at $\Sigma_0$ as $u$ does, since $\chi$ is supported away from $\Sigma_0$, implying that
\[
	\De_{(0, b_0, 0)}\P_2 \circ \Rop_{(0, b_0, 0)}(f, u_0, u_1)
			= (u_0, u_1)
\]
as claimed.
We now check that indeed
\[
	\De_{(0, b_0, 0)}\P_1 \circ \Rop_{(0, b_0, 0)}(f, u_0, u_1)
			= f
\]
for all $f \in \bar H^{\infty, \a}(\Omega; S^2T^*\Omega)$.
Using that $\De \Ric_g(\de^*_g \o) = 0$ for any one-form $\o$, together with~\eqref{eq: omega gauged} and~\eqref{eq: g' gauged}, we note for $(\tilde g', b', \theta')$ as in~\eqref{eq: R at KdS} that
\begin{align*}
	&\De_{(0, b_0, 0)} \P_1(\tilde g', b', \theta') \\
		&\quad = \frac12 \Lop_{b_0} \tilde g' + \left( \De_{g_{b_0}} \Ric - \Lambda \right) \chi g_{b_0}'(b') + \left(\de_{g_{b_0}}^* + \A_{g_{b_0}} \right) \theta' \\
		&\quad = \frac12 \Lop_{b_0} u - \frac12 \Lop_{b_0} \de^*_{g_{b_0}} \left( \sum_{j = 2}^N \sum_{l = 1}^{l_j} c_{jl} \chi \o_{jl} + \sum_{l = 1}^n c_l \chi \o_l \right) \\*
		&\quad \qquad - \frac12 \Lop_{b_0} \sum_{l = 1}^{l_1} c_l \chi g'_{b_0}(b_l') + \left( \De_{g_{b_0}} \Ric - \Lambda \right) \chi g_{b_0}'(b') + \left(\de_{g_{b_0}}^* + \A_{g_{b_0}} \right) \theta' \\
		&\quad = f + \left( \de_{g_{b_0}}^* + \A_{g_{b_0}} \right) \de_{g_{b_0}} \T_{g_{b_0}} \de_{g_{b_0}}^* \left( \sum_{j = 2}^N \sum_{l = 1}^{l_j} c_{jl} \chi \o_{jl} + \sum_{l = 1}^{l_1} c_l \chi \o_l \right) \\*
		&\quad \qquad - \frac12 \Lop_{b_0} \chi g'_{b_0}(b') + \left( \De_{g_{b_0}} \Ric - \Lambda \right) \chi g_{b_0}'(b') + \left(\de_{g_{b_0}}^* + \A_{g_{b_0}} \right) \theta' \\
		&\quad = f + \left( \de_{g_{b_0}}^* + \A_{g_{b_0}} \right) \left( \sum_{j = 2}^N \sum_{l = 1}^{l_j} c_{jl} [\de_{g_{b_0}} \T_g \de_g^*, \chi] \o_{jl} + \sum_{l = 1}^{l_1} c_l [\de_{g_{b_0}} \T_g \de_g^*, \chi] \o_l \right) \\*
		&\quad \qquad + \left( \de_{g_{b_0}}^* + \A_{g_{b_0}} \right) \left( \sum_{l = 1}^{l_1} c_l \chi \de_{g_{b_0}} \T_{g_{b_0}} \de_{g_0}^* \o_l \right) \\*
		&\quad \qquad + \left( \de_{g_{b_0}}^* + \A_{g_{b_0}} \right) \de_{g_{b_0}} \T_{g_{b_0}} \chi g_{b_0}'(b') - \left( \de_{g_{b_0}}^* + \A_{g_{b_0}} \right) \theta' \\
		&\quad = f + \left( \de_{g_{b_0}}^* + \A_{g_{b_0}} \right) \left( \sum_{j = 2}^N \sum_{l = 1}^{l_j} c_{jl}\theta_{jl} + \sum_{l = 1}^{l_1} c_l [\de_{g_{b_0}} \T_g \de_g^*, \chi] \o_l \right) \\*
		&\quad \qquad + \left( \de_{g_{b_0}}^* + \A_{g_{b_0}} \right) \left( \sum_{l = 1}^{l_1} c_l \chi \de_{g_{b_0}} \T_{g_{b_0}} \de_{g_0}^* \o_l + \chi \de_{g_{b_0}} \T_{g_{b_0}} g_{b_0}'(b') \right) \\*
		&\quad \qquad + \left( \de_{g_{b_0}}^* + \A_{g_{b_0}} \right) [\de_{g_{b_0}}, \chi] \T_{g_{b_0}} g_{b_0}'(b') - \left( \de_{g_{b_0}}^* + \A_{g_{b_0}} \right) \theta' \\
		&\quad = f + \left( \de_{g_{b_0}}^* + \A_{g_{b_0}} \right) \sum_{j = 2}^N \sum_{l = 1}^{l_j} c_{jl}\theta_{jl} \\*
		&\quad \qquad + \left( \de_{g_{b_0}}^* + \A_{g_{b_0}} \right) \sum_{l = 1}^{l_1} c_l \left( [\de_{g_{b_0}} \T_g \de_g^*, \chi] \o_l + [\de_{g_{b_0}}, \chi] \T_{g_{b_0}} g_{b_0}'(b_l') \right) \\*
		&\quad \qquad + \left( \de_{g_{b_0}}^* + \A_{g_{b_0}} \right) \chi \left( \sum_{l = 1}^{l_1} c_l \de_{g_{b_0}} \T_{g_{b_0}} \left( \de_{g_0}^* \o_l +  g_{b_0}'(b_l') \right) \right) \\*
		&\quad \qquad - \left( \de_{g_{b_0}}^* + \A_{g_{b_0}} \right) \theta' \\
		&\quad = f + \left( \de_{g_{b_0}}^* + \A_{g_{b_0}} \right) \left( \sum_{j = 2}^N \sum_{l = 1}^{l_j} c_{jl}\theta_{jl} + \sum_{l = 1}^{l_1} c_l \theta_l \right) - \left( \de_{g_{b_0}}^* + \A_{g_{b_0}} \right) \theta' \\
		&\quad = f.
\end{align*}
This completes the proof.
\end{proof}

The linear stability is a simple consequence.

\begin{cor}[Linear stability] \label{cor: linear stability}
Let $(M, g_{b_0})$ be a subextremal Kerr--de~Sitter spacetime extended beyond the horizons as in~\eqref{eq: manifold} for which Assumption~\ref{ass: mode stability} holds.
There exists $\a > 0$ such that given any smooth solution $(h', k')$ to the linearized constraint equations at $\Sigma_0$, there is a smooth solution $g'$ to
\[
	\left( \De_{g_{b_0}} \Ric - \Lambda \right) g'
		= 0,
\]
such that $(h', k')$ are the induced linearized first and second fundamental forms and
\[
	g' - g'_{b_0}(b')
		\in \O_\infty\left(e^{-\a t_*} \right)
\]
for some $b' \in T_{b_0}B$, along with all its stationary derivatives with respect to $T$.
\end{cor}

\begin{proof}[Proof of Corollary~\ref{cor: linear stability}]
We apply the Theorem~\ref{thm: linearized surjectivity} to
\[
	(u_0, u_1)
		= \De_{(h_{b_0}, k_{b_0})}\iota_{b_0}(h',k')
\]
and $f = 0$.
This produces a smooth solution $(\tilde g', b', \theta')$ such that
\begin{align*}
	\De_{(0, b_0, 0)}\P_1(\tilde g', b', \theta')
		&= 0, \\
	\De_{(0, b_0)}\P_2(\tilde g', b')
		&= \De_{(h_{b_0}, k_{b_0})}\iota_{b_0}(h',k'),
\end{align*}
or in other words
\begin{equation} \label{eq: lin gauged equations}
\begin{split}
	0 
		&= \De_{g_{b_0}} \left( \Ric - \Lambda \right) (g') + \left(\de_{g_{b_0}}^* + \A_{g_{b_0}} \right) \De_{g_{b_0}} \Upsilon_{g_{b_0}}( \tilde g' - \theta' ), \\
	\gamma_0(g')
		&= \De_{(h_{b_0}, k_{b_0})}\iota_{b_0}(h',k'),
\end{split}
\end{equation}
where we have defined
\[
	g'
		:= \tilde g' + \chi g_{b_0}'(b').
\]
By the identity
\[
	\de_{g_{b_0}} \T_{g_{b_0}} \De_{g_{b_0}} \left( \Ric - \Lambda \right)(g')
		= 0,
\]
it follows that
\[
	\de_{g_{b_0}} \T_{g_{b_0}} \left(\de_{g_{b_0}}^* + \A_{g_{b_0}} \right) \De_{g_{b_0}} \Upsilon_{g_{b_0}}( \tilde g' - \theta' )
		= 0,
\]
which is a homogeneous \emph{wave equation} for $\De_{g_{b_0}} \Upsilon_{g_{b_0}}( \tilde g' - \theta' )$.
Moreover, by construction of the map $\iota_{b_0}$, we know that 
\[
	\De_{g_{b_0}} \Upsilon_{g_{b_0}}( \tilde g' - \theta' )|_{\Sigma_0}
		= \De_{g_{b_0}} \Upsilon_{g_{b_0}}( g' )|_{\Sigma_0}
		= 0.
\]
Further, the linearized constraint equations imply that
\[
	\n \De_{g_{b_0}} \Upsilon_{g_{b_0}}( \tilde g' - \theta' )|_{\Sigma_0}
		= \n \De_{g_{b_0}} \Upsilon_{g_{b_0}}( g')|_{\Sigma_0}
		= 0.
\]
Therefore uniqueness of solutions to wave equations imply that
\[
	\De_{g_{b_0}} \Upsilon_{g_{b_0}}( \tilde g' - \theta' )
		= 0,
\]
and we conclude from~\eqref{eq: lin gauged equations} that indeed
\[
	\De_{g_{b_0}} \left( \Ric - \Lambda \right) (g')
		= 0.
\]
The proof is complete.
\end{proof}

\subsection{Non-linear stability}

For the Nash-Moser iteration, we want to extend Theorem~\ref{thm: linearized surjectivity} to construct a right inverse $\Rop_{(\tilde g, b, \theta)}$ to $\De_{(\tilde g, b, \theta)}\P$ for any
\[
	(\tilde g, b, \theta)
		\approx (0, b_0, 0).
\]
We also want to show that $\De_{(\tilde g, b, \theta)} \P$ is satisfying sufficient \emph{tame estimates} for the Nash--Moser iteration:

\begin{thm}[The right inverse of the linearization at nearby points] \label{thm: right inverse}
Let $(M, g_{b_0})$ be a subextremal Kerr--de~Sitter spacetime extended beyond the horizons as in~\eqref{eq: manifold} for which Assumption~\ref{ass: mode stability} holds.
There exist $s_0 \in \R$ and $\a > 0$ such that if $s \geq s_0$, then there is a continuous map
\begin{align*}
	\Rop_{(\tilde g, b, \theta)} \colon D^{s, \a}(\Omega; S^2 T^*\Omega)
		&\to \bar H^{s, \a}(\Omega; S^2T^*\Omega) \times T_b B \times \Theta_{b_0},
\end{align*}
depending continuously on $(\tilde g, b, \theta)$ and such that
\[
	\De_{(\tilde g, b, \theta)}\P \circ \Rop_{(\tilde g, b, \theta)}(f, u_0, u_1)
			= (f, u_0, u_1)
\]
for all 
\[
	(f, u_0, u_1) \in D^{\infty, \a}(\Omega; S^2 T^*\Omega).
\]
Moreover, there exist $s_0, d > 0$ such that $(\tilde g', b', \theta') := \Rop_{(\tilde g, b, \theta)}(f, u_0, u_1)$ satisfies the tame estimates
\begin{equation} \label{eq: right inverse tame estimates}
\begin{split}
	\norm{b'}_{T_{b_0}B}
		&\leq C \norm{(f, u_0, u_1)}_{D^{d, \a}}, \\
	\norm{\theta'}_{\Theta_{b_0}}
		&\leq C \norm{(f, u_0, u_1)}_{D^{d, \a}}, \\
	\norm{\tilde g'}_{\bar H^{s, \a}}
		&\leq C_s \left( \norm{(f, u_0, u_1)}_{D^{s+d, \a}} + \left( 1 + \norm{\tilde g}_{\bar H^{s+d, \a}} \right) \norm{(f, u_0, u_1)}_{D^{d, \a}}\right) 
\end{split}
\end{equation}
for $s \geq s_0$.
If $d$ is large enough, there is a small enough $\e > 0$ such that the map $\Rop_{(\tilde g, b, \theta)}$ is defined for any $\tilde g \in \bar H^{d, \a}(S^2 T^*M)$ with
\[
	\norm{\tilde g}_{\bar H^{d, \a}}
		< \e
\]
and $(f, u_0, u_1)$ for which the norms on the right-hand sides of~\eqref{eq: right inverse tame estimates} are finite, and produces a solution $(\tilde g', b', \theta')$ to
\[
	\De_{(\tilde g, b, \theta)}\P(\tilde g', b', \theta')
		= (f, u_0, u_1),
\]
satisfying~\eqref{eq: right inverse tame estimates}.
\end{thm}

These estimates together with the Nash--Moser theorem, in the version of Saint Raymond in \cite{SR1989}, will prove Theorem~\ref{thm: main}.
Following the ideas of \cite{HV2018}*{Sec.~11}, we want to prove Theorem~\ref{thm: right inverse} by a perturbation argument, using Theorem~\ref{thm: linearized surjectivity} as the key input.
Concretely, we want to prove and apply a generalization of \cite{HV2018}*{Thm.~5.14} to the full subextremal range of Kerr--de~Sitter spacetimes.
For a fixed $N_W \in \N_0$, $\e > 0$ and $w_0 \in \R^{N_W}$, define
\begin{equation} \label{eq: W def}
	W
		:= \{ w \in \R^{N_W} \colon \abs{w - w_0} < \e \}.
\end{equation}
In our application, $W \subset B$, $w_0 = b_0$ and $N_W = 4$.
We assume that $L_w$ is a continuous family of linear wave operators satisfying the conditions of \cite{HV2018}*{Sec.~5.1.2} for all $w \in W$.
In our application, 
\[
	L_w
		= \frac12 \Lop_b,
\]
for any $b \in W \subset B$, where $\Lop_b$ is the linearized gauge fixed Einstein operator defined in~\eqref{eq: gauge fixed lin}.
\begin{remark} \label{rmk: L w 0 assumptions}
The assumptions on the family of linear wave operators $L_w$ in \cite{HV2018}*{Sec.~5.1.2} concern two aspects:
\begin{enumerate}
	\item The global structure of the bicharacteristic flow in the spacetime.
	\item Conditions on the subprincipal symbol at the critical points of the bicharacteristic flow, i.e., at the radial points and at the trapped set.
\end{enumerate}
For linearized wave operators, like the linearized gauge fixed Einstein operator, the structural assumptions for the bicharacteristic flow were proven to hold for all subextremal Kerr--de~Sitter spacetimes in \cite{PV2023}, in the case $r_0 \in (r_e, r_c)$ is the unique point such that $\mu'(r_0) = 0$, and later in \cite{PV2024} for any $r_0 \in [r_e, r_c]$.
The subprincipal symbol at the radial point essentially decides the regularity threshold $s_0$ in Theorem~\ref{thm: main}.
In principal, it should be possible to compute an explicit bound for $s_0$, but the existence of such an $s_0$ is sufficient for our purposes here.
Finally, the subprincipal symbol condition at the trapped set is more delicate (see \cite{HV2018}*{Eq.~(5.5)}) and is precisely what we prove in Theorem~\ref{thm: spectral gap} below for $\Lop_b$, for any $b \in B$.
Our linear wave operators $L_w = \frac12 \Lop_b$, for any $b \in B$, thus satisfies all assumptions in \cite{HV2018}*{Sec.~5.1.2}.
\end{remark}

For the Nash--Moser iteration, we also need to consider wave operators associated with spacetime perturbations outside the subextremal Kerr--de~Sitter family.
We therefore need to introduce a second parameter accounting for the addition of a non-smooth, but decaying, perturbation $\tilde g$.
For any $s, \a \in \R$ and a vector bundle $E \to \Omega$, define
\begin{equation} \label{eq: tilde W def}
	\widetilde W^s
		:= \left\{ \tilde w \in \bar H^{s, \a}(\Omega; E) \colon \norm{\tilde w}_{\bar H^{s_0, \a}} < \e \right\}.
\end{equation}
In our application, $\tilde w = \tilde g$.
For each $w \in W$ and $\tilde w \in \widetilde W^s$, we want a corresponding perturbation
\[
	\tilde L_{w, \tilde w}
		\in \bar H^{s, \a}(\Omega) \mathrm{Diff}_{\mathrm b}^2(\Omega),
\]
where we refer to \cite{HV2018}*{Sec.~2} for the notation and the precise assumptions, depending on $w$ on $\tilde w$ in a tame fashion, that is,
\begin{equation} \label{eq: w tilde w bounds}
	\norm{\tilde L_{w_1, \tilde w_1} - \tilde L_{w_2, \tilde w_2}}_{\bar H^{s, \a}(\Omega) \mathrm{Diff}_{\mathrm b}^2(\Omega)}
		\leq C_s \left( \abs{w_1 - w_2} + \norm{\tilde w_1 - \tilde w_2}_{\bar H^{s, \a}} \right) 
\end{equation}
for all $s \geq s_0$.
The linear wave operator for which the perturbation theorem is formulated, Theorem~\ref{thm: lambda} below, is
\[
	L_{w, \tilde w}
		:= L_w + \tilde L_{w, \tilde w},
\]
consisting of the part $L_w$ depending on the parameter $w$ and a perturbation part $\tilde L_{w, \tilde w}$, which depends on both the parameter $w$ and the perturbation $\tilde w$.

With $w = b$ and $\tilde w = \tilde g$ in our application,
\[
	L_{w, \tilde w}
		:= \De_{(g_{b_0, b} + \tilde g)} \left( \Ric - \Lambda \right) + \left(\de_{g_{b_0}}^* + \A_{g_{b_0}} \right) \De_{(g_{b_0, b} + \tilde g)} \Upsilon_{g_{b_0}}.
\]
This wave operator is indeed close to 
\[
	L_w = \De_{g_{b_0, b}} \left( \Ric - \Lambda \right) + \left(\de_{g_{b_0}}^* + \A_{g_{b_0}} \right) \De_{g_{b_0, b}} \Upsilon_{g_{b_0}}
\]
for $t_* > t_0 + 2$, if $\tilde g \approx 0$, and satisfies the bound~\eqref{eq: w tilde w bounds}.

Now, in order to state Theorem~\ref{thm: lambda}, which is the abstract perturbation result needed for Theorem~\ref{thm: right inverse}, we need to introduce some further notation.
Given any linear wave operator $L$, consider the initial value problem
\[
	(L, \gamma_0) v
		= (f, u_0, u_1).
\]
We want to analyze the future part of the solution $u := H v$, where $H = H(t_* - t_0)$ is the Heaviside function.
Let first $w$ being the unique solution to
\[
	(L, \gamma_0) w
		= (0, u_0, u_1).
\]
Note that $u_I := H w$ satisfies
\[
	L u_I
		= L H w
		= [L, H] w,
\]
where $[L, H]$ is a linear differential operator with coefficients given by differentiating $\delta$-distributions at most once.
In particular, $[L, H] w$ is a distribution supported only at $\Sigma_0$.
We will therefore use the notation
\[
	[L, H](u_0, u_1)
		:= [L, H]w,
\]
and note that $u$ is the unique forward solution to
\[
	Lu
		= H f + [L, H](u_0, u_1).
\]
Let us now define
\begin{align*}
	R^*
		:= \bigoplus_{\Im(\s) \geq 0} \mathrm{Res}^* \big\{ \mathrm{res}_{\zeta = \bar \s} \big( e^{-it_* \zeta} \widehat{L^*}(\zeta)^{-1} p(\zeta): 
		& \ p(\zeta) \text{ is a polynomial in } \zeta \\
		& \ \text{ with values in } \dot {\D}'(\Sigma_0; E|_{\Sigma_0}) \big) \big\},
\end{align*}
c.f.~\cite{HV2018}*{Eq.~(5.20)}.
This corresponds to all the non-decaying resonance states of $L$.
We proceed by defining the map
\begin{align*}
	\lambda_{\mathrm{IVP}}: D^{s-1, \a}(\Omega; E)
		&\to \L(R^*, \bar \C), \\
	(f, u_0, u_1)
		&\mapsto \ldr{Hf + [L, H](u_0, u_1), \cdot},
\end{align*}
where $s \in \R$.

\begin{thm} \label{thm: lambda}
For $\a, s_0 \in \R$, assume for all $s \geq s_0$ that there is a continuous map
\begin{align*}
	z: W \times \widetilde W^s \times \C^{N_{\mathcal Z}}
		&\to D^{s, \a}(\Omega; E), \\
	(w, \widetilde w, c)
		&\mapsto z^c_{w, \widetilde w} := z(w, \widetilde w, c),
\end{align*}
which is linear in the last argument.
With
\[
	\Xi
		:= \mathrm{Res}\left(L_{w_0, 0}\right) \cap \{ \Im \ \s > - \a \},
\]
suppose moreover that the map
\begin{equation} \label{eq: c map}
\begin{split}
	\C^{N_{\mathcal Z}}
		&\to \L\left( \mathrm{Res}^*\left( L_{w_0, 0}, \Xi \right), \bar \C \right), \\
	c
		&\mapsto \lambda_{\mathrm{IVP}}\left( z_{w_0, 0}^c \right)
\end{split}
\end{equation}
is surjective.
Then, if $\e > 0$ in~\eqref{eq: W def} and~\eqref{eq: tilde W def} is small enough, then there exists a continuous map
\begin{align*}
	S: W \times \widetilde W^s \times D^{\infty, \a}(\Omega; E)
		&\to \C^{N_{\mathcal Z}} \oplus \bar H^{\infty, \a}(\Omega; E), \\
	(w, \tilde w, (f, u_0, u_1))
		&\mapsto (c, u),
\end{align*}
linear in $(f, u_0, u_1)$, such that the function $u$ is a solution to of
\begin{equation} \label{eq: perturbed ivp 1}
	(L_{w, \tilde w}, \gamma_0)(u)
		= (f, u_0, u_1) + z^c_{w, \tilde w}.
\end{equation}
Moreover, there are $s_0, d > 0$ such that $S$ satisfies the tame estimates
\begin{equation} \label{eq: tame estimates}
\begin{split}
	\norm{c}_{\C^{N_{\mathcal Z}}}
		&\leq C \norm{(f, u_0, u_1)}_{D^{d, \a}}, \\
	\norm{u}_{\bar H^{s, \a}}
		&\leq C_s \left( \norm{(f, u_0, u_1)}_{D^{s+d, \a}} + \left( 1 + \norm{\tilde w}_{\bar H^{s+d, \a}} \right) \norm{(f, u_0, u_1)}_{D^{d, \a}}\right) 
\end{split}
\end{equation}
for $s \geq s_0$.
In fact, if $d$ is large enough, the map $S$ is defined for any $\tilde w \in \widetilde W^d$ and $(f, u_0, u_1)$ for which the norms on the right-hand sides of~\eqref{eq: tame estimates} are finite, and produces a solution $(c, u)$ to~\eqref{eq: perturbed ivp 1} satisfying~\eqref{eq: tame estimates}.
Moreover, if the map~\eqref{eq: c map} is bijective, then the map $S$ with these properties is unique.
\end{thm}

\begin{proof}[Proof of Theorem~\ref{thm: lambda}]
The proof of \cite{HV2018}*{Thm.~5.14}, which is the same statement as Theorem~\eqref{thm: lambda} for slowly rotating Kerr--de~Sitter spacetimes, now goes through in the full subextremal range using the observations in \cites{PV2023, PV2024}.
Indeed, the proof consists of two new ingredients compared to the proof of Theorem~\ref{thm: expansion linearized} or Corollary~\ref{cor: expansion linearized - initial data}, namely:
\begin{enumerate}
	\item The perturbation theory developed in \cite{HV2018}*{Sec.~5}.
	\item The tame estimates proven in \cite{HV2017_IMRN}*{Sec.~5}.
\end{enumerate}
The perturbation theory relies on showing a perturbation stability property in the construction of the resonances. 
The construction of the resonances is completely analogous as in the slowly rotating case as in the full subextremal range, provided that the global dynamics of the bicharacteristics is analogous. 
Since this was shown to be the case in all subextremal Kerr--de~Sitter spacetimes in \cites{PV2023, PV2024}, c.f.~Remark~\ref{rmk: L w 0 assumptions}, the methods of \cite{HV2018}*{Sec.~5} carry through analogously.
The tame estimates in \cite{HV2017_IMRN}*{Sec.~5} rely again on the global structure of the bicharacteristic flow, in particular at the critical points, i.e., the radial point structure and the trapped set.
Again, these conditions were shown to hold in all subextremal Kerr--de~Sitter spacetimes in \cites{PV2023, PV2024}, implying tame estimates in the full subextremal range. 
The only concrete difference to the slowly rotating case is that we do not compute any explicit bounds for the constants $s_0$ and $d$ in the full subextremal range, as was done in the slowly rotating case (where the thresholds could be computed for the Schwarzschild--de~Sitter spacetime).
\end{proof}

We use this to prove Theorem~\ref{thm: right inverse}.

\begin{proof}[Proof of Theorem~\ref{thm: right inverse}]
Let $\a > 0$ be so small that all resonances $\s$ for $\Lop_{b_0}$ satisfy $\Im(\s) \geq 0$.
In our application of Theorem~\ref{thm: right inverse}, we choose
\begin{align*}
	W
		& \subset B, \\
	\widetilde W^s
		& \subset \bar H^{s, \a}(S^2 T^*\Omega),
\end{align*}
and
\[
	N_{\mathcal Z}
		:= \sum_{j = 1}^N l_j,
\]
where $N$ and $l_1, \hdots, l_N$ are as in~\eqref{eq: solution asymptotics}, and the maps
\begin{align*}
	L_{b, \tilde g}(\tilde g')
		:= & \ \De_{(g_{b_0, b} + \tilde g)} \left( \Ric - \Lambda \right) ( \tilde g' ) + 2 \left(\de_{g_{b_0}}^* + \A_{g_{b_0}} \right) \De_{(g_{b_0, b} + \tilde g)} \Upsilon_{g_{b_0}}( \tilde g' ) \\
	z^c_{b, \tilde g}
		:= & \ \Big( - \De_{(g_{b_0, b} + \tilde g)} \left( \Ric - \Lambda \right) ( \chi g_{b_0}'(b'(c))) \\*
		& \ - \left(\de_{g_{b_0}}^* + \A_{g_{b_0}} \right) \De_{(g_{b_0, b} + \tilde g)} \Upsilon_{g_{b_0}}( \chi g_{b_0}'(b'(c)) ) \\
		& \ + \left(\de_{g_{b_0}}^* + \A_{g_{b_0}} \right) \left( \De_{g_{b_0, b}}\Upsilon_{g_{b_0}} (\chi g_{b_0}'(b'(c))) + \theta'(c) \right), 0, 0 \Big)
\end{align*}
where $b'$ and $\theta'$ are defined by the formulas in~\eqref{eq: R at KdS}, with
\[
	c
		= (c_1, \hdots, c_{l_1}, c_{2, 1}, \hdots, c_{N, l_N}).
\]
The continuity is clear and the main point is that 
\begin{equation} \label{eq: linearization trick}
	\De_{(\tilde g, b, \theta)} \P(\tilde g', b'(c), \theta'(c))
		= \left( \left( L_{b, \tilde g}(\tilde g'), \gamma_0(\tilde g) \right) - z^c_{b, \tilde g} \right),
\end{equation}
c.f.~Equations~\eqref{eq: D P1} and~\eqref{eq: D P2}.
Furthermore, we note that
\begin{align*}
	L_{b_0, 0}
		&= \frac 12 \Lop_{b_0}, \\
	z^{c}_{b_0, 0}
		&= - \left( \De_{g_{b_0}} \Ric - \Lambda \right) \chi g_{b_0}'(b'(c)) - \left(\de_{g_{b_0}}^* + \A_{g_{b_0}} \right) \theta'(c)
\end{align*}
Theorem~\ref{thm: linearized surjectivity} implies that  the map $\lambda_{\mathrm{IVP}}$ is surjective, c.f.~\cite{HV2018}*{Prop.~5.7}.
Theorem~\ref{thm: lambda} now implies that, for a small enough $\e > 0$ in~\eqref{eq: W def} and~\eqref{eq: tilde W def}, there exists a continuous map
\begin{align*}
	S: W \times \widetilde W^s \times D^{\infty, \a}(\Omega)
		&\to \C^{N_{\mathcal Z}} \oplus \bar H_{\mathrm b}^{\infty, \a}(\Omega), \\
	(b, \tilde g, (f, u_0, u_1))
		&\mapsto (c, \tilde g'),
\end{align*}
linear in $(f, u_0, u_1)$, such that $\tilde g'$ is a solution to of
\[
	(L_{b, \tilde g}, \gamma_0)(\tilde g')
		= (f, u_0, u_1) + z^c_{b, \tilde g}.
\]
Inserting this in~\eqref{eq: linearization trick} shows that
\[
	\De_{(\tilde g, b, \theta)} \P(\tilde g', b'(c), \theta'(c))
		= (f, u_0, u_1),
\]
and we have thus found our right inverse
\[
	R_{(\tilde g, b, \theta)}(f, u_0, u_1)
		:= (\tilde g', b'(c), \theta'(c)),
\]
where $c$ is given the map $S$.
The tame estimates for $R_{(\tilde g, b, \theta)}$ follow by~\eqref{eq: tame estimates}.
\end{proof}

We may now finally prove the main result of this paper, remembering that it still relies on Theorem~\ref{thm: spectral gap} and Theorem~\ref{thm: stable constraint propagation}, both proven below:

\begin{proof}[Proof of Theorem~\ref{thm: main}, assuming Theorem~\ref{thm: spectral gap} and Theorem~\ref{thm: stable constraint propagation}]
We apply a variant of the Nash--Moser implicit function theorem, due to Saint Raymond in \cite{SR1989}.
In the notation of \cite{HV2018}*{Thm.~11.1}, we choose
\begin{align*}
	B^s
		&= \bar H^{s, \a}(\Omega; S^2T^*\Omega) \times T_b B \times \Theta_{b_0}, \\
	{\bf B}^s 
		&= D^{s, \a}(\Omega; S^2T^*\Omega), \\
	u
		&= (\tilde g, b, \theta), \\
	\phi(\tilde g, b, \theta)
		&= \P(\tilde g, b, \theta), \\
	\psi(\tilde g, b, \theta)
		&= \Rop_{(\tilde g, b, \theta)}.
\end{align*}
Then \cite{HV2018}*{Thm.~11.1}, i.e., the Saint Raymond version of the Nash--Moser implicit function, implies that if
\[
	\norm{\P(0, b_0, 0)}_{D^{s_0, \a}(\Omega)}
\]
is small enough, which is the same as saying that
\[
	\norm{\gamma_0(\tilde g) - \iota_{b_0}(h, k)}_{\bar H^{s_0 + 1} \left(\Sigma_0 \right) \oplus  \bar H^{s_0} \left(\Sigma_0\right)}
\]
is small enough (since $\P_1(0, b_0, 0) = 0$), there is a unique $(\tilde g, b, \theta)$ such that
\[
	\P(\tilde g, b, \theta)
		= 0, 
\]
which means that 
\begin{align*}
	0
		&= \left( \Ric - \Lambda \right) (g_{b_0, b} + \tilde g) + \left(\de_{g_{b_0}}^* + \A_{g_{b_0}} \right) \left( \Upsilon_{g_{b_0}}(g_{b_0, b} + \tilde g) - \Upsilon_{g_{b_0}}(g_{b_0, b}) - \theta \right), \\
	0
		&= \gamma_0(\tilde g) - \iota_{b_0}(h, k).
\end{align*}
The argument outlined in the beginning of this section now implies that in fact $g := g_{b_0, b} + \tilde g$ satisfies
\[
	\Ric_g - \Lambda g
		= 0
\]
and $(h,k)$ are the first and second fundamental forms induced by $g$.
Since 
\[
	\tilde g 
		\in \bar H^{\infty, \a}(\Omega, S^2 T^*\Omega),
\]
it follows that
\[
	g - g_b
		\in \O\left(e^{-\a t_*}\right),
\]
as $t_* \to \infty$.
This finishes the proof.
\end{proof}

\section{Subprincipal symbol at the trapped set (ST)}
\label{sec: spectral gap}

We begin by recalling that 
\[
	\c^{\sharp}
		:= \d_{t_*} + \frac a{r^2 + a^2} \d_{\phi_*} + \g f(r) \d_r + \g^2 \sin(\theta) \d_\theta,
\]
for some $\kappa > 0$. 
Let $\pi: T^*M \to M$ denote the canonical projection.
Let 
\[
	\Sigma^\pm
		:= \{\xi \in \Char(\Box_g) \mid \pm G(\md t_*, \xi ) > 0\}
\]
be the two components of the characteristic set of $\Box_g$, and let $\Gamma^\pm \subset \Sigma^\pm$ denote the two components of the trapped set.
In this section, we want to prove the following
\begin{thm}[ST] \label{thm: spectral gap}
Let $(M, g_{b_0})$ be a subextremal Kerr--de~Sitter spacetime extended beyond the horizons as in~\eqref{eq: manifold}.
Let $\varepsilon > 0$.
There are constants $\g_0, e_0, h_0 > 0$ such that if 
\[
	\g \in (0, \g_0), \quad e \in (0, e_0), \quad h \in (0, h_0),
\]
then there exists a stationary, smooth, positive definite fiber inner product $\mathfrak C$ on $\pi^*S^2 T^*M$ near $\Gamma^\pm$ such that
\[
	\pm \frac1{2i \abs{\s}} \left( S_{\mathrm{sub}}\left( \Lop_g \right) - S_{\mathrm{sub}}\left( \Lop_g \right)^{*\mathfrak C} \right) |_{\Gamma^\pm}
		\leq \e,
\]
where $S_{\mathrm{sub}}\left( \Lop_g \right)$ is the subprincipal part of the linearized gauge fixed Einstein operator, introduced in~\eqref{eq: lin gauged equations}.
\end{thm}

\subsection{The Marck tetrad}

In the proof of Theorem~\ref{thm: spectral gap}, we will trivialize the bundle $\pi^* T^*M$ over an open neighbourhood of the trapped set $\Gamma$.
For this, we follow a construction due to Marck in \cite{M1983}.
The Kerr--de~Sitter metric is Boyer--Lindquist coordinates is given by
\begin{equation}
\begin{split}
	g
		&= - \frac{\mu(r)}{b^2\varrho^2}\left(\md t - a \sin^2(\theta) \md \phi \right)^2 + \frac{\varrho^2}{\mu(r)} \md r^2 \\*
		&\qquad + \frac{c(\theta)\sin^2(\theta)}{b^2\varrho^2}\left(a\md t - \left(r^2 + a^2\right) \md \phi \right)^2 + \frac{\varrho^2}{c(\theta)}\md \theta^2,
\end{split}
\end{equation}
where $\varrho^2 = r^2 + a^2 \cos^2 (\theta)$.
In these coordinates, the metric is obviously diagonal with respect to the \emph{Carter tetrad}
\begin{align*}
	\o_0
		:= & \ \frac{\sqrt{\mu(r)}}{b\varrho} \left(\md t - a \sin^2(\theta) \md \phi \right), & 
	\o_1
		:= & \ \frac{\varrho}{\sqrt{\mu(r)}} \md r, \\*
	\o_2
		:= & \ \frac{\varrho}{\sqrt{c(\theta)}} \md \theta, 
		&
	\o_3
		:= & \ \frac{\sqrt{c(\theta)}}{b \varrho} \sin(\theta) \left(a\md t - \left(r^2 + a^2\right) \md \phi\right),
\end{align*}
so
\[
	g
		= - \omega_0^2 + \omega_1^2 + \omega_2^2 + \omega_3^2,
\]
with corresponding dual frame
\begin{align*}
	e_0
		:= & \ - \frac b {\varrho \sqrt{\mu(r)}}\left((r^2 + a^2)\d_t + a \d_\phi \right),
	& 
	e_1
		:= & \ \frac{\sqrt{\mu(r)}}{\varrho} \d_r, \\
	e_2
		:= & \ \frac{\sqrt{c(\theta)}}{\varrho} \d_\theta, 
	& 
	e_3
		:= & \ \frac b{\varrho \sqrt{c(\theta)}} \frac1{\sin(\theta)} \left(a \sin^2(\theta) \d_t + \d_\phi \right).
\end{align*}
Consequently, the dual metric $G$ to $g$ is given by
\begin{equation} \label{eq: dual metric}
\begin{split}
	\varrho^2 G(\xi, \xi)
		&= \varrho^2 \left( - \xi(e_0)^2 + \xi(e_1)^2 + \xi(e_2)^2 + \xi(e_3)^2 \right) \\
		&= - \frac{b^2}{\mu(r)} \left( (r^2 + a^2)\xi_t + a \xi_\phi \right)^2 + \mu(r) \xi_r^2 + K(\xi),
\end{split}
\end{equation}
for any $\xi \in T^*M$, where $K(\xi)$ is the \emph{Carter constant} given by
\[
	K(\xi)
		:= c(\theta) \xi_\theta^2 + \frac{b^2}{c(\theta) \sin^2(\theta)}\left(a \sin^2(\theta)\xi_t + \xi_\phi \right)^2.
\]
Since $\d_t$ and $\d_\phi$ are Killing vector fields, the dual quantities $\xi_t$ and $\xi_\phi$ are constant along bicharacteristics which readily implies that the Carter constant indeed is \emph{constant} along all bicharacteristics of $\varrho^2 G$ (i.e., integral curves of $\H_{\varrho^2 G}$ in the characteristic set).

\begin{remark} \label{rmk: Carter constant}
Note that
\begin{align*}
	K(\xi)
		& = \varrho^2 \left( \xi(e_2)^2 + \xi(e_3)^2 \right) \\
		& = \varrho^2 \left( \o_2(\xi^\sharp)^2 + \o_3(\xi^\sharp)^2 \right). 
\end{align*}
for any $\xi \in T^*M$.
If, in addition, $\xi$ is lightlike, then also
\begin{align*}
	K(\xi)
		& = \varrho^2 \left( \xi(e_0)^2 - \xi(e_1)^2 \right) \\
		& = \varrho^2 \left( \o_0(\xi^\sharp)^2 - \o_1(\xi^\sharp)^2 \right),
\end{align*}
by~\eqref{eq: dual metric}.
\end{remark}

Let us denote
\[
	\Sigma := \Sigma^+ \cup \Sigma^-, \qquad \Gamma := \Gamma^+ \cup \Gamma^- \subset \Sigma,
\]
and let us define the open subset
\[
	\U
		:= \{\xi \in \Sigma \mid K(\xi) > 0\}.
\]

\begin{remark} \label{rmk: positive Carter constant}
By \cite{PV2023}*{Thm.~3.2}, we have $\xi_r = 0$ for all $\xi \in \Gamma$. 
Since also $G(\xi, \xi) = 0$, Equation~\eqref{eq: dual metric} implies that
\[
	 \Gamma \subset \U.
\]
\end{remark}

Marck's choice of co-frame in \cite{M1983} is based on the following algebraic feature of the Kerr--de~Sitter spacetime:

\begin{lemma} \label{le: Killing-Yano}
Define the tensor
\[
	\F
		:= r \o_2 \wedge \o_3 + a \cos(\theta) \o_0 \wedge \o_1.
\]
\begin{itemize}
	\item $\F$ extends smoothly to all of $M$.
	\item $\F$ is a Killing-Yano tensor, i.e., satisfying
	\[
		\n_X \F (Y, Z)
			= \n_Z \F (Y, X),
	\]
	for all vectors $X, Y, Y$.
	\item Let $*$ denote the Hodge-star operator. Then
	\begin{equation} \label{eq: F star}
		*\F
			= - \frac1{\varrho^2} \left( a \cos(\theta) \o_2 \wedge \o_3 + r \o_0 \wedge \o_1 \right) 
	\end{equation}
	also extends smoothly to $M$.
\end{itemize}
\end{lemma}
\begin{proof}
Firstly, the smooth extension to the north and south poles $\theta = 0, \pi$, where the Boyer-Lindquist coordinates are not defined, follows by noting that $\o_0$ and $\o_1$ and
\[
	\o_2 \wedge \o_3
		= - \frac a b \left( \md \cos(\theta) \right) \wedge \md t - \frac{r^2 + a^2}b \md t \wedge \sin(\theta) \md \phi
\]
are smoothly extendible.
Secondly, the Killing-Yano property is a computation (which we verified using Mathematica).
Thirdly, the formula~\eqref{eq: F star} follows readily from the defining property of the Hodge-star operator, i.e., that
\[
	\b \wedge *\F
		= G(\b, \F) \o_0 \wedge \o_1 \wedge \o_2 \wedge \o_3,
\]
for any two-form $\b$.
$*\F$ extends smoothly to $M$ since $\F$ does.
\end{proof}

\begin{remark} \label{rmk: H}
Again since $\o_0$ and $\o_1$ extend smoothly to $M$, it follows that the symmetric $(0, 2)$-tensor given by
\[
	\mathcal H
		:= g + 2 \left( \o_0^2 - \o_1^2 \right)
		= \o_0^2 - \o_1^2 + \o_2^2 + \o_3^2
\]
extends smoothly to $M$.
\end{remark}

Generalizing \cite{H2024}*{Sec.~3.2.3} to the Kerr--de~Sitter spacetime, we define the following sections in the pull-back bundle $\pi^*T^*M|_{\U}$, where $\pi: T^*M \to M$ is the projection, by
\begin{align*}
	\a_1
		:= & \ \frac1{\sqrt{K(\xi)}}\xi
	& 
	\a_2
		:= & \ \frac{\varrho^2}{\sqrt{K(\xi)}} \left( *\F\right) (\xi^\sharp, \cdot), \\
	\a_3
		:= & \ \frac{- \varrho^2}{2 \sqrt{K(\xi)}} \mathcal H(\xi^\sharp, \cdot), 
	&
	\a_4
		:= & \ \frac 2{\sqrt{K(\xi)}}\F(\xi^\sharp, \cdot),
\end{align*}
for any $\xi \in \U$, where $\F$ and $\mathcal H$ where introduced in Lemma~\ref{le: Killing-Yano} and Remark~\ref{rmk: H}, respectively.
By the arguments in Lemma~\ref{le: Killing-Yano} and Remark~\ref{rmk: H}, $\a_1, \hdots, \a_4$ are indeed smooth everywhere in $\U$ (also where the Boyer-Lindquist coordinates are not defined).

\begin{remark}
Since $\sqrt{K(\xi)}$ is homogeneous of degree 1, note that $\a_1, \hdots, \a_4$ are invariant under dilations $\pi^* T^*M|_{\U}$.
\end{remark}

\begin{lemma} \label{le: quasiorthonormal}
We have
\[
	\left( G(\a_i, \a_j) \right)_{i,j}
		= \begin{pmatrix}
		0 & 0 & -1 & 0 \\
		0 & 1 & 0 & 0 \\
		-1 & 0 & 0 & 0 \\
		0 & 0 & 0 & 1
		\end{pmatrix}
\]
in $\U$.
In other words, the metric is given by
\[
	g
		= -2 \a_1 \otimes_s \a_3 + \a_2^2 + \a_4^2.
\]
In particular, $\a_1, \a_2, \a_3, \a_4$ are linearly independent.
\end{lemma}
\begin{proof}
Since $\xi$ is lightlike, 
\[
	G(\a_1, \a_1) 
		= \frac1{K(\xi)} G(\xi, \xi) = 0,
\]
and similarly 
\[
	G(\a_3, \a_3) 
		= \frac{\varrho^2}{4 K(\xi)} G(\xi, \xi) = 0.
\]
Next, using Remark~\ref{rmk: Carter constant}, note that
\begin{align*}
	G(\a_4, \a_4)
		&= \frac1{K(\xi)} \left( a^2 \cos^2(\theta) \left( - \o_1(\xi^\sharp)^2 + \o_0 (\xi^\sharp)^2 \right) + r^2 \left( \o_3(\xi^\sharp)^2 + \o_2(\xi^\sharp)^2 \right) \right) \\*
		&= 1, \\
	G(\a_2, \a_2)
		&= \frac1{K(\xi)} \left( r^2 \left( - \o_1(\xi^\sharp)^2 + \o_0 (\xi^\sharp)^2 \right) + a^2 \cos^2(\theta) \left( \o_3(\xi^\sharp)^2 + \o_2(\xi^\sharp)^2 \right) \right) \\*
		&= 1, \\
	G(\a_1, \a_3)
		& = \frac{- \varrho^2}{2 K(\xi)} \mathcal H(\xi^\sharp, \xi^\sharp) \\
		&= \frac{- \varrho^2}{2 K(\xi)} \left( g(\xi^\sharp, \xi^\sharp) + 2 \left( \o_0(\xi^\sharp)^2 - \o_1(\xi^\sharp)^2 \right) \right) \\
		& = - 1.
\end{align*}
Moreover,
\begin{align*}
	G(\a_2, \a_4)
		& = \frac1{K(\xi)} \left( - r a \cos(\theta) \o_1(\xi^\sharp)^2 + r a \cos(\theta) \o_0 (\xi^\sharp)^2 \right. \\
		&\qquad \left. - r a \cos(\theta) \o_3(\xi^\sharp)^2 - r a \cos(\theta) \o_2(\xi^\sharp)^2 \right) \\
		& = - \frac{r a \cos(\theta)}{K(\xi)} G(\xi, \xi) \\
		& = 0.
\end{align*}
One readily checks that the remaining entries vanish.
\end{proof}

\subsection{The subprincipal symbol at the trapped set}

By \cite{PV2023}*{Thm.~3.2}, the trapped set $\Gamma \subset \Sigma$ is contained in the domain of outer communications, i.e., in the region where $r \in (r_e, r_c)$.
We may therefore describe it using the Boyer-Lindquist coordinates $(t, r, \phi, \theta)$ and dual coordinates $(\xi_t, \xi_r, \xi_\phi, \xi_\theta)$.
Concretely, for any $\xi \in \Sigma$ the function
\[
	\frac{\left( (r^2 + a^2)\xi_t + a \xi_\phi \right)^2 }{\mu(r)},
\]
has a unique positive minimum in $(r_e, r_c)$ at, say $r_{\xi_t, \xi_\phi}$, see \cite{PV2023}*{Thm.~3.2}.
The trapped set in $M$ is given by
\[
	\Gamma
		:= \bigcup_{(\xi_t, \xi_\phi) \neq (0, 0)} \{ \xi_r = 0, r = r_{\xi_t, \xi_\phi} \} \cap \Sigma.
\]
We will discuss the skew-adjoint part of the linearized gauge fixed Einstein operator defined in~\eqref{eq: gauge fixed lin} at the trapped set.

\begin{lemma} \label{le: cov Hamilton derivative}
The pull-back of the Levi-Civita connection acting on $\pi^*T^*M$ is given by
\begin{equation} \label{eq: Hamilton derivative matrix}
	\pi^* \n_{\H_G}
		= \H_G +
		\mathfrak a
		\begin{pmatrix}
			0 & 1 & 0 & 0 \\
			0 & 0 & 1 & 0 \\
			0 & 0 & 0 & 0 \\
			0 & 0 & 0 & 0
		\end{pmatrix}
\end{equation}
in $\U$, with respect to the splitting
\begin{equation} \label{eq: pull back split}
	\pi^* T^*M|_{\U}
		= \ldr{\a_1} \oplus \ldr{\a_2} \oplus \ldr{\a_3} \oplus \ldr{\a_4},
\end{equation}
where $\mathfrak a$ is homogeneous of degree $1$ with respect to the dilations in $\pi^* TM |_{\U}$.
\end{lemma}

\begin{proof}[Proof of Lemma~\ref{le: cov Hamilton derivative}]
We need to compute the matrix 
\[
	G\left(\n_{\H_G}\a_i, \a_j\right),
\]
in $\U$, for $i,j = 1, \hdots, 4$.
Using that
\[
	\pi_*\left( \H_G \right)
		= \sum_{j,k = 0}^n 2 G^{jk} \xi_j \d_k,
\]
we first compute, in geodesic normal coordinates around an arbitrary point $p$,
\begin{align*}
	\left( \pi^*\n \right)_{\H_G}(\xi)|_p
		&= \sum_{j = 0}^n \left( \pi^*\n \right) _{\H_G}(\xi_j \pi^* \md x^j)|_p \\
		&= \sum_{j = 0}^n \H_G ( \xi_j ) \pi^*\md x^j + \xi_j \pi^*\left( \n_{\pi_*(\H_G)} \md x^j \right)|_p \\
		&= \sum_{j, k, l = 0}^n - \left( \d_j G^{kl}\right) \xi_k \xi_l \pi^* \md x^j + 2 \xi_j G^{kl} \xi_k \pi^*\left( \n_{\d_l} \md x^j \right)|_p \\
		&= 0.
\end{align*}
Since $p$ was arbitrary and $\left( \pi^*\n \right)_{\H_G}(\xi)$ is coordinate invariant, it does indeed vanish everywhere.
Since $K(\xi)$ is constant along the Hamiltonian flow, we conclude that 
\[
	\left( \pi^*\n \right)_{\H_G}\a_1 = 0,
\]
proving that all entries in the first column in~\eqref{eq: Hamilton derivative matrix} are zero.
In order not to make the notation overly complicated, we will simply denote $\pi^* G$ by $G$, and similarly for the pullback of other tensors when suitable.
Since $\n G = 0$, it follows that $\pi^* \n G = 0$.
Also, the Killing-Yano property is inherited in the pullback, i.e.
\[
	\left( \pi^* \n \right)_X \F (Y, Z)
		= \left( \pi^* \n \right)_Z \F (Y, X),
\]
Using this, and the fact that
\[
	\pi_*\left( \xi^\sharp \right)
		= \frac12 \pi_*\left( \H_G \right),
\]
 we compute
\begin{align*}
	\left( \pi^*\nabla \right)_{\H_G} \left( \F(\xi^\sharp, \cdot) \right) (Y) 
		&= G \left( \left(\pi^*\nabla \right)_{\H_G} \xi, \F(\cdot, Y) \right) + G \left( \xi, \left( \left(\pi^*\nabla \right)_{\H_G} \F \right)(\cdot, Y) \right) \\
		&= \left(\pi^*\nabla \right)_Y \F (\xi^\sharp, \H_G) \\
		&= 2 \n_Y \F \left( \pi_*\left( \xi^\sharp \right), \pi_*\left( \xi^\sharp \right) \right) \\
		&= 0,
\end{align*}
since $\F$ is a two-form.
Again, since $K(\xi)$ is constant along the Hamiltonian flow, we conclude that 
\[
	\left( \pi^*\n \right)_{\H_G}\a_4
		= 0,
\]
proving that all entries in the fourth column in~\eqref{eq: Hamilton derivative matrix} are zero.
Using that $G(\a_j, \a_k)$ is constant for all $j, k$, we note that
\begin{align*}
	0
		&= \H_G G \left( \a_j, \a_j \right)
		= 2 G \left( \pi^* \n_{H_G} \a_j, \a_j \right)
\end{align*}
for $j = 2,3$, and  
\begin{align*}
	G\left(\pi^*\n_{\H_G} \a_j, \a_k \right)
		&= - G\left(\a_j, \pi^*\n_{\H_G}\a_k \right) 
		= 0,
\end{align*}
for all $j$ and $k = 1, 4$, proving that $\pi^*\n_{\H_G}\a_j$ cannot have any $\a_2$ or $\a_4$ components, i.e., proving that all entries in the first and fourth rows are zero.
Finally, we note that
\begin{align*}
	G\left(\pi^*\n_{\H_G} \a_2, \a_1 \right)
		&= - G\left(\a_2, \pi^*\n_{\H_G} \a_1 \right) 
		= 0, \\
	G\left(\pi^*\n_{\H_G} \a_3, \a_1 \right)
		&= - G\left(\a_3, \pi^*\n_{\H_G} \a_1 \right) 
		= 0,
\end{align*}
and that
\[
	\mathfrak a
		:= - G\left(\pi^*\n_{\H_G} \a_2, \a_3 \right)
		= G\left(\a_2, \pi^*\n_{\H_G} \a_3 \right),
\]
which implies that
\begin{align*}
	\pi^*\n_{\H_G} \a_2
		&= \mathfrak a \a_1, \\
	\pi^*\n_{\H_G} \a_3
		&= \mathfrak a \a_2.
\end{align*}
Moreover, $\mathfrak a$ is homogeneous of degree $1$, since $\a_2, \a_3$ are homogeneous of degree $0$ and $\H_G$ is homogeneous of degree $1$.
This finishes the proof.
\end{proof}

Using this, one quickly deduces the analogous statement for symmetric $(0,2)$-tensors.

\begin{lemma} \label{le: cov Hamilton derivative two tensors}
The pull-back of the Levi-Civita connection acting on $\pi^* S^2 T^*M$ is given by
\begin{equation} \label{eq: Hamilton derivative matrix two tensors}
	\pi^* \n_{\H_G}
		= \H_G +
		\mathfrak a
		\begin{pmatrix}
			0 & 2 & 0 & 0 & 0 & 0 & 0 & 0 & 0 & 0 \\
			0 & 0 & 1 & 0 & 0 & 1 & 0 & 0 & 0 & 0 \\
			0 & 0 & 0 & 0 & 0 & 0 & 0 & 0 & 1 & 0 \\
			0 & 0 & 0 & 0 & 1 & 0 & 0 & 0 & 0 & 0 \\
			0 & 0 & 0 & 0 & 0 & 0 & 0 & 1 & 0 & 0 \\
			0 & 0 & 0 & 0 & 0 & 0 & 0 & 0 & 1 & 0 \\
			0 & 0 & 0 & 0 & 0 & 0 & 0 & 0 & 0 & 0 \\
			0 & 0 & 0 & 0 & 0 & 0 & 0 & 0 & 0 & 0 \\
			0 & 0 & 0 & 0 & 0 & 0 & 0 & 0 & 0 & 1 \\
			0 & 0 & 0 & 0 & 0 & 0 & 0 & 0 & 0 & 0 \\
		\end{pmatrix},
\end{equation}
in $\U$, with respect to the splitting
\begin{equation}\label{eq: alpha 2 tensor split new}
\begin{split}
	\pi^* S^2 T^*M|_{\U}
		&= \ldr{\a_1 \a_1} \oplus \ldr{2\a_1\a_2} \oplus \ldr{2\a_1\a_3} \oplus \ldr{2\a_1\a_4} \oplus \ldr{2 \a_2\a_4} \\
		&\qquad \oplus \ldr{\a_2\a_2 - \a_4\a_4} \oplus \ldr{\a_4\a_4} \oplus \ldr{2 \a_3\a_4}\oplus \ldr{2 \a_2\a_3} \oplus \ldr{\a_3\a_3},
\end{split}
\end{equation}
where $\mathfrak a$ is homogeneous of degree $1$ with respect to the dilations in $\pi^* TM |_{\U}$.
\end{lemma}
\begin{remark}
The point of choosing the splitting~\eqref{eq: alpha 2 tensor split new} is that the matrix in~\eqref{eq: Hamilton derivative matrix two tensors} and the matrix in Lemma~\ref{le: upper triangular} both are upper triangular with respect to this splitting. 
\end{remark}
\begin{proof}
Using that
\[
	\pi^* \n_{\H_G} \left( \a_j \a_k \right)
		= \left( \pi^* \n_{\H_G} \a_j \right) \a_k + \a_j \left( \left( \pi^* \n_{\H_G} \a_k \right) \right),
\]
the statement is a straightforward consequence of Lemma~\ref{le: cov Hamilton derivative}.
\end{proof}

In order to analyse the first order part of~\eqref{eq: gauge fixed lin}, given by $2 \A_g \de_g \T_g$, we first recall that
\begin{equation} \label{eq: A g form} 
	\A_g \o
		= - h^{-1} \T_g^{-1} \left( \c \otimes_s \o - e G(\c, \o) g \right),
\end{equation}
where $\c$ is everywhere timelike. 
We are going to first treat the case when $e = 0$ and then conclude the necessary statements for sufficiently small $e > 0$ by a perturbation argument.
Moreover, $\T_g = \T_g^{-1}$ in spacetime dimension $4$, so what we need to compute is
\[
	\s_1(2 \A_g \de_g \T_g)(\xi)|_{e = 0}
		= - 2 h^{-1} \T_g \circ \left( \c \otimes_s  \cdot \right) \circ \sigma_1(\de_g)(\xi) \circ \T_g
\]
in terms of the splitting~\eqref{eq: alpha 2 tensor split new}.
It will be convenient to express $\c$ in terms of the scalars $\c_1, \c_2, \c_3, \c_4$, so
\begin{align*}
	\c
		= & \ \c_1 \a_1 + \c_2 \a_2 + \c_3\a_3 + \c_4 \a_4.
\end{align*}

\begin{lemma} \label{le: upper triangular}
We have
\begin{align*}
	- 2 h^{-1} \T_g 
		& \left( \c \otimes_s \cdot \right) \circ \s_1(\de_g) \T_g \\
		& =
		- 2 i h^{-1} \sqrt{K(\xi)}
		\begin{pmatrix}
			0 & 0 & 0 & 0 & 0 & 0 & \c_1 & 0 & 0 & 0 \\
			0 & 0 & 0 & 0 & 0 & 0 & \frac{\c_2}2 & 0 & 0 & 0 \\
			0 & 0 & 0 & 0 & 0 & 0 & \c_4 & \c_2 & 0 & 0 \\
			0 & 0 & 0 & 0 & 0 & 0 & \frac{\c_4}2 & \c_1 & 0 & 0 \\			
			0 & 0 & 0 & 0 & 0 & 0 & 0 & \c_2 & \c_4 & 0 \\
			0 & 0 & 0 & 0 & 0 & 0 & \frac{\c_3}2 & - \c_4 & \c_2 & \c_1 \\
			0 & 0 & 0 & 0 & 0 & 0 & \c_3 & 0 & 0 & 2\c_1 \\
			0 & 0 & 0 & 0 & 0 & 0 & 0 & \c_3 & 0 & \c_4 \\
			0 & 0 & 0 & 0 & 0 & 0 & 0 & 0 & \c_3 & \c_2 \\
			0 & 0 & 0 & 0 & 0 & 0 & 0 & 0 & 0 & 2\c_3
		\end{pmatrix}
\end{align*}
in $\U$, with respect to the splitting~\eqref{eq: alpha 2 tensor split new}.
\end{lemma}

\begin{proof}

We write down the matrices with respect to the splittings~\eqref{eq: pull back split} and~\eqref{eq: alpha 2 tensor split new} for each of the involved factor.
Using that,
\[
	\sigma(i \de_g)(\xi)
		= \xi^\sharp
		= \sqrt{K(\xi)}\a_1^\sharp,
\]
we conclude that
\[
	\sigma(\de_g)(\xi)
		=
		i \sqrt{K(\xi)}
		\begin{pmatrix}
			0 & 0 & 1 & 0 & 0 & 0 & 0 & 0 & 0 & 0 \\
			0 & 0 & 0 & 0 & 0 & 0 & 0 & 0 & 1 & 0 \\
			0 & 0 & 0 & 0 & 0 & 0 & 0 & 0 & 0 & 1 \\
			0 & 0 & 0 & 0 & 0 & 0 & 0 & 1 & 0 & 0
		\end{pmatrix}.
\]
Further, we have
\begin{align*}
	2 \c \otimes_s (\cdot) 
		&= 
		\begin{pmatrix}
			2\c_1 & 0 & 0 & 0 \\
			\c_2 & \c_1 & 0 & 0 \\
			\c_3 & 0 & \c_1 & 0 \\
			\c_4 & 0 & 0 & \c_1 \\
			0 & \c_4 & 0 & \c_2 \\
			0 & 2\c_2 & 0 & 0 \\
			0 & 2\c_2 & 0 & 2\c_4 \\
			0 & 0 & \c_4 & \c_3 \\
			0 & \c_3 & \c_2 & 0 \\
			0 & 0 & 2\c_3 & 0
		\end{pmatrix} \\
	\tr_g
		&= \left( 0, 0, -2, 0, 0, 0, 1, 0, 0, 0 \right), \\
	g
		&= (0, 0, -1, 0, 0, 1, 2, 0, 0, 0)^T.
\end{align*}
Multiplying the matrices together gives the result.
\end{proof}

\begin{remark}
Note that
\[
	\c_3
		= - G(\c, \a_1)
		= - \frac1{\sqrt{K(\xi)}} G(\c, \xi)
\]
is either positive or negative, since both $\c$ and $\xi$ are causal one-forms.
\end{remark}

We are now ready to prove Theorem~\ref{thm: spectral gap}.

\begin{proof}[Proof of Theorem~\ref{thm: spectral gap}]
By \cite{H2015}*{Prop.~4.1}, the subprincipal operator (c.f.~\cite{H2015}*{Def.~3.8}) is given by
\begin{align*}
	S_{\mathrm{sub}}(\Lop_g)
		& = S_{\mathrm{sub}}(\Box_g) + S_{\mathrm{sub}}(2 \A_g \de_g \T_g ) \\
		& = - i \pi^* \n_{\H_G} + \s_1\left( 2 \A_g \de_g \T_g  \right)(\xi) \\
		& = - i \pi^* \n_{\H_G} - 2 h^{-1} \T_g \left( \c \otimes_s \cdot \right) \s_1(\de_g)(\xi) \T_g \\*
		& \qquad + e \cdot 2 h^{-1} \T_g \left( G(\c, \cdot) g \right) \circ \s_1(\de_g)(\xi) \T_g.
\end{align*}
By Lemma~\ref{le: upper triangular},
\begin{align*}
	i S_{\mathrm{sub}}(\Lop_g)
		& = \H_G + \mathfrak a N_1 + \frac{\sqrt{K(\xi)}}4 \left( D + N_2 \right) + e B,
\end{align*}
where 
\[
	D
		= \mathrm{diag} \left( 0,0,0,0, 0, 0, \c_3, \c_3, \c_3, 2 \c_3 \right),
\]
$N_1, N_2$ are upper triangular matrices with zeros on the diagonal, and $B$ is a fixed $e$-independent matrix, with respect to the splitting~\eqref{eq: alpha 2 tensor split new}.
We now choose
\[
	\mathfrak C(v, w)
		:= \ldr{v, E^2w},
\]
where
\[
	E
		:= \mathrm{diag} \left( 1, \frac \digamma \e, \hdots, \left( \frac \digamma \e \right)^9 \right)
\]
and $\ldr{\cdot, \cdot}$ denotes the standard inner product with respect to the splitting~\eqref{eq: alpha 2 tensor split new}.
The desired inequality
\[
	\pm \frac1{2i \abs{\s}} \left( S_{\mathrm{sub}}\left( \Lop_{b_0} \right) - S_{\mathrm{sub}}\left( \Lop_{b_0} \right)^{*\mathfrak C} \right) |_{\Gamma^\pm}
		\leq \e,
\]
with respect to $\mathfrak C$ is thus equivalent to
\[
	\pm \frac1{2i \abs{\s}} \left( \left( E  S_{\mathrm{sub}}\left( \Lop_{b_0} \right) E^{-1} \right) \left( \Lop_g \right) - \left( E S_{\mathrm{sub}} \left(\Lop_{b_0} \right) E^{-1} \right)^* \right) |_{\Gamma^\pm}
		\leq \e,
\]
with respect to the standard inner product $\ldr{\cdot, \cdot}$.
The diagonal matrix $D$ gives a non-positive contribution, and the matrix $e B$ can be made arbitrarily small by choosing $e_0$ small enough.
For the upper triangular matrices with zero on the diagonal, we note that their contributions $E N_j E^{-1}$, for $j = 1, 2$, are bounded by a constant times $\frac{\e}{\digamma}$, and can be made smaller than $\e$ by choosing $\digamma$ large enough.
The term $\H_G$ does not contribute to the skew-adjoint part of the subprincipal part. 
This finishes the proof.
\end{proof}

\section{Constraint damping (CD)} \label{sec: CD}

Recall the definition of $\T_g$ in~\eqref{eq: trace reversal}, and the function $f(r)$ from~\eqref{EqIf}. The theorem we want to prove in this section is the following:
\begin{thm}[Constraint damping] \label{thm: stable constraint propagation}
Let $(M, g)$ be a subextremal Kerr--de~Sitter spacetime extended beyond the horizons as in~\eqref{eq: manifold}.
There are $\kappa_0, h_0, e_0 > 0$ such that if
\[
	\kappa \in (0, \kappa_0), \quad h \in (0, h_0), \quad e \in (0, e_0),
\]
and
\[
	\A_g 
		:= - h^{-1}\T_g^{-1}\left( \c \otimes_s \o - e G(\c, \o) g \right),
\]
where
\begin{equation} \label{eq: c sharp}
	\c^{\sharp}
		:= \d_{t_*} + \frac a{r^2 + a^2} \d_{\phi_*} + \g f(r) \d_r + \g^2 \sin(\theta) \d_\theta,
\end{equation}
the following holds:
If $\o \in C^\infty(T^*M)$ is a non-trivial solution to
\begin{align} 
	\de_g \T_g \left( \de_g^* + \A_g \right) \o
		&= 0, \label{eq: constrained wave operator} \\
	\left( \L_T + i \s \right)^k \o
		&= 0, \nonumber
\end{align}
for some $k \in \N$, then 
\[
	\Im \ \s < 0.
\]
\end{thm}

\begin{remark}
The standard computation~\eqref{eq: wave equation identity} implies that~\eqref{eq: constrained wave operator} is a wave operator.
\end{remark}

We work with $h$ as a \emph{semiclassical parameter} in our analysis, where it is supposed to be small.
The properties will be analysed by means of propagation of semiclassical singularities.

\subsection{The general framework for constraint damping} \label{subsec: c general}

We will introduce our setup in general Lorentzian b-spacetimes, and only restrict to the Kerr--de~Sitter spacetime in the next section, with the spacetime specific computations and estimates.

\begin{assumption} \label{ass: M and c}
Let in this section $M$ be a manifold of dimension $n + 1 \geq 2$  with one smooth boundary component, equipped with a Lorentzian b-metric $g$, and let $\c$ be a smooth \emph{everywhere timelike} b-one-form.
\end{assumption}

\subsubsection{The semiclassical constraint operator}

The proof of Theorem~\ref{thm: stable constraint propagation} will generalize the arguments in \cite{HV2018}*{Sec.~8}, where the Schwarzschild-de~Sitter spacetime was considered.
The framework is based on semiclassical b-analysis, where the $h > 0$ in Theorem~\ref{thm: stable constraint propagation} is the semiclassical parameter.
Indeed, the idea is that we want to provide a damping term in~\eqref{eq: constrained wave operator} of order $h^{-1}$ in order to make the operator~\eqref{eq: constrained wave operator} invertible on modes with $\Im(\s) \geq - d$, for some $d > 0$.
We therefore naturally get at a semiclassical problem with semiclassical parameter $h$.
We define the semiclassical b-differential operator
\begin{align*}
	\P_e \o
		: = & \ 2 h^2 \de_g \left( \T_g \de_g^* \o - h^{-1} \left( \c \otimes_s \o - e G(\c, \o) g \right) \right) \\
		= & \ h^2 \Box + i \Lop_e - h^2 \Ric^\sharp
\end{align*} 
where
\begin{align*}
	\Lop_e \o
		:= & \ 2 i h \de_g \left( \c \otimes_s \o - e G(\c, \o) g \right),
\end{align*}
and where $h$ is the semiclassical parameter.
The semiclassical order of $\P_e$ is $0$, whereas the differential order of $\P_e$ is $2$.

\begin{lemma} \label{le: semiclassical principal symbol}
Given any $e \in \R$, the semiclassical principal symbol $\p_e$ of $\P_e$ at $\xi \in \bT_x^* M$ is the endomorphism on $\pi^*\left( \bT^* M\right)|_{(x, \xi)} =  \bT_x^* M$ given by
\[
	\p_e(\xi)
		= G(\xi, \xi) \id + i \l_e(\xi),
\]
where 
\[
	\l_e(\xi) \o
		= G(\c, \xi) \o + G(\xi, \o) \c - 2 e G(\o, \c) \xi
\]
is the principal symbol of $\Lop_e$.
\end{lemma}
\begin{proof}
We first note that
\[
	\P_e \o
		= h^2 \Box \o - h^2 \o(\Ric_g^\sharp) - 2h \de_g \left( \c \otimes_s \o \right) - 2h e \md \left( G(\c, \o) \right).
\]
The principal symbol is now readily computed using that
\begin{align*}
	\de_g( \c \otimes_s (f \o))
		&= - \frac12 G(\md f, \c) \o - \frac12 G(\md f, \o) \c,
\end{align*}
for all $f \in C^\infty(M)$.
\end{proof}

\subsubsection{The modified constraint operator}

As discussed in \cite{HV2018}*{Rmk.~8.7}, at least for the natural choice of $\c$ in the Schwarzschild-de~Sitter spacetime, the semiclassical principal symbol $\p_e$ is not symmetric with respect to any positive definite inner product on $\bT^* M$ when $e = 0$.
However, for any $e > 0$ it is possible to find such an inner product.
In order for the inner product to be independent of $e$, let us first modify $\P_e$ by a conjugation which depends on $e > 0$:
\begin{align*}
	\tP_e 
		:= & \ b_{(2e)^{1/2}} \P_e b_{(2e)^{-1/2}} \\
		= & \ \tBox_e + i \tLop_e - h^2 \tilde \Ric^\sharp,
\end{align*}
where
\begin{align} 
	\tBox_e
		:= & \ h^2 b_{(2e)^{1/2}} \Box b_{(2e)^{-1/2}}, \nonumber \\
 	\tLop_e
 		:= & \ b_{(2e)^{1/2}} \Lop_e b_{(2e)^{-1/2}}, \label{eq: definition tilde L e} \\
 	\tilde \Ric_e^\sharp
 		:= & \ b_{(2e)^{1/2}} \Ric_g^\sharp b_{(2e)^{-1/2}} \nonumber
\end{align}
and where
\begin{align*}
	b_{t}
		:=& \ \left(t - 1\right) \frac{\c \otimes \c^\sharp}{G(\c, \c)} + \id,
\end{align*}
with inverse $b_{t^{-1}}$. 
Note that $b_t$ simply scales the $\c$-direction by the factor $t$ and leaves the orthogonal directions (with respect to $G$) unchanged.
The term $- h^2 \tilde \Ric_g^\sharp$ will be irrelevant for the analysis, as it contributes neither to the leading order part of the operator, nor to the leading order part of the skew-adjoint part of the operator.
Semiclassical estimates for $\tLop_e$ translate to semiclassical estimates for $\Lop_e$ and vice versa, but with bounds depending on $e > 0$. 
We can therefore just as well work with $\tLop_e$ instead of $\Lop_e$.
The semiclassical principal symbol of $\tP_e$ is given by
\begin{equation} \label{eq: tilde p e principal symbol}
	\tilde \p_e(\xi)
		= G(\xi, \xi) \id + i \tl_e(\xi),
\end{equation}
where $\tl_e$ is the principal symbol of $\tLop_e$.

Now, recall that we conjugated $\P_e$ with $b_{(2e)^{1/2}}$ in order to be able to construct a positive definite inner product $\mB$ on $\bT^* M$ for which $\tl_e$ is symmetric.
We will prove in Proposition~\ref{prop: tilde L e properties} below that the following is a suitable choice:
\[
	\mB
		:= G - 2 \frac{\c^\sharp \otimes \c^\sharp}{G(\c, \c)}.
\]
Since $G$ is positive definite on $\c^\perp \otimes \c^\perp$ and $\mB(\c, \c) = - G(\c, \c) > 0$, it follows that $\mB$ is positive definite.
The (scalar) first term $G(\xi, \xi)\id$ in $\tilde \p_e$ is clearly symmetric with respect to $\mB$.
The following proposition analyses $\tl_e$.

\begin{prop} \label{prop: tilde L e properties} \
\begin{enumerate}[(a)]
	\item 
	\label{item: tilde L e smooth} 
	The family $\tLop_e$ of semiclassical differential operators depends smoothly on $e > 0$.
	\item The semiclassical principal symbol of $\tLop_e$ is given by
	\begin{equation} \label{eq: tilde l e}
	\begin{split}
		\tl_e(\xi) \o
			= & \ G(\c, \xi) \o + \sqrt{2e} \left( G(\xi, \o) \c - G(\o, \c) \xi \right) \\
			& \ + (1 - 2 e) \frac{G(\o, \c) G(\xi, \c)}{G(\c, \c)} \c,
	\end{split}
	\end{equation}
	for any $\xi, \o \in \bT^*M$.
	In particular, 
	\begin{equation} \label{eq: tilde l 0}
		\lim_{e \to 0} \tl_e(\xi)
			= G(\c, \xi) 
				\left( \id + \frac{\c \otimes \c^\sharp}{G(\c, \c)} \right),
	\end{equation}
	for any $\xi \in \bT^* M$.
	Moreover, for any $e > 0$, $\tl_e(\xi)$ is symmetric with respect to $\mB$. \label{item: tilde L e principal symbol}
	\item \label{item: tilde L e skew-adjoint part} 
	Assume that 
	\begin{equation} \label{eq: assumptions critical point}
	\begin{split}
		\n_{\c^\sharp} \c|_p 
			&= 0, \\
		\md G(\c, \c)|_p 
			&= 0,
	\end{split}
	\end{equation}
	at a point $p \in M$.
	Then the skew-adjoint part of $\tLop_e$, with respect to $\mB$, is given by
	\[
		\frac{\tLop_e - \tLop_e^{*\mB}}{2 i h}|_p
			=
				\begin{pmatrix}
					(1 + e) \de_g(\c)|_p & 0 \\
					0 & \frac12\de_g(\c) \id |_p - \de_g^*(\c)^\sharp|_p
				\end{pmatrix},
	\]
	with respect to the splitting $\bT^* M = \R \c \oplus \c^\perp$.
\end{enumerate}
\end{prop}

\begin{remark}
In our application of this to the Kerr--de~Sitter spacetime, the assumptions~\eqref{eq: assumptions critical point} turn out to be naturally satisfied wherever the Hamiltonian vector field of the principal symbol $\tl_e$ (essentially given by $\c^\sharp$ for small $e$, c.f.~\eqref{eq: tilde l 0}) vanishes.
\end{remark}

\begin{proof}[Proof of Proposition~\ref{prop: tilde L e properties}]
Assertion~\eqref{item: tilde L e smooth} is clear by construction, c.f.~\eqref{eq: definition tilde L e}.

\noindent
For assertion~\eqref{item: tilde L e principal symbol}, recall from Lemma~\ref{le: semiclassical principal symbol} that
\[
	\l_e(\xi) \o
		= G(\c, \xi) \o + G(\xi, \o) \c - 2 e G(\o, \c) \xi.
\]
It follows that
\begin{align*}
	\l_e(\xi) \left( \frac{\c \otimes \c^\sharp}{G(\c, \c)} \o \right)
		&= 2 \left( \frac{G(\o, \c) G(\xi, \c)}{G(\c, \c)} \c - e G(\o, \c)\xi \right), \\
	\left( \frac{\c \otimes \c^\sharp}{G(\c, \c)} \right) \l_e(\xi) \o
		&= \left( (1 - 2e) \frac{G(\o, \c) G(\xi, \c)}{G(\c, \c)} + G(\o, \xi) \right) \c.
\end{align*}
We use this to compute
\begin{align*}
	\tl_e(\xi) \o
		&= b_{(2e)^{1/2}} \l_e(\xi) b_{(2e)^{-1/2}} \\
		&= \left( \left( (2e)^{\frac12} - 1 \right) \frac{\c \otimes \c^\sharp}{G(\c, \c)} + \id \right) \l_e(\xi) \left( \left( (2e)^{-\frac12} - 1 \right) \frac{\c \otimes \c^\sharp}{G(\c, \c)} + \id \right) \\
		&= 2 \left( (2e)^{\frac12} - 1 \right) \left( (2e)^{-\frac12} - 1 \right) (1 - e) \frac{G(\o, \c) G(\xi, \c)}{G(\c, \c)} \c \\
		&\qquad + \left( (2e)^{\frac12} - 1 \right)\left( (1 - 2e) \frac{G(\o, \c) G(\xi, \c)}{G(\c, \c)} + G(\o, \xi) \right) \c \\
		&\qquad + 2 \left( (2e)^{-\frac12} - 1 \right) \left( \frac{G(\o, \c) G(\xi, \c)}{G(\c, \c)} \c - e G(\o, \c)\xi \right) \\
		&\qquad + G(\c, \xi) \o + G(\xi, \o) \c - 2 e G(\o, \c) \xi \\
		&= G(\c, \xi) \o + (2e)^{\frac12} \left( G(\o, \xi) \c - G(\o, \c) \xi \right) + (1 - 2e) \frac{G(\o, \c) G(\xi, \c)}{G(\c, \c)}\c,
\end{align*}
proving~\eqref{eq: tilde l e}.
Inserting $e = 0$ in~\eqref{eq: tilde l e}, we get
\[
	\lim_{e \to 0} \tl_e(\xi) \o
		= G(\c, \xi) \left( \o + \frac{G(\o, \c)} {G(\c, \c)}\c \right),
\]
proving~\eqref{eq: tilde l 0}.
The first and last terms in~\eqref{eq: tilde l e} are clearly symmetric with respect to $\mB$, so we only need to check the second term:
\begin{align*}
	\mB( G(\xi, \o) \c - G(\o, \c) \xi, \eta)
		&= G(\xi, \o) G(\c, \eta) - G(\o, \c) G(\xi, \eta) \\*
		&\qquad - 2 \left( G(\xi, \o) G(\c, \eta) - \frac{G(\o, \c) G(\xi, \c) G(\c, \eta)}{G(\c, \c)} \right) \\
		&= - G(\xi, \o) G(\c, \eta) - G(\o, \c) G(\xi, \eta) \\
		&\qquad + 2 \frac{G(\o, \c) G(\xi, \c) G(\c, \eta)}{G(\c, \c)},
\end{align*}
which is symmetric in $\o$ and $\eta$.
This proves assertion~\eqref{item: tilde L e principal symbol}.

\noindent
For assertion~\eqref{item: tilde L e skew-adjoint part}, using that $b_t$ is symmetric with respect to $\mB$ for all $t \in \R$, we note that
\begin{align*}
	\frac{\tLop_e - \tLop_e^{*\mB}}{2 i h}|_p
		&= \frac{b_{(2e)^{1/2}} \Lop_e b_{(2e)^{-1/2}} - \left( b_{(2e)^{1/2}} \Lop_e b_{(2e)^{-1/2}} \right)^{*\mB}}{2 i h}\Big|_p \\
		&= b_{(2e)^{-1/2}} \frac{ b_{2e} \Lop_e - \left( b_{2e} \Lop_e \right)^{*\mB}}{2 i h} b_{(2e)^{-1/2}}\Big|_p.
\end{align*}
Consequently, assertion~\eqref{item: tilde L e skew-adjoint part} is equivalent to
\begin{equation} \label{eq: simplified skew-adjoint}
	\frac{ b_{2e} \Lop_e - \left( b_{2e} \Lop_e \right)^{*\mB}}{2 i h}\Big|_p
		= \de_g(\c)|_p \left( 2e (1 + e) \cdot \id|_{\R \c} + \frac12 \id|_{\c^\perp} \right) - \de_g^* (\c)^\sharp \id|_{\c^\perp} \Big|_p.
\end{equation}
Since this is an endomorphism, we can without loss of generality restrict the following computations to one-forms $\o$ satisfying $\n \o|_p = 0$.
Under this condition on $\o$ at $p$, we first note that
\[
	\frac{\Lop_e \o}{2 i h}(X)|_p
		= \frac12 \left( \de_g (\c) \o(X) - \n_{\o^\sharp} \c(X) + 2 e G(\n_X \c, \o) \right)|_p.
\]
Since
\[
	b_{2e}
		= \id - (1 - 2 e) \frac{\c \otimes \c^\sharp}{G(\c, \c)},
\]
we first compute, using that $\n_{\c^\sharp}\c|_p = 0$ and $\md G(\c, \c)|_p = 0$, the term
\begin{align*}
	- (1 - 2 e) 
		&\frac{\c \otimes \c^\sharp}{G(\c, \c)} \frac{\Lop_e \o}{2 i h}\Big|_p \\
		&= - \frac{1 - 2 e}2 \left( \de_g(\c) \frac{\o(\c^\sharp)}{G(\c, \c)} - \frac{G(\n_{\o^\sharp}\c, \c)}{G(\c, \c)} + 2e G(\n_{\c^\sharp}\c, \o) \right) \c \Big|_p \\
		&= - \frac{1 - 2 e}2 \de_g(\c) \frac{\o(\c^\sharp)}{G(\c, \c)} \c \Big|_p,
\end{align*}
where we used that $G(\n_{\o^\sharp}\c, \c)|_p = \frac12 \md G(\c, \c)(\o^\sharp)|_p = 0$.
We thus conclude that
\begin{equation} \label{eq: first skew adjoint part}
\begin{split}
	\frac{b_{2e} \Lop_e \o}{2 i h} (X)\Big|_p
		=& \ \frac12 \left( \de_g (\c) \o(X) - \n_{\o^\sharp} \c(X) + 2 e G(\n_X \c, \o) \right)\Big|_p \\
		& \ - \frac{1-2e}2 \de_g(\c) \frac{G(\o, \c)}{G(\c, \c)}\c(X)\Big|_p.
\end{split}
\end{equation}
For the adjoint term, let us first note that $\mB(\o, \eta) = G(b_{-1} \o, \eta) = G(\o, b_{-1} \eta)$ for all one-forms $\o$ and $\eta$.
It follows that
\begin{align*}
	\int_M \mB(b_{2e} \Lop_e \o, \eta) \md \Vol_g
		&= \int_M G(b_{2e} \Lop_e \o, b_{-1} \eta) \md \Vol_g \\
		&= \int_M G(\o, \Lop_e^{*G} b_{2e} b_{-1} \eta) \md \Vol_g \\
		&= \int_M \mB(\o, b_{-1} \Lop_e^{*G} b_{-2e} \eta) \md \Vol_g,
\end{align*}
so we conclude that $\left( b_{2e} \Lop_e \right)^{*\mB} = b_{-1} \Lop_e^{*G} b_{-2e}$.
Since $\de_g^*\o$ is a symmetric tensor, we note that
\[
	\Lop_e^{*G}\o
		= - 2 i h \left( \de_g^* (\o) (\c^\sharp, \cdot) + e \de_g(\o) \c \right).
\]
We compute, using that $\n_{\c^\sharp}\c|_p = 0$ and $\md G(\c, \c)|_p = 0$,
\begin{align*}
	\de_g \left( \frac{\c \otimes \c^\sharp}{G(\c, \c)} \o \right)\Big|_p
		&= \de_g \left( \frac{G(\o, \c)}{G(\c, \c)}\c \right)\Big|_p \\
		&= \frac{G(\o, \c)}{G(\c, \c)} \de_g(\c)\Big|_p, \\
	\de_g^*\left(\frac{\c \otimes \c^\sharp}{G(\c, \c)} \o \right)(\c^\sharp, X)\Big|_p
		&= \frac12 \n_{\c^\sharp} \left( \frac{G(\o, \c) \c}{G(\c, \c)} \right)(X)\Big|_p + \frac12 \n_X \left( \frac{G(\o, \c) \c}{G(\c, \c)} \right)(\c^\sharp)\Big|_p \\
		&= \frac12 G(\o, \n_X \c)\Big|_p.
\end{align*}
Since $\Lop_e^{*G} \o|_p = 0$, we get
\begin{align*}
	\frac{- \Lop_e^{*G} b_{-2e} \o }{2 i h} (X)\Big|_p 
		&= \frac{-1}{2 i h} \Lop_e^{*G} \left( \o - (1 + 2e) \frac{\c \otimes \c^\sharp}{G(\c, \c)} \o\right) (X)\Big|_p \\
		&= \frac{1 + 2e}{2 i h} \Lop_e^{*G} \left(\frac{\c \otimes \c^\sharp}{G(\c, \c)} \o\right) (X)\Big|_p \\
		&= - \frac{1 + 2e}2 \left( G(\o, \n_X \c) + 2 e \de_g(\c) \frac{G(\o, \c)}{G(\c, \c)} \c(X) \right)\Big|_p.
\end{align*}
From this, we finally compute
\begin{align*}
	\frac{- \left( b_{2e} \Lop_e \right)^{*\mB} \o}{2 i h}(X)\Big|_p 
		&= - \frac{b_{-1} \Lop_e^{*G} b_{-2e} \o }{2 i h} (X)\Big|_p \\
		&= - \frac{1 + 2e}2 \left( G(\o, \n_X \c)  - 2 e \de_g(\c) \frac{G(\o, \c)}{G(\c, \c)} \c (X) \right)\Big|_p.
\end{align*}
Adding this to~\eqref{eq: first skew adjoint part} verifies~\eqref{eq: simplified skew-adjoint}.
This completes the proof.
\end{proof}

\begin{remark}
Once can check that at points where $\n_{\c^\sharp} \c|_p \neq 0$, the subcritical part of $\tLop_e$ grows like $e^{-1}$ as $e \to 0$.
It is therefore natural to assume that $\n_{\c^\sharp} \c|_p = 0$, as we do in the lemma.
\end{remark}

\subsubsection{Ellipticity at intermediate frequencies}

The following proposition implies that we only need to analyze the low frequency and high frequency behaviour for $\tP_e$:

\begin{prop}
\label{prop: ellipticity finite frequencies}
For any compact subset $K \subset \bT^* M \backslash o$, there is an $e_0 > 0$ such that
\[
	K
		\subseteq \ell(\tP_e)
\]
for all $e \in (0, e_0)$.
\end{prop}

\begin{proof}
By continuity in $e$, it suffices to prove that $\tilde \p_e(\xi)$ is invertible for all $\xi \neq 0$ in the limit $e \to 0$.
By~\eqref{eq: tilde p e principal symbol} and~\eqref{eq: tilde l 0},
\[
	\lim_{e \to 0}\tilde \p_e(\xi)
		= \begin{pmatrix}
			G(\xi, \xi) + 2 i G(\xi, \c)  & 0 \\
			0 & G(\xi, \xi) + i G(\xi, \c)
		\end{pmatrix}
\]
for any $\xi \in \bT^* M$, with respect to the splitting $\bT^* M  = \ldr{\c} \oplus \c^\perp$.
This is invertible for $\xi \neq 0$, since $G(\xi, \xi) = G(\c, \xi) = 0$ would imply that $\xi = 0$ since $\c$ is timelike.
\end{proof}

\subsubsection{The setup for low frequency analysis}
\label{sec: low frequency preparation}

Proposition~\ref{prop: ellipticity finite frequencies} implies that we only need to prove semiclassical estimates at (small open neighborhoods of) fiber infinity $\bS M \subset \boT M$ and the zero section $o \subset \boT M$.
In this subsection, we focus on a small open neighborhood of the zero section.
Note that the first term $G(\xi, \xi)\id$ in $\tilde p_e(\xi)$, see~\eqref{eq: tilde p e principal symbol}, vanishes quadratically at the zero section, whereas the second term $i \tl_e(\xi)$ only vanishes to first order.
Following the argument in \cite{HV2018}*{Sec.~8.3}, it is therefore convenient to again modify the operator in order to write it as the sum of $\tLop_e$ and higher order vanishing terms at the zero section.
Since $1 + i \tLop_e^*$ is elliptic at the zero section, we may consider
\[
	\hat \P_e
		:= - i (1 - i \tLop_e^*) \tP_e
		= \tLop_e + \Jop_e + i \Qop_e,
\]
where
\begin{align*}
	\Jop_e
		:=& \ - \Re \left( \tLop_e^* \left( \tBox_e - h^2 \tilde \Ric^\sharp \right) \right) + \Im \left( \tBox_e - h^2 \tilde \Ric^\sharp \right), \\
	\Qop_e
		:=& \ - \tLop_e^*\tLop_e - \Im \left( \tLop_e^* \left( \tBox_e - h^2 \tilde \Ric^\sharp \right) \right) - \Re\left( \tBox_e - h^2 \tilde \Ric^\sharp \right), 
\end{align*}
where the formal adjoint and the real and imaginary parts are defined with respect to $\mB$.
For a principally scalar operator 
\[
	\Aop = \Aop^{*\mB} \in \Psi_{\be, h}^{k,l,m}(M; \bT^* M)
\]
(here $k, l, m$ are the differential order, the b-decay order and the semiclassical order, respectively), we get
\begin{align}
	& \Im \ldr{\frac1h \Aop \hat \P_e u, \Aop u} - \Re \ldr{\frac1h \Aop \Qop_e u, \Aop u} \nonumber \\
	&\qquad = \ldr{\frac i{2h} \Aop u, \Aop (\tLop_e + \Jop_e) u} - \ldr{\frac i{2h} \Aop (\tLop_e + \Jop_e) u, \Aop u} \nonumber \\
	&\qquad = \ldr{\frac i{2h} \left[ \tLop_e + \Jop_e, \Aop^2 \right] + \frac1 {2 i h} \left( \tLop_e - \tLop_e^{*\mB} \right) \Aop^2 u, u} \label{eq: commutator equality}
\end{align}
for any $u \in \dot C_c^\infty(M)$.
In practice, we will be interested in estimates on the weighted Sobolev spaces $H^{k,\rho}_{\be, h}$, where $\rho \in \R$ is the b-weight.
Near the zero section $o \subset \boT M$, we can ignore the differential order and consider $H^{0,\rho}_{\be, h}$ for simplicity.
Moreover, it suffices to consider commutants $\Aop \in \Psi_{\be, h}^{-\infty,\rho,0}(M; \bT^* M)$.
Estimating the right-hand-side of~\eqref{eq: commutator equality} for suitably chosen $\Aop$ will be the main goal in Section~\ref{sec: low frequency estimates} below.
Estimating the left-hand-side is less delicate and can be done rather generally as follows.
\begin{lemma} \label{le: the right-hand-side low}
Let $\de > 0$ and $\rho \in \R$.
Assume that $\Aop \in \Psi_{\be, h}^{-\infty, \rho, 0}$ and that the principal symbol $\as$ of $\Aop$ is supported within $\{\abs{\xi} \leq \de\} \subseteq \boT M$.
If $\Sop \in \Psi_{\be, h}^{-\infty, 0, 0}(M; \bT^* M)$ satisfies
\[
	\WFb(\Aop)
		\subseteq \ell(\Sop),
\]
then
\begin{align*}
	- \Im \ldr{\frac1h \Aop \hat \P_e u, \Aop u} +
		&\Re \ldr{\frac1h \Aop \Qop_e u, \Aop u} \\
		&\lesssim_{k, e} \de \norm{\Aop u}_{L^2}^2 + h^{-1} \norm{\Sop \hat \P_e u}_{H^{0, \rho}_{\be, h}} \norm{\Aop u}_{L^2} \\
		&\qquad + h \norm{\Sop u}^2_{H^{0, \rho}_{\be, h}} + h^k \norm{u}^2_{H^{0, \rho}_{\be, h}}
\end{align*}
for all $u \in \dot C^\infty_c(M)$, $k \in \N$, $e \in (0, e_0)$ and $h \in (0, h_0)$, for some $e_0, h_0 > 0$.
\end{lemma}
\begin{proof}
Microlocal elliptic regularity theory implies that
\begin{align*}
	- \Im \ldr{\frac1h \Aop \hat \P_e u, \Aop u}
		&\leq h^{-1}\norm{\Aop \hat \P_e u}_{L^2} \norm{\Aop u}_{L^2} \\
		&\lesssim_{k, e} \left( h^{-1} \norm{\Sop \hat \P_e u}_{H^{0, \rho}_{\be, h}} + h^k \norm{u}_{H^{0, \rho}_{\be, h}} \right) \norm{\Aop u}_{L^2} \\
		&\lesssim_{k, e} h^{-1} \norm{\Sop \hat \P_e u}_{H^{0, \rho}_{\be, h}} \norm{\Aop u}_{L^2} + h^k \norm{\Sop u}_{H^{0, \rho}_{\be, h}}^2 + h^k \norm{u}_{H^{0, \rho}_{\be, h}}^2.
\end{align*}

It therefore remains to consider the term involving $\Qop_e$.
Since $\Aop = \Aop^*$ and $\Qop_e = \Qop_e^*$, we note that 
\begin{align*}
	h^{-1} \Re \ldr{\Aop \Qop_e u, \Aop u}
		&= \Re \ldr{h^{-1}[\Aop, \Qop_e] u, \Aop u} + h^{-1} \ldr{\Qop_e \Aop u, \Aop u} \\
		&= h \Re \ldr{h^{-1} [\Aop, h^{-1} [\Aop, \Qop_e]] u, u} + h^{-1} \ldr{\Qop_e \Aop u, \Aop u}.
\end{align*}
The first term is estimated as 
\[
	h \Re \ldr{h^{-1} [\Aop, h^{-1} [\Aop, \Qop_e]] u, u}
		\lesssim_{k,e} h \norm{\Sop u}^2_{H^{0, \rho}_{\be, h}} + h^k \norm{u}^2_{H^{0, \rho}_{\be, h}},
\]
for any $k \in \N$.

Let us now turn to the second term.
The semiclassical principal symbol of $- \Qop_e$ is given by
\[
	\q_e(\xi)
		= \tl_e(\xi)^2 + G(\xi, \xi) \id.
\]
By~\eqref{eq: tilde l 0}, we note that
\[
	\lim_{e \to 0} \q_e(\xi)
		= G(\c, \xi)^2 \left( \id + \frac{\c \otimes \c^\sharp}{G(\c, \c)} \right)^2 + G(\xi, \xi) \id.
\]
Since $\c$ is everywhere timelike, this implies that $\lim_{e \to 0} \q_e(\xi)$ is a positive definite quadratic form.
Consequently, $\q_e(\xi)$ is a positive definite quadratic form for all $e \in (0, e_0)$ if $e_0 > 0$ is small enough, with a non-degenerate limit as $e \to 0$. 
Let now $\n$ be a b-connection on $\bT^* M$.
Define $\mathcal Q_e$ to be the unique smooth self-adjoint endomorphism of $\bT^* M \otimes \bT^* M$ such that
\[
	\mB \left( \mathcal Q_e(\xi \otimes \eta), (\xi \otimes \eta) \right)
		= \mB \left( \q_e(\xi) \eta, \eta \right).
\]
It follows that $\n^* \mathcal Q_e \n$ is a second order differential operator, where $\n^*$ is the formal adjoint of $\n: C^\infty(\bT^* M) \to C^\infty(\bT^* M \otimes \bT^* M)$, with respect to $\mB$.
Consequently, the full symbols of both $\Qop_e$ and $h^2 \n^* \mathcal Q_e \n$ vanish quadratically in $(h, \xi)$, and the principal symbols coincide.
It therefore follows that
\[
	- \Qop_e - h^2 \n^* \mathcal Q_e \n
		= F_1(h \n) + h^2 F_2
\]
for some smooth homomorphism fields $F_1, F_2$.
We thus get
\begin{align}
	h^{-1} \ldr{\Qop_e \Aop u, \Aop u}
		= & \ - h \ldr{\mathcal Q_e (\n \Aop u), \n \Aop u} \nonumber \\*
		& - \ldr{ \left( F_1(h \n) + h F_2 \right) \Aop u, \Aop u} \nonumber \\
		\leq & \norm{\left( F_1(h \n) + h F_2 \right) \Aop u}_{L^2} \norm{\Aop u}_{L^2}. \label{eq: first lower bound Qe}
\end{align}
The semiclassical principal symbol of $\left( F_1(h \n) + h F_2 \right) \Aop$ is given by $F_1(\xi) \circ \as(\xi)$, which is supported in $\{ \abs \xi \leq \de \}$.
Hence 
\[
	F_1(\xi) \circ \as (\xi)
		= \chi\left( \abs{\xi} \right) F_1(\xi) \circ \as (\xi),
\]
for any $\chi \in C_c^\infty(\R)$ with $\chi(s) = 1$ for all $\abs s \leq \de$.
Since $F_1(\xi)$ is linear in $\xi$, it follows that
\[
	\abs{ \chi\left( \abs{\xi} \right) F_1(\xi) }
		\leq C \de,
\]
for some $C > 0$.
We conclude that
\begin{align*}
	\norm{\left( F_1(h \n) + h F_2 \right) \Aop u}_{L^2} 
		&\lesssim_{e} \norm{\Op_h\left(\chi F_1 \right) \Aop u}_{L^2} + h \norm{\Aop u}_{L^2} \\
		&\lesssim_{e} \de \norm{\Aop u}_{L^2}.
\end{align*}
Combining this with~\eqref{eq: first lower bound Qe} finishes the proof.
\end{proof}

\subsubsection{The setup for high frequency analysis} \label{sec: high freq general}

Since Proposition~\ref{prop: ellipticity finite frequencies} implies that we only need to prove semiclassical estimates at (small open neighborhoods of) fiber infinity $\bS M \subset \boT M$ and the zero section $o \subset \boT M$, and the previous subsection prepares for the low-frequency estimates, we focus here on the high frequency estimates. 
By~\eqref{eq: tilde p e principal symbol}, the principal symbol of $\tP_e$, restricted to fiber infinity $\bS M$ is given by
\[
	\hat \rho^2 \tilde \p_e|_{\bS M}
		= \hat \rho^2 G \ \id|_{\bS M},
\]
where $\hat \rho$ is a boundary defining function for $\bS M \subseteq \bT M$.
The analysis of our operator $\tP_e$ at (and by continuity near) fiber infinity $\bS M$ is therefore given by estimates for linear wave equations on Kerr--de~Sitter spacetime, with one exception; the skew-adjoint part of $\tP_e$ depends on more than $\hat \rho^2 \tilde \p_e|_{\bS M}$.
We consider the two parts of the characteristic sets, defined by
\[
	\Sigma^\pm
		:= \{ \xi \in \bS M \mid G(\xi, \xi) = 0\ \text{and}\ {\mp}G(\c, \xi) > 0 \}.
\]
Since $\c$ is globally timelike, this really splits up the characteristic set into two disjoint pieces.

\begin{remark}
Note that the sign convention $\Sigma^\pm$ implies that if $\c$ is future directed, then elements in $\Sigma^+$ are also future directed.
As a consequence, the Hamiltonian vector field will propagate elements in $\Sigma^+$ to the future.
\end{remark}

\noindent
The following observation is crucial for the analysis.

\begin{lemma} \label{le: skew adjoint part}
The principal symbol of the skew-adjoint part of $\tP_e$ is given by
\[
	\s_{\be, h} \left( \frac1{2i} \left( \tP_e - \tP_e^{*\mB} \right) \right)
		= \tl_e(\xi).
\]
Consequently,
\[
	\lim_{e \to 0} \s_{\be, h} \left( \frac1{2i} \left( \tP_e - \tP_e^{*\mB} \right) \right)
		= G(\c, \xi) 
				\left( \id + \frac{\c \otimes \c^\sharp}{G(\c, \c)} \right).
\]
In particular, there exists $e_0 > 0$ such that if $e \in (0, e_0]$, then 
\[
	\mp \s_{\be, h} \left( \frac1{2i} \left( \tP_e - \tP_e^{*\mB} \right) \right)\Big|_{\Sigma^\pm}
		> 0,
\]
as quadratic forms.
\end{lemma}
\begin{proof}
The first assertion is immediate from~\eqref{eq: tilde p e principal symbol}, the second assertion follows from~\eqref{eq: tilde l 0} and the third assertion follows by (pre)compactness of $\Sigma^\pm$.
\end{proof}

\noindent
As we will see below, this means in practice that the estimates propagate nicely \emph{forwards} in $\Sigma^+$ and \emph{backwards} in $\Sigma^-$ along the Hamiltonian flow, perfectly suitable for solving the \emph{forward in time} Cauchy problem.

\subsection{The damping vector field in Kerr--de~Sitter spacetimes} \label{subsec: partial compactification}

Let us now work concretely with the Kerr--de~Sitter spacetime.
\subsubsection{The partial compactification}
We begin by adding a border at an appropriate timelike infinity in the Kerr--de~Sitter spacetime $(M,g)$, given by~\eqref{eq: manifold} and~\eqref{eq: metric}.
Let us introduce the new coordinates
\begin{align*}
	\begin{pmatrix}
		s \\
		\psi
	\end{pmatrix}
		= 
	\begin{pmatrix}
		e^{-t_*} \\
		\phi_* - \frac a{r_0^2 + a^2} t_*
	\end{pmatrix}
\end{align*}
on $M$.
In these coordinates, the metric~\eqref{eq: metric} now takes the form 
\begin{equation}
\begin{split}
	g
		=& \ - \frac{\mu(r)}{b^2\varrho^2}\left( - \frac{r_0^2 + a^2 \cos^2(\theta)}{r_0^2 + a^2} \frac{\md s}{s} - a \sin^2(\theta) \md \psi \right)^2 + \varrho^2 \frac{1 - f(r)^2}{\mu(r)} \md r^2 \\*
		& \ - \frac 2 b f(r) \left( - \frac{r_0^2 + a^2 \cos^2(\theta)}{r_0^2 + a^2} \frac{\md s}{s} - a \sin^2(\theta) \md \psi \right)\md r \\*
		& \ + \frac{c(\theta)\sin^2(\theta)}{b^2\varrho^2}\left(- a \frac{r_0^2 - r^2}{r_0^2 + a^2} \frac{\md s}{s} - \left(r^2 + a^2\right) \md \psi\right)^2 + \varrho^2\frac{\md \theta^2}{c(\theta)}
\end{split}
\label{eq: compactification metric}
\end{equation}
on 
\[
	M
		\cong (0, \infty)_{s} \times (r_e - 2\e_M, r_c + 2\e_M)_r \times \mathbb{S}^2_{\psi, \theta}.
\]
In particular, this computation shows the following:
\begin{lemma}[The appropriate compactification]
Consider the coordinate system $(s, r, \psi, \theta)$ on 
\[
	M 
		\cong (0, \infty)_{s} \times (r_e - 2\e_M, r_c + 2\e_M)_r \times \mathbb{S}^2_{\psi, \theta}.
\]
The metric $g$ in these coordinates extends uniquely to a b-metric, which we also denote by $g$, on the Kerr--de~Sitter spacetime with a boundary at timelike infinity:
\[
	\mathcal M
		:= [0, \infty)_{s} \times (r_e - 2\e_M, r_c + 2\e_M)_r \times \mathbb{S}^2_{\psi, \theta}.
\]
\end{lemma}

\begin{remark}
Note that this $r_0$ is different from the one used in \cite{PV2021}, where $r_0$ was the critical point of $\mu$. 
Instead, this is a natural generalization of the critical point used for constraint damping in \cite{HV2018}, which gives rise to a \emph{different} compactification than in \cite{PV2021}.
This compactification was however treated in the more general \cite{PV2023}.
\end{remark}

\subsubsection{The damping vector field}

We let $\c$ be as in Theorem~\ref{thm: stable constraint propagation} and check that Assumption~\ref{ass: M and c} is satisfied:

\begin{prop}[The damping vector field] \label{prop: c smooth and timelike}
Consider the partially compactified Kerr--de~Sitter spacetime $(\mathcal M, g)$.
If $\ep_M > 0$ is small enough, then the following holds.
\begin{itemize}
	\item The vector field $\c^\sharp$ in Equation~\eqref{eq: c sharp} extends to a smooth b-vector field on $\mathcal M$, given in the coordinates $(s, r, \psi, \theta)$ by
\begin{equation} \label{eq: c sharp in new coordinates}
	\c^\sharp
		= - s \d_{s} + a \frac{r^2 - r_0^2}{(r^2 + a^2)(r_0^2 + a^2)} \d_{\psi} + \g f(r) \d_r + \g^2 \sin(\theta) \d_\theta.
\end{equation}
	\item As a smooth vector field on $\mathcal M$, $\c^\sharp$ vanishes precisely at the two points
\begin{align*}
	p_N
		:= & \ \{s = 0\} \cap \{r = r_0 \} \cap \{\theta = 0\}, \\
	p_S
		:= & \ \{s = 0\} \cap \{ r = r_0 \} \cap \{\theta = \pi\}.
\end{align*}
	\item There is a $\g_0 > 0$ such that $\c^\sharp$ is timelike $b$-vector field if $\g \in (0, \g_0)$.
\end{itemize}
\end{prop}

\begin{remark} \label{rmk: norm of c sharp}
A simple computation using either~\eqref{eq: metric} or~\eqref{eq: compactification metric} shows that
\begin{equation} \label{eq: c causality}
\begin{split}
	\varrho^{-2}g(\c^\sharp, \c^\sharp)
		&= \g \left( \g \frac{r^2 + a^2}{\mu(r)} (1 - f(r)^2) - \frac 2b \right) \frac{f(r)^2}{r^2 + a^2} \\
		&\qquad - \frac{\mu(r)}{b^2(r^2 + a^2)^2} + \g^4 \frac{\sin(\theta)^2}{c(\theta)}.
\end{split}
\end{equation}
\end{remark}

\begin{proof}[Proof of Proposition~\ref{prop: c smooth and timelike}]
The first two points are immediate.
We use~\eqref{eq: c causality} to check the causality of $\c$ on $\mathcal M$.
Note first that 
\[
	\frac{1 - f(r)^2}{\mu(r)}
\]
is continuous on $(r_e - 2\e_M, r_c + 2\e_M)$ if $\e_M > 0$ is small enough.
We may therefore choose $\g_0 > 0$ small enough so that
\[
	\g \frac{r^2 + a^2}{\mu(r)} \left( 1 - f(r)^2 \right) - \frac 2 b 
		< - \frac1 b
\]
for all $\g \in (0, \g_0)$ and $r \in (r_e - 2\e_M, r_c + 2\e_M)$.
Consequently, 
\begin{align*}
	&\g \left( \g \frac{r^2 + a^2}{\mu(r)} (1 - f(r)^2) - \frac 2b \right) \frac{f(r)^2}{r^2 + a^2} 
		- \frac{\mu(r)}{b^2(r^2 + a^2)^2} \\
			& \quad \leq - \g \left( \frac{f(r)^2}{b(r^2 + a^2)} + \frac{\mu(r)}{b^2(r^2 + a^2)^2} \right) \\
			& \quad < - \g^4 \frac{\sin(\theta)^2}{c(\theta)}
\end{align*}
for all $\g \in (0, \g_0)$, if $\g_0 \in (0, 1)$ is small enough, since
\[
	\frac{f(r)^2}{b(r^2 + a^2)} + \frac{\mu(r)}{b^2(r^2 + a^2)^2}
\]
is continuous and negative on $[r_e, r_c]$.
This proves the assertion.
\end{proof}

\begin{remark} \label{rmk: integral curves}
The analysis for low frequencies will be a propagation estimate along the integral curves of $\c^\sharp$.
Let $\gamma: I \to M$, be an inextendible integral curve of $\c^\sharp$ through a point $p_0 \in M$.
\begin{itemize}
	\item If $\pm r(p_0) > r_0$, then $\gamma$ transversally intersects the hypersurface
		\[
			\{r = r_c + \e_M \} \cap M
		\]
		or 
		\[
			\{r = r_e - \e_M \} \cap M,
		\]
		respectively, in finite future time.
	\item If $r(p_0) = r_0$, then 
		\[
			\lim_{\tau \to \infty}
				\gamma(\tau)
				=
			\begin{cases} 
				p_N, & \text{ if } \theta(p_0) = 0, \\
				p_S, & \text{ if } \theta(p_0) = (0, \pi].
			\end{cases}
		\]
	\item If $s(p_0) = 0$, then
	\[
		\lim_{\tau \to - \infty} \gamma(\tau) 
			= 
			\begin{cases} 
				p_N, & \text{ if } \theta(p_0) \in [0, \pi), \\
				p_S, & \text{ if } \theta(p_0) = \pi.
			\end{cases}
	\]
	\item If $s(p_0) > 0$, then $\gamma$ transversally intersects the hypersurface 
	\[
		\{s = 1\} \cap M
	\]
	in a finite past time.
\end{itemize}
In particular, the flow along $\c^\sharp$ has a \emph{saddle point} structure at $p_N$ and $p_S$.
Moreover, $p_N$ acts as a source and $p_S$ acts as a sink for the flow of $\c^\sharp$ restricted to 
\[
	\{s = 0\} \cap \{r = r_0\}.
\]
\end{remark}

\subsection{Low frequency propagation estimates} \label{sec: low frequency estimates}

The purpose of this section is to prove the following estimate.

\begin{prop}[The low frequency propagation estimate] \label{prop: main low frequency estimate}
There is a $\rho_0 > 0$ such that given any $s_0 > 0$ and $\ep \in (0, 1)$, there is an $e_0 > 0$ such that if $\rho \in [0, \rho_0]$ and $e \in (0, e_0)$, then the following holds:

If $K \subset \bT^* \M$ is a compact subset and $\Bop_1, \Bop_2, \Sop \in \Psi_{\be, h}^{- \infty, 0, 0}(\M)$ are such that
\begin{align*}
	o \cap \{ s \in [\ep s_0, s_0] \}
		&\subseteq \ell(\Bop_1), \\
	\WFb(\Bop_2)
		&\subseteq \{ s \in [0, s_0] \} \cap  K, \\
	K
		&\subseteq \ell(\Sop),
\end{align*}
then
\begin{equation} \label{eq: main low frequency propagation estimate}
	\norm{\Bop_2 u}_{H^{0, \rho}_{\be, h}}
		\lesssim_{k, e} \norm{\Bop_1 u}_{H^{0, \rho}_{\be, h}} + h^{-1} \norm{\Sop \hat \P_e u}_{H^{0, \rho}_{\be, h}} + h^k \norm{u}_{H^{0, \rho}_{\be, h}},
\end{equation}
for all $u \in \dot C_c^\infty$, $k \in \N_0$ and $h \in (0, h_0)$, for some $h_0 > 0$ depending on $e$.
\end{prop}

This statement should be compared with Proposition~\ref{prop: ellipticity finite frequencies}, where the ellipticity of $\tP_e$ was established at finite \emph{non-zero} frequencies.
Here we may include all finite frequencies, at the cost of the estimate being a propagation estimate.
We also have the adjoint estimate, where the propagation goes in the other direction.

\begin{remark}
As described in Section~\ref{sec: low frequency preparation}, it is convenient to work with $\hat \P_e$ in place of $\tP_e$.
However, since $\hat \P_e = -i (1 - i \tLop_e^{*\mB}) \tP_e$ and $(1 - i \tLop_e^{*\mB})$ is elliptic,~\eqref{eq: main low frequency propagation estimate} is equivalent to the same estimate for $\tP_e$.
\end{remark}

\begin{prop}[The low frequency adjoint propagation estimate] \label{prop: main low frequency adjoint estimate}
There is a $\rho_0 > 0$ such that given any $s_0 > 0$ and $\ep \in (0, 1)$, there is an $e_0 > 0$ such that if $\rho \in [0, \rho_0]$ and $e \in (0, e_0)$, then the following holds:

If $K \subset \bT^* \M$ is a compact subset and $\Bop_1, \Bop_2, \Sop \in \Psi_{\be, h}^{- \infty, 0, 0}(\M)$ are such that
\begin{align*}
	o \cap \left( \{ r \in [r_e - \ep_M, r_e - \ep \cdot \ep_M] \cup [r_c + \ep \cdot \ep_M, r_c + \ep_M] \} \right)
		&\subseteq \ell(\Bop_1), \\
	\WFb(\Bop_2)
		&\subseteq \{ s \in [0, s_0] \} \cap  K, \\
	K
		&\subseteq \ell(\Sop),
\end{align*}
then
\begin{equation} \label{eq: main low frequency adjoint propagation estimate}
	\norm{\Bop_2 u}_{H^{0, - \rho}_{\be, h}}
		\lesssim_{k, e} \norm{\Bop_1 u}_{H^{0, - \rho}_{\be, h}} + h^{-1} \norm{\Sop \hat \P_e^{*\mB} u}_{H^{0, - \rho}_{\be, h}} + h^k \norm{u}_{H^{0, - \rho}_{\be, h}},
\end{equation}
for all $u \in \dot C_c^\infty$, $k \in \N_0$ and $h \in (0, h_0)$, for some $h_0 > 0$ depending on $e$.
\end{prop}

In other words, we conclude control in the region $\{s \in [0, s_0]\}$ at all \emph{finite} frequencies.
Note that Proposition~\ref{prop: main low frequency adjoint estimate} is formulated on the dual spaces compared to Proposition~\ref{prop: main low frequency estimate}.
Recalling from Equation~\eqref{eq: tilde p e principal symbol} that the semiclassical principal symbol of $\tP_e$ is given by
\[
	\tilde \p_e(\xi)
		= G(\xi, \xi) \id + i \tl_e(\xi),
\]
Proposition~\ref{prop: tilde L e properties} implies that the adjoint operator $\tP_e^{*\mB}$ has principal symbol
\[
	\tilde \p_e(\xi)
		= G(\xi, \xi) \id - i \tl_e(\xi).
\]
Near zero frequency, the terms $i \tl_e$ and $- i \tl_e$ will be dominating the propagation for $\tP_e$ and $\tP_e^{*\mB}$, respectively.
Again by Proposition~\ref{prop: tilde L e properties}, in the limit $e \to 0$, this effectively (i.e., up to lower order terms in $e$) implies switching from propagation along $\c^\sharp$ to propagation along $-\c^\sharp$.
This is why the weights $\rho$ in~\eqref{eq: main low frequency propagation estimate} are replaced by $-\rho$ in~\eqref{eq: main low frequency adjoint propagation estimate}, but the method of proof is the same.
We will therefore only present the proof of Proposition~\ref{prop: main low frequency estimate} in detail and comment on the differences in proving Proposition~\ref{prop: main low frequency adjoint estimate}, see Remark~\ref{rmk: adjoint near critical points}, Remark~\ref{rmk: adjoint between critical points} and Remark~\ref{rmk: adjoint global low freq estimate}.

\subsubsection{The skew-adjoint part at the critical points}

The strategy to prove Proposition~\ref{prop: main low frequency estimate} (and similarly Proposition~\ref{prop: main low frequency adjoint estimate}) is to start near the critical points and successively establish~\eqref{eq: main low frequency propagation estimate} on larger and larger regions.
By Proposition~\ref{prop: c smooth and timelike}, $\c^\sharp$ vanishes precisely at the two points
\begin{align*}
	p_N
		= & \ \{s = 0\} \cap \{r = r_0 \} \cap \{\theta = 0\}, \\
	p_S
		= & \ \{s = 0\} \cap \{ r = r_0 \} \cap \{\theta = \pi\}
\end{align*}
in $\M$.
Recalling~\eqref{eq: commutator equality}, the key term in the propagation estimates near $p_N$ and $p_S$ at zero frequencies is the skew-adjoint part of the operator $\tLop_e$.

\begin{prop} \label{prop: KdS skew adjoint term}
Fix any $e_0 > 0$. 
There is a $\g_0 > 0$ such that if $\g \in (0, \g_0)$ in the definition of $\c^\sharp$ in~\eqref{eq: c sharp in new coordinates}, then there is a $c > 0$ such that
\[
	\frac{\tLop_e - \tLop_e^{*\mB}}{2 i h}\Big|_{p_{N/S}}
		\leq - c \ \id
\]
for all $e \in (0, e_0)$.
\end{prop}

For the proof, we need to compute the relevant terms to insert into the expression in Proposition~\ref{prop: tilde L e properties}.
Since $p_N$ is located at the north pole and $p_S$ at the south pole of a $2$-sphere, it is convenient to work in coordinates which are smooth near those points.
We therefore use the coordinates
\begin{align*}
	x 
		:= & \ \sin(\theta) \cos(\psi), \\
	y 
		:= & \ \sin(\theta) \sin(\psi),
\end{align*}
in place of $\psi, \theta$, where $x$ and $y$ are defined \emph{either} near $p_N$, where $\theta \approx 0$, or near $p_S$, where $\theta \approx \pi$.

\begin{lemma} \label{le: linearization}
In the coordinate system $(s, r, x, y)$ near $p_N$ or $p_S$,
\begin{align*}
	\c^\sharp
		=& \ {-}s \d_{s} + a \frac{r^2 - r_0^2}{(r^2 + a^2)(r_0^2 + a^2)} \left( - y \d_x + x \d_y \right) \\
		&\qquad + \g f(r) \d_r \pm \g^2 \sqrt{1 - (x^2 + y^2)}\left( x \d_x + y \d_y \right) \\
		=& \ {-}s \d_s + \g f'(r_0) (r - r_0) \d_r \pm \g^2 \left( x \d_x + y \d_y \right) + V,
\end{align*}
where $V$ is a smooth b-vector field on $M$ which is quadratically vanishing at $p_N$ or $p_S$, respectively.
\end{lemma}

\begin{proof}
Note that
\begin{align*}
	\d_\psi
		&= - y \d_x + x \d_y, \\
	\sin(\theta) \d_\theta
		&= \pm \sqrt{1 - (x^2 + y^2)}\left( x \d_x + y \d_y \right),
\end{align*}
near $p_N$ and $p_S$, respectively.
Substituting these into~\eqref{eq: c sharp in new coordinates} provides the asserted expressions.
\end{proof}

We may now compute the relevant terms to insert in Proposition~\ref{prop: tilde L e properties}:

\begin{lemma} \label{le: k conformal computation}
For any $\g > 0$, we have 
\begin{align*}
	\de^*_g(\c)|_{p_{N/S}}
		&= \g f'(r_0) \md r \otimes \d_r \pm \g^2 \left( \md x \otimes \d_x + \md y \otimes \d_y \right)|_{p_{N/S}}, \\
	\de_g(\c) |_{p_{N/S}}
		&= - \g f'(r_0) \mp \g^2, \\
	\md G(\c, \c) |_{p_{N/S}}
		&= 0, \\
	\n_{\c^\sharp} \c |_{p_{N/S}}
		&= 0.
\end{align*}
\end{lemma}
\begin{proof}
Noting that $\md x^2 + \md y^2 = \cos^2(\theta) \md \theta^2 + \sin^2(\theta) \md \psi^2$, it follows that
\[
	g |_{p_{N/S}}
		= \frac{r_0^2 + a^2}{\mu(r_0)} \md r^2 - \frac1{b^2} \frac{\mu(r_0)}{r_0^2 + a^2} \frac{\md s^2}{s^2}  + \frac{r_0^2 + a^2}b \left( \md x^2 + \md y^2 \right).
\]
By Lemma~\ref{le: linearization},
\begin{align*}
	\c^\sharp
		=& - s \d_{s} + a \frac{r^2 - r_0^2}{(r^2 + a^2)(r_0^2 + a^2)} \left( - y \d_x + x \d_y \right) \\
		&\qquad + \g f(r) \d_r \pm \g^2 \sqrt{1 - (x^2 + y^2)}\left( x \d_x + y \d_y \right).
\end{align*}
near $p_N$ and $p_S$, respectively.
Since $s\d_s$ is a Killing vector field, we get
\begin{align*}
	\L_{\c^\sharp}g|_{p_{N/S}}
		&= 2 \g\,\md f \otimes g(\d_r, \cdot)|_{p_{N/S}}  \pm 2 \g^2 \left( \md x \otimes g(\d_x, \cdot) + \md y \otimes g(\d_y, \cdot) \right)|_{p_{N/S}} \\
		&= 2 \g f'(r_0) \frac{r_0^2 + a^2}{\mu(r_0)} \md r^2|_{p_{N/S}}  \pm 2 \g^2 \frac{r_0^2 + a^2}b \left( \md x^2 + \md y^2 \right) |_{p_{N/S}}.
\end{align*}
Hence,
\begin{align*}
	\de^*_g(\c)|_{p_{N/S}}
		&= \frac12 \L_{\c^\sharp}g^\sharp|_{p_{N/S}} \\
		&= \g f'(r_0) \md r \otimes \d_r \pm \g^2 \left( \md x \otimes \d_x + \md y \otimes \d_y \right)|_{p_{N/S}},
\end{align*}
as claimed.
A direct consequence is that
\[
	\de_g(\c)|_{p_{N/S}}
		= - \tr \left( \de^*_g(\c) \right)|_{p_{N/S}}
		= - \g f'(r_0) \mp \g^2,
\]
as claimed.
By Remark~\ref{rmk: norm of c sharp} and since $f(r_0) = 0$ and $\cos^2(\theta) = 1 - (x^2 + y^2)$, we get
\begin{align*}
	\md G(\c, \c)|_{p_{N/S}}
		&= - \frac1{b^2} \md \left( \frac{\mu(r)(r^2 + a^2 \left( 1- (x^2 + y^2) \right)}{(r^2 + a^2)^2} \right)\Big|_{p_{N/S}} \\
		&= - \frac1{b^2} \frac{\md}{\md r} \left( \frac{\mu(r)}{r^2 + a^2} \right) \md r \Big|_{p_{N/S}} \\
		&= 0,
\end{align*}
by our choice that $r_0$ satisfies~\eqref{eq: choice of r 0}.
Finally, for any vector field $X$, we get
\begin{align*}
	\n_{\c^\sharp}\c(X)|_{p_{N/S}}
		&= g(\n_{\c^\sharp} \c^\sharp, X)|_{p_{N/S}} \\
		&= \L_{\c^\sharp}g(\c^\sharp, X)|_{p_{N/S}} - \frac12 X g(\c^\sharp, \c^\sharp)|_{p_{N/S}} \\
		&= - \frac12 \md G(\c, \c)(X)|_{p_{N/S}} \\
		&= 0,
\end{align*}
finishing the proof.
\end{proof}

\begin{proof}[Proof of Proposition~\ref{prop: KdS skew adjoint term}]
Lemma~\ref{le: k conformal computation} implies that~\eqref{eq: assumptions critical point} is satisfied at $p_N$ and $p_S$.
Hence Proposition~\ref{prop: tilde L e properties}, part~\eqref{item: tilde L e skew-adjoint part}, implies that
\[
	\frac{\tLop_e - \tLop_e^{*\mB}}{2 i h}|_{p_{N/S}}
		=
			\begin{pmatrix}
				(1 + e) \de_g(\c)|_p & 0 \\
				0 & \frac12 \de_g(\c)|_p - \de_g^*(\c)^\sharp|_p
			\end{pmatrix}\bigg|_{p_{N/S}}.
\]
with respect to the splitting $\bT^* M = \R \c \oplus \c^\perp$.
By Lemma~\ref{le: k conformal computation}, we conclude that this depends continuously on $\g$, up to $\g = 0$.
For this reason, we may initially ignore $\g^2$-terms, so
\begin{align*}
	\frac{\tLop_e - \tLop_e^{*\mB}}{2 i h}|_{p_{N/S}}
		&= \begin{pmatrix}
			- \g ( 1 + e) f'(r_0) & 0 \\
			0 & - \g f'(r_0) \left( \frac12 \id + \md r \otimes \d_r \right)
		\end{pmatrix} + \O(\g^2) \\
		&\leq - \begin{pmatrix}
			\g f'(r_0) & 0 \\
			0 & \g f'(r_0) \left( \frac12 \id + \md r \otimes \d_r \right)
		\end{pmatrix} + \O(\g^2)
\end{align*}
for any $e \in [0, e_0]$.
By continuity, uniform negativity holds true for small enough $\g > 0$, proving the statement.
\end{proof}

\subsubsection{Near the critical points}

We now want to use Proposition~\ref{prop: KdS skew adjoint term} to prove propagation estimates for $\hat \P_e$ near the critical points $p_N$ and $p_S$.
We work in the coordinate system $(s, r, x, y)$, introduced in the previous subsection, near either $p_N$ or $p_S$.
Let us denote points $z \in M$ near $p_N$ or $p_S$ by the coordinates
\[
	z = (s, r, x, y),
\]
and b-one-forms $\xi \in \bT^*_z M$ by
\[
	\xi = - \xi_s \frac{\md s}s + \xi_r \md r + \xi_x \md x + \xi_y \md y.
\]
We also introduce the notation
\begin{align*}
	\Omega_\de^{p_N}
		:= & \ \{s, \abs{v} \in [0, \de] \} \cap \{\abs{\xi} \in [0, \de]\}, \\
	\Omega_\de^{p_S}
		:= & \ \{\abs{w}, \abs{r - r_0} \in [0, \de] \} \cap \{ \abs{\xi} \in [0, \de]\}, 
\end{align*}
where in this section
\begin{align*}
	\abs{v}
		:= & \ \sqrt{(r - r_0)^2 + x^2 + y^2}, \\
	\abs{w}
		:= & \ \sqrt{s^2 + x^2 + y^2}, \\
	\abs{\xi}
		:= & \ \sqrt{\xi_s^2 + \xi_r^2 + \xi_x^2 + \xi_y^2},
\end{align*}
and where again $x, y$ are defined either near $p_N$ or $p_S$.
Note that indeed 
\[
	p_{N/S} \in \Omega_\de^{p_{N/S}}.
\]

\begin{prop} \label{prop: critical point estimates}
There are $\rho_0, \de_0 > 0$ such that if $\rho \in [0, \rho_0]$, $\de \in (0, \de_0)$, $\ep \in (0, 1)$ and $e \in \left(0, \de^3\right)$, then the following holds:

If $\Bop_1, \Bop_2, \Sop \in \Psi_{\be, h}^{- \infty, 0, 0}(M)$ are such that
\begin{align*}
	\{ s \in [\de, (1 + \ep)\de] \} \cap \Omega^{p_N}_{(1 + \ep)\de}
		&\subseteq \ell(\Bop_1), \\
	\WFb(\Bop_2)
		&\subseteq \Omega^{p_N}_{\de}, \\
	\Omega^{p_N}_{(1 + \ep)\de}
		&\subseteq \ell(\Sop),
\end{align*}
then~\eqref{eq: main low frequency propagation estimate} holds for all $u \in \dot C_c^\infty$, $k \in \N_0$ and $h \in (0, h_0)$, for some $h_0 > 0$ depending on $e$.

If instead $\Bop_1, \Bop_2, \Sop \in \Psi_{\be, h}^{- \infty, 0, 0}(M)$ are such that
\begin{align*}
	\{ \abs { w } \in [\de, (1 + \ep)\de] \} \cap \Omega^{p_S}_{(1 + \ep)\de}
		&\subseteq \ell(\Bop_1), \\
	\WFb(\Bop_2)
		&\subseteq \Omega^{p_S}_{\de}, \\
	\Omega^{p_S}_{(1 + \ep)\de}
		&\subseteq \ell(\Sop),
\end{align*}
then~\eqref{eq: main low frequency propagation estimate} holds for all $u \in \dot C_c^\infty$, $k \in \N_0$ and $h \in (0, h_0)$, for some $h_0 > 0$ depending on $e$.
\end{prop}

\begin{remark} \label{rmk: de k would work}
In the proposition, we choose the condition $e \in (0, \de^3)$ since it is convenient for the estimates in the proof.
We could equally well have chosen $e \in (0, \de^k)$ for any $k \geq 3$.
The point is that when applying this statement, we choose a $\de > 0$ and are then guaranteed to be able to choose $e$ in some interval of positive numbers that depends explicitly on $\de$.
\end{remark}

The proof of Proposition~\ref{prop: critical point estimates} is a propagation estimate in the direction of $\c^\sharp$, up to relatively small errors.
The error terms will be small when $e$ is small and the frequency $\xi$ is small.
We begin by constructing the principal symbol for our two commutators, one localized near $p_N$ and the other one localized near $p_S$.
Given any $\de > 0$ and $\ep \in (0, 1)$, fix a $\chi \in C_c^\infty(\R)$ satisfying
\begin{itemize}
	\item $\chi \geq 0$,
	\item $\supp(\chi) \subseteq [-(1 + \ep)\de, (1 + \ep)\de]$,
	\item $\chi(s) \equiv 1$ for $\abs s \leq \de$, 
	\item $\left(-s\chi'(s)\chi(s)\right)^{\frac12} \in C^\infty_c(\R)$.
\end{itemize}
In particular, $s \chi'(s) \leq 0$ for all $s \in \R$.
For any fixed $\rho \geq 0$, define the symbols $\as_N, \as_S \in \S^{ -\infty, \rho, 0}_{\be, h}(M; \bT^* M)$ by
\begin{align*}
	\as_N (z, \xi)
		:= & \ s^{-\rho} \chi (s) \chi (\abs v) \chi(\abs{\xi}) \cdot \id, \\
	\as_S (z, \xi)
		:= & \ s^{-\rho} \chi (\abs w) \chi (r - r_0) \chi(\abs{\xi}) \cdot \id
\end{align*}
for $s > 0$.
The growing factor $s^{-\rho}$ reflects that we want to prove estimates on \emph{decaying} function spaces, as $s \to 0$.
Note that 
\[
	\supp \left( \as_{N /S} \right)
		\subseteq \Omega_{(1 + \ep)\de}^{p_N/p_S},
\]
respectively.
Indeed, $\as_{N/S}$ is localizing around either $p_N$ or $p_S$, where $x = y = 0$, as well as near zero frequency.
Let now 
\[
	\Aop_{N/S}
		:= \Re \left( \Op_h(\as_{N/S}) \right).
\]
Since $\as_{N/S}$ is real-valued, the semiclassical principal symbol of $\Aop_{N/S}$ is given by $\as_{N/S}$.
The following lemma is the symbolic computation, from which Proposition~\ref{prop: critical point estimates} will be a simple consequence.

\begin{lemma} \label{le: prop at critical points}
There are $\rho_0, \de_0 > 0$ such that if $\rho \in (0, \rho_0]$, $\de \in (0, \de_0)$, $\e \in (0, 1)$ and $e \in \left(0, \de^3\right)$, then the following holds (recall that $\as_N$ and $\as_S$ depend on $\rho)$:

There are (endomorphism-valued) symbols $\bs_N, \f_N, \k_N \in \S_{\be, h}^{-\infty, \rho, 0}$ such that
\begin{align*}
	\esssupp(\f_N)
		&\subseteq \{ s \in [\de, (1 + \ep)\de] \} \cap \Omega_{(1 + \ep)\de}^{p_N}, \\
	\esssupp(\k_N)
		&\subseteq \{ \abs{\xi} \in [\de, (1 + \ep)\de] \} \cap \Omega_{(1 + \ep)\de}^{p_N},
\end{align*}
and
\begin{equation} \label{eq: north pole symbol computation}
	\s \left( \frac i h \left[\tLop_e + \Jop_e, \Aop_N^2\right] + \frac i h \left( \tLop_e^{*\mB} - \tLop_e \right) \Aop_N^2 \right)
		= - \rho_0 \as_N^2 - \bs_N^2 + \as_N \f_N + \as_N \k_N.
\end{equation}

There are (endomorphism-valued) symbols $\bs_S, \f_S, \k_S \in \S_{\be, h}^{-\infty, \rho, 0}$ such that
\begin{align*}
	\esssupp(\f_S)
		&\subseteq \{ \abs w \in [\de, (1 + \ep)\de]\} \cap \Omega_{(1 + \ep)\de}^{p_S}, \\
	\esssupp(\k_S)
		&\subseteq \{ \abs{\xi} \in [\de, (1 + \ep)\de] \} \cap \Omega_{(1 + \ep)\de}^{p_S},
\end{align*}
and
\begin{equation} \label{eq: south pole symbol computation}
	\s \left( \frac i h \left[ \tLop_e + \Jop_e, A_\de^2 \right] + \frac i h \left( \tLop_e^{*\mB} - \tLop_e \right) \Aop_S^2 \right)
		= - \rho_0 \as_S^2 - \bs_S^2 + \as_S \f_S + \as_S \k_S.
\end{equation}
\end{lemma}

The following will simplify the notation:

\begin{remark}
Let $\md_z$ an $\md_\xi$ denote
\begin{align*}
	\md_z f
		:= \left( - s\d_sf \right) \d_{\xi_s} + \left( \d_rf \right) \d_{\xi_r} + \left( \d_x f \right) \d_{\xi_x} + \left(\d_y f \right) \d_{\xi_y}, \\
	\md_\xi f
		:= \left( \d_{\xi_s} f \right) (-s \d_s) + \left( \d_{\xi_r} f \right) \d_r + \left( \d_{\xi_x} f \right) \d_x + \left( \d_{\xi_y} f \right) \d_y,
\end{align*}
for any $f: \bT^* M \to \C$.
If $\p$ and $\q$ are the (endomorphism-valued) semiclassical b-principal symbols of $\P$ and $\mathrm Q$, respectively, then the semiclassical b-principal symbol of $[\P, \mathrm Q]$ is given by the b-Poisson bracket between the principal symbols, i.e., by
\begin{align*}
	\{\p, \q\}
		:=& \ \md_\xi \p \cdot \md_z \q - \md_z \p \cdot \md_\xi \q \\
		=& \ \d_{\xi_s} \p\,(- s\d_s) \q + \d_{\xi_r}\p\,\d_r \q + \d_{\xi_x}\p\,\d_x \q + \d_{\xi_y}\p\,\d_y \q \\
		& \ - (- s\d_s) \p\,\d_{\xi_s} \q - \d_r \p\,\d_{\xi_r}\q - \d_x \p\,\d_{\xi_x}\q - \d_y \p\,\d_{\xi_y}\q,
\end{align*}
where here $\cdot$ denotes applying the b-symplectic form 
\[
	- \frac{\md s}s \wedge \md \xi_s + \md r \wedge \md \xi_r + \md x \wedge \md \xi_x + \md y \wedge \md \xi_y.
\]
We will sometimes identify $\md_z f$ and $\md_\xi f$ with the corresponding b-one-form, with respect to the b-symplectic form, without explicitly mentioning it.
\end{remark}

\begin{proof}[Proof of Lemma~\ref{le: prop at critical points}]
The principal symbols of $\tLop_e$ and $\Jop_e$ are given by
\begin{align*}
	\tl_e(\xi)
		= & \ G(\c, \xi) \left( \id + \frac{\c \otimes \c^\sharp}{G(\c, \c)} \right) + \sqrt{2e} \cdot \r_e(\xi), \\
	\j_e(\xi)
		= & G(\xi, \xi) \tl_e(\xi),
\end{align*}
where $\r_e$ is the (endomorphism-valued) symbol given by
\begin{align*}
	\r_e(\xi) \o
		= & \ G(\xi, \o) \c - G(\o, \c) \xi - \sqrt{2 e} \frac{G(\o, \c) G(\xi, \c)}{G(\c, \c)} \c,
\end{align*}
for any $\xi, \o \in \bT^* M$.
Since $\tl_e$ is homogeneous of degree $1$ in $\xi$, note that
\[
	\md_\xi \tl_e \cdot \md_z \as_N^2
		= \tl_e \left(\md_z \as_N^2 \right).
\]
Using this and defining
\[
	\k_N
		:= - 2 \md_z \left( \tl_e + \j_e \right) \cdot \md_\xi a_N,
\]
we compute
\begin{align}
	\s \left( \frac i h [\tLop_e + \Jop_e, \Aop_N^2] \right) 
		&= \left\{ \tl_e + \j_e, \as_N^2 \right\} \nonumber \\
		&= \md_\xi \left( \tl_e + \j_e \right) \cdot \md_z \as_N^2 - \md_z \left( \tl_e + \j_e \right) \cdot \md_\xi \as_N^2 \nonumber \\
		&= \ \tl_e(\md_z \as_N^2) + \md_\xi \j_e \cdot \md_z \as_N^2 + \as_N \k_N \nonumber \\
		&= \ \left( \left( \id + \frac{\c \otimes \c^\sharp}{G(\c, \c)} \right)\c^\sharp + \sqrt{2e} \cdot \r_e + \md_\xi \j_e \right)\left(\md_z \as_N^2 \right) + \as_N \k_N.  \label{eq: commutator term critical point}
\end{align}
Note that indeed $\supp(\k_N) \subseteq \{ \abs{\xi} \in [\de, (1 + \ep)\de] \} \cap \Omega_{(1 + \ep)\de}^{p_N}$ for any $e > 0$, since $\supp(\chi') \subseteq [\de, (1 + \ep)\de]$.
We note that
\begin{align*}
	\md_z \as_N
		&= \left( \rho \as_N - s^{-\rho} s \chi'(s) \chi(\abs{v}) \chi(\abs{\xi}) \right) \left( - \frac{\md s}s \right) \\*
		&\qquad + s^{-\rho} \frac1{\abs v}\chi(s) \chi'(\abs{v}) \chi(\abs{\xi}) \left( (r - r_0) \md r + x \md x + y \md y \right).
\end{align*}
Moreover, using Lemma~\ref{le: linearization}, we note that
\begin{align*}
	\c^\sharp(\md_z \as_N)
		&= \rho \as_N - s^{-\rho} s \chi'(s) \chi(\abs{v}) \chi(\abs{\xi}) \\*
		&\qquad + \Bigg( \left( \a f'(r_0) + \O(\abs{r - r_0}) \right) (r - r_0)^2 + \left( \b + \O(\abs{x}) \right) x^2 \\*
		&\qquad \qquad  + \left( \b + \O(\abs{y}) \right) y^2 \Bigg) \frac1{\abs v} s^{-\rho} \chi(s) \chi'(\abs v) \chi(\abs{\xi}),
\end{align*}
where the $\O$-terms are smooth functions.
We thus formally get
\begin{align*}
	\left( \left( \id + \frac{\c \otimes \c^\sharp}{G(\c, \c)} \right)\c^\sharp + \sqrt{2e} \cdot \r_e + \md_\xi \j_e \right)\left(\md_z \as_N^2\right)
		= \rho \as_N^2 \t_N + \as_N \f_N - \check \be_N^2
\end{align*}
where
\begin{align*}
	\t_N
		&= 2 \left( \id + \frac{\c \otimes \c^\sharp}{G(\c, \c)} + \left( \sqrt{2e} \cdot \r_e + \md_\xi \j_e \right) \left( - \frac{\md s}s \right) \cdot \id \right), \\
	\f_N
		&= - s^{-\rho} s \chi'(s) \chi(\abs{v}) \chi(\abs{\xi})\t_N,
\end{align*}
and where 
\begin{align*}
	\check \be_N
		&= s^{-\rho} \frac1{\abs v}\chi(s)\chi(\abs{\xi}) \left(- 2 \abs v \chi'(\abs v) \chi(\abs v) \right)^{\frac12} \\
		&\qquad \cdot \Bigg( \Bigg( \left( \a f'(r_0) + \O(\abs{r - r_0}) \right) (r - r_0)^2 + \left( \b + \O(\abs{x}) \right) x^2 \\*
		&\qquad \qquad + \left( \b + \O(\abs{y}) \right) y^2 \Bigg) \left( \id + \frac{\c \otimes \c^\sharp}{G(\c, \c)} \right) \\
		&\qquad \qquad + \left( \sqrt{2e} \cdot \r_e + \md_\xi j_e \right) \left( (r - r_0) \md r + x \md x + y \md y \right) \cdot \id \Bigg)^{\frac12}.
\end{align*}
Since $\t_N$ is a smooth endomorphism field on $\bT^* M$, we conclude that $\f_N \in \S^{-\infty, -\infty, 0}_{\be, h}$, with
\[
	\esssupp(\f_N)
		\subseteq \{ s \in [\de, (1 + \ep)\de] \} \cap \Omega_{(1 + \ep)\de}^{p_N},
\]
as desired.
However, we still need to choose $\de_0$ small enough to ensure that $\check \be_N$ is a well-defined symbol.
Since $v \neq 0$ on the support of $\chi'(\abs{v})$, the factor
\[
	- 2 \abs v \chi'(\abs v) \chi(\abs v)
\]
is a positive smooth function.
Moreover, by assumption, $(-s\chi'(s)\chi(s)))^{\frac12}$ is smooth everywhere, hence 
\[
	\left(- 2 \abs v \chi'(\abs v) \chi(\abs v) \right)^{\frac12}
\]
is smooth everywhere.
Next, there is a $C > 0$ such that 
\[
	\abs{\r_e(\xi)} \leq C \abs{\xi} \id,
\]
as quadratic forms, and, since $\j_e(\xi)$ is cubically vanishing at the zero section, 
\[
	\abs{\md_\xi \j_e(\xi)}
		\leq C \abs{\xi}^2 \id,
\]
for all $e \leq 1$.
Since $\supp\left(\check \be_N\right) \subset \Omega^N_{(1 + \ep)\de}$ and $e \in \left(0, \de^3 \right)$, it follows that
\begin{align}
	\left( \sqrt{2e} \cdot \r_e + \md_\xi j_e \right) \left( (r - r_0) \md r + x \md x + y \md y \right) \cdot \id
		&\leq C \de \left( \sqrt{2e} + \de^2 \right) \cdot \id \nonumber \\
		&\leq C \de^{5/2} \cdot \id \label{eq: rest terms}
\end{align} 
in $\supp\left(\check \be_N\right)$.
On the other hand,
\begin{align*}
	\Bigg( 
		&\left( \a f'(r_0) + \O(\abs{r - r_0}) \right) (r - r_0)^2 + \left( \b + \O(\abs{x}) \right) x^2 \\
		&+ \left( \b + \O(\abs{y}) \right) y^2 \Bigg) \left( \id + \frac{\c \otimes \c^\sharp}{G(\c, \c)} \right)  \\
		&\geq C \de^2 \cdot \id,
\end{align*}
for sufficiently small $\de_0$.
Combining this with~\eqref{eq: rest terms} shows that the expression inside the second square root in the definition of $\check \be_N$ is positive, and thereby that the square root is well-defined, for small enough $\de_0 > 0$.
In conclusion, we have shown that
\begin{equation} \label{eq: the critical point commutator}
	\s \left( \frac i h [\tLop_e + \Jop_e, \Aop_N^2] \right)
		= \rho \as_N^2 \t_N + \as_N \f_N - \check \be_N^2 + \as_N \k_N,
\end{equation}
where $\t_N$ is a smooth endomorphism field, $\f_N$ and $\k_N$ are as in the statement of the lemma, and $\check \be_N \in \S^{-\infty, \rho, 0}_{\be, h}$. 

The weight $\rho$ in~\eqref{eq: the critical point commutator} will be compensated for by the skew-adjoint term.
By Proposition~\ref{prop: KdS skew adjoint term} and by continuity, there is a $\rho_0 > 0$ such that if $\de_0 > 0$ is small enough, then
\[
	\frac ih \left( \tLop_e^{*\mB} - \tLop_e \right)
		= \frac 1{ih} \left( \tLop_e - \tLop_e^{*\mB} \right)
		\leq - \rho_0 \left( \t_N + \id \right)
\]
for all $e \in (0, \de^3)$.
For all $\rho \leq \rho_0$, it follows that
\[
	\bs_N
		:= \left( \check \bs_N^2 + \left( - \frac ih \left( \tLop_e^{*\mB} - \tLop_e \right) - \rho \t_N - \rho_0 \id \right) \as_N^2 \right)^{\frac12}
\]
is a well-defined symbol in $\S_{\be, h}^{-\infty, \rho, 0}$.
Together with~\eqref{eq: the critical point commutator}, this proves~\eqref{eq: north pole symbol computation}.

Equation~\eqref{eq: south pole symbol computation} follows similarly, with the key differences being described by
\begin{align*}
	\md_z \as_S
		&= \rho \as_S \left( - \frac{\md s}s \right) + \frac1{\abs w}s^{-\rho} \chi'(\abs{w}) \chi(r-r_0) \chi(\abs{\xi}) \left( - s^2 \left( - \frac{\md s}s \right) + x \md x + y \md y \right)  \\*
		&\qquad + s^{-\rho} \chi(\abs w) \chi'(r-r_0) \chi(\abs{\xi}) \md r,
\end{align*}
giving
\begin{align*}
	\c^\sharp(\md_z \as_S)
		&= \rho \as_S - \left( s^2 + \left( \b + \O(\abs{x}) \right) x^2 + \left( \b + \O(\abs{y}) \right) y^2 \right) \\*
		&\qquad \qquad \cdot \frac1{\abs w}s^{-\rho} \chi'(\abs{w}) \chi(r-r_0) \chi(\abs{\xi}) \\
		&\qquad + \left( \a f'(r_0) + \O(\abs{r - r_0}) \right) (r - r_0) s^{-\rho} \chi(\abs{w}) \chi'(r-r_0) \chi(\abs{\xi}),
\end{align*}
and thereby
\begin{align*}
	\left( \left( \id + \frac{\c \otimes \c^\sharp}{G(\c, \c)} \right)\c^\sharp + \sqrt{2e} \cdot \r_e + \md_\xi \j_e \right)\left(\md_z \as_S^2\right)
		= \rho \as_S^2 \t_S + \as_S \f_S - \check \be_S^2
\end{align*}
where
\begin{align*}
	\t_S
		&= 2\left( \id + \frac{\c \otimes \c^\sharp}{G(\c, \c)} + \left( \sqrt{2e} \cdot \r_e + \md_\xi \j_e \right) \left( - \frac{\md s}s \right) \cdot \id \right), \\
	\f_S
		&= \frac1{\abs w}s^{-\rho} \chi'(\abs{w}) \chi(r-r_0) \chi(\abs{\xi}) \\
		&\qquad \cdot \Bigg( - \left( s^2 + \left( \b + \O(\abs{x}) \right) x^2 + \left( \b + \O(\abs{y}) \right) y^2 \right) \left( \id + \frac{\c \otimes \c^\sharp}{G(\c, \c)} \right) \\
		&\qquad \qquad + \left( \sqrt{2e} \cdot \r_e + \md_\xi j_e \right)\left( - s^2 \left( - \frac{\md s}s \right) + x \md x + y \md y \right) \cdot \id \Bigg),
\end{align*}
and where 
\begin{align*}
	\check \be_S
		&= s^{-\rho} \chi(\abs{w}) \chi(\abs{\xi}) \left(- 2 (r - r_0)\chi'(r-r_0)\chi(r-r_0) \right)^{\frac12} \\
		&\qquad \cdot \Bigg( \left( \left( \a f'(r_0) + \O(\abs{r - r_0}) \right) \right) \left( \id + \frac{\c \otimes \c^\sharp}{G(\c, \c)} \right) \\
		&\qquad \qquad + \left( \sqrt{2e} \cdot \r_e + \md_\xi j_e \right) \left(\md r \right) \cdot \id \Bigg)^{\frac12}.
\end{align*}
Similar to $\check \be_S$, one checks that $\check \be_S \in \S^{-\infty, \rho, 0}_{\be, h}$ is well-defined, if $\de_0$ is small enough.
The rest of the proof is analogous for $p_S$ as for $p_N$.
\end{proof}

\begin{proof}[Proof of Proposition~\ref{prop: critical point estimates}]
Quantizing the symbol equality~\eqref{eq: north pole symbol computation}, we get
\begin{align*}
	\frac i h 
		&\left[\tLop_e + \Jop_e, \Aop_N^2\right] + \frac i h \left( \tLop_e^{*\mB} - \tLop_e \right) \Aop_N^2 \\
		&= - \rho_0 \Aop_N^2 - \Bop_N^* \Bop_N + \Aop_N \Fop_N + \Aop_N \Kop_N + h \Rop_N,
\end{align*}
where $\Bop_N, \Fop_N, \Kop_N \in \Psi^{-\infty, \rho, 0}_{\be, h}$ and $\Rop_N \in \Psi^{-\infty, 2 \rho, 0}_{\be, h}$ and
\begin{align*}
	\WFb(\Fop_N)
		&\subseteq \{ s \in [\de, (1 + \ep)\de] \} \cap \Omega_{(1 + \ep)\de}^{p_N}, \\
	\WFb(\Kop_N)
		&\subseteq \{ \abs{\xi} \in [\de, (1 + \ep)\de] \} \cap \Omega_{(1 + \ep)\de}^{p_N}, \\
	\WFb(\Rop_N)
		&\subseteq \Omega_{(1 + \ep)\de}^{p_N}.
\end{align*}
It follows that
\begin{align}
	&\ldr{\frac i {2h} \left( \left[\tLop_e + \Jop_e, \Aop_N^2\right] + \left( \tLop_e^{*\mB} - \tLop_e \right) \Aop_N^2\right)u, u} \nonumber \\
		&\quad = - \rho_0 \norm{\Aop_N u}_{L^2}^2 - \norm{\Bop_N u}_{L^2}^2 + \ldr{\Fop_N u, \Aop_N u} + \ldr{\Kop_N u, \Aop_N u} + h \ldr{\Rop_N u, u} \nonumber \\
		&\quad \leq - \rho_0 \norm{\Aop_N u}_{L^2}^2 + \left( \norm{\Fop_N u}_{L^2} + \norm{\Kop_N u}_{L^2} \right) \norm{\Aop_N u}_{L^2} + h \ldr{\Rop_N u, u}. \label{eq: low order commutator}
\end{align}
Since $\Sop \hat \P_e$ is elliptic on $\WFb(\Kop_N)$, by Proposition~\ref{prop: ellipticity finite frequencies}, and since $\Bop_1$ is elliptic on $\WFb(\Fop_N)$, by assumption, microlocal elliptic regularity theory implies that
\begin{align*}
	\norm{\Kop_N u}_{L^2}
		&\lesssim_{k, e} \norm{\Sop \hat \P_e u}_{H^{0, \rho}_{\be, h}} + h^k \norm{u}_{H^{0, \rho}_{\be, h}}, \\
	\norm{\Fop_N u}_{L^2}
		&\lesssim_{k, e} \norm{\Bop_1 u}_{H^{0, \rho}_{\be, h}} + h^k \norm{u}_{H^{0, \rho}_{\be, h}},
\end{align*}
for all $k \in \N$.
Finally, we estimate
\[
	\ldr{\Rop_N u, u}_{L^2}
		= \ldr{s^{\rho} \Rop_N u, s^{-\rho}u}_{L^2}
		\lesssim_{k,e} \norm{\Sop u}_{H^{0, \rho}_{\be, h}}^2 + h^{2k}\norm{u}_{H^{0, \rho}_{\be, h}}^2,
\]
for all $k \in \N$.
Combining~\eqref{eq: low order commutator} with~\eqref{eq: commutator equality} and applying Lemma~\ref{le: the right-hand-side low}, we conclude the estimate
\begin{align*}
	\rho_0 \norm{\Aop_N u}_{L^2}^2
		&\lesssim_{k, e} \left( \norm{\Bop_1 u}_{H^{0, \rho}_{\be, h}} + h^{-1}\norm{\Sop \hat \P_e u}_{H^{0, \rho}_{\be, h}} + h^k \norm{u}_{H^{0, \rho}_{\be, h}} \right) \norm{\Aop_N u}_{L^2} \\
		&\quad \qquad + \left(\de + h^{1/2} \right) \norm{\Aop_N u}_{L^2}^2 + h \norm{\Sop u}^2_{H^{0, \rho}_{\be, h}} + h^k \norm{u}_{H^{0, \rho}_{\be, h}}^2,
\end{align*}
which implies the estimate
\begin{align*}
	\norm{\Aop_N u}_{L^2}
		&\lesssim_{k, e} \norm{\Bop_1 u}_{H^{0, \rho}_{\be, h}} + h^{-1}\norm{\Sop \hat \P_e u}_{H^{0, \rho}_{\be, h}} + h^{1/2} \norm{\Sop u}_{H^{0, \rho}_{\be, h}} + h^k \norm{u}_{H^{0, \rho}_{\be, h}},
\end{align*}
if $\de_0$ and $h_0$ are small enough.
Since $\Aop_N$ is elliptic on $\WFb(\Bop_2)$,
\[
	\norm{\Bop_2 u}_{H^{0, \rho}_{\be, h}}
		\lesssim_{k, e} \norm{\Aop_N u}_{L^2} + h^k \norm{u}_{H^{0, \rho}_{\be, h}},
\]
and we therefore conclude
\begin{equation} \label{eq: non-optimal critical point estimate}
	\norm{\Bop_2 u}_{H^{0, \rho}_{\be, h}}
		\lesssim_{k, e} \norm{\Bop_1 u}_{H^{0, \rho}_{\be, h}} + h^{-1}\norm{\Sop \hat \P_e u}_{H^{0, \rho}_{\be, h}} + h^{1/2} \norm{\Sop u}_{H^{0, \rho}_{\be, h}} + h^k \norm{u}_{H^{0, \rho}_{\be, h}}.
\end{equation}
This is almost the desired estimate; we still want to get rid of the term with the factor $h^{1/2}$.
This is done with a standard iteration argument, which proceeds as follows.
Let $\de, \ep$ be given as in the theorem.
By what we have shown,~\eqref{eq: non-optimal critical point estimate} in particular holds if
\begin{align*}
	\WFb(\Bop_2)
		&\subseteq \Omega^{p_N}_{\de}, \\
	\Omega^{p_N}_{(1 + \ep/3)\de}
		&\subseteq \ell(\Sop)
\end{align*}
or if 
\begin{align*}
	\WFb(\Bop_2)
		&\subseteq \Omega^{p_N}_{(1 + 2\ep/3)\de}, \\
	\Omega^{p_N}_{(1 + \ep)\de}
		&\subseteq \ell(\Sop).
\end{align*}
With $\Sop$ playing the role of $\Bop_2$ in the second case, we may therefore improve~\eqref{eq: non-optimal critical point estimate} to
\begin{align*}
	&\norm{\Bop_2 u}_{H^{0, \rho}_{\be, h}} \\
		&\lesssim_{k, e} \norm{\Bop_1 u}_{H^{0, \rho}_{\be, h}} + h^{-1}\norm{\Sop \hat \P_e u}_{H^{0, \rho}_{\be, h}} + h^{1/2} \norm{\Sop u}_{H^{0, \rho}_{\be, h}} + h^k \norm{u}_{H^{0, \rho}_{\be, h}} \\
		&\lesssim_{k, e} \norm{\Bop_1 u}_{H^{0, \rho}_{\be, h}} + h^{-1}\norm{\Sop \hat \P_e u}_{H^{0, \rho}_{\be, h}} + h^k \norm{u}_{H^{0, \rho}_{\be, h}} \\
		&\qquad + h^{1/2} \left( \norm{\Bop_1 u}_{H^{0, \rho}_{\be, h}} + h^{-1}\norm{\Sop \hat \P_e u}_{H^{0, \rho}_{\be, h}} + h^{1/2} \norm{\Sop u}_{H^{0, \rho}_{\be, h}} + h^k \norm{u}_{H^{0, \rho}_{\be, h}} \right) \\
		&\lesssim_{k, e} \norm{\Bop_1 u}_{H^{0, \rho}_{\be, h}} + h^{-1}\norm{\Sop \hat \P_e u}_{H^{0, \rho}_{\be, h}} + h \norm{\Sop u}_{H^{0, \rho}_{\be, h}} + h^k \norm{u}_{H^{0, \rho}_{\be, h}},
\end{align*}
which is a $1/2$ order better in $h$ in the $\norm{S u}$ term.
With an iteration of this argument, the term $h \norm{S u}$ can eventually get absorbed into the term $h^k \norm{u}$, by microlocal regularity.
This completes the proof of the desired estimate near $p_N$.
The proof of the estimate near $p_S$ is completely analogous.
\end{proof}

\begin{remark}[A corresponding adjoint estimate] \label{rmk: adjoint near critical points}
Note that, by choosing $\de_0 > 0$ small enough in the above proof, we could similarly have shown that $\f_N \as_N$ and $\f_S \as_S$ are squares of well-defined symbols. 
This would similarly give an estimate for the \emph{adjoint} operator near the critical points, where the regularity is propagated in the opposite direction (i.e., in the direction of $-\c^\sharp$), but with the weight $\rho$ replaced by $-\rho$.
\end{remark}

\subsubsection{Between the critical points}

Proposition~\ref{prop: critical point estimates} provides the propagation estimates near the critical points of $\c^\sharp$.
We would now like to improve this to get control in a region of the form
\[
	\Omega^{\mathbb{S}^2}_{\de}
		:= \{s, \abs{r - r_0} \in [0, \de]\} \cap \{ \abs \xi \in [0, \de^3] \}.
\]
Note that
\[
	\Omega_\de^{p_N} \cup \Omega_\de^{p_S}
		\subseteq \Omega_\de^{\mathbb{S}^2}.
\]
In order to achieve this, a convenient intermediate step is to propagate the control at the north pole critical point, $p_N$, to control in a `parabolic' region of the form
\[
	\Omega_\de^{\mathrm{par}}
		:= \left\{ s \in [0, \de] \right\} \cap \left\{ \theta + \pi \de^{-2} (r-r_0)^2 \in (0, \pi - \de^3) \right\} \cap \left\{ \theta \in [\de , \pi) \right\} \cap \left\{\abs \xi \in [0, \de^3] \right\},
\]
for some small enough $\de > 0$.

\begin{remark} \label{rmk: out versus sphere regions}
Note that $\Omega_\de^{\mathrm{par}} \subset \Omega_\de^{\mathbb{S}^2}$. 
In particular, $\abs{r - r_0} \leq \de$ in $\Omega_\de^{\mathrm{par}}$.
\end{remark}

The main estimate of this subsection is given in Corollary~\ref{cor: sphere propagation}, but we first prove the following intermediate step.

\begin{prop} \label{prop: between critical point estimates}

There are $\rho_0, \de_0 > 0$ such that if $\rho \in [0, \rho_0]$, $\de \in (0, \de_0)$, $\ep \in (0, 1)$ and $e \in \left(0, \de^{10}\right)$, then the following holds:

If $\Bop_1, \Bop_2, S \in \Psi_{\be, h}^{- \infty, 0, 0}(M)$ are such that
\begin{align*}
	\left( \{ s \in [\de, (1 + \ep)\de ] \} \cup  \{ \theta \in [\de, (1 + \ep)\de ] \} \right) \cap \Omega_{(1 + \ep)\de}^{\mathbb{S}^2}
		&\subseteq \ell(\Bop_1), \\
	\WFb(\Bop_2)
		&\subseteq \Omega_{\de}^{\mathrm{par}}, \\
	\Omega_{(1 + \ep)\de}^{\mathbb{S}^2}
		&\subseteq \ell(\Sop),
\end{align*}
then~\eqref{eq: main low frequency propagation estimate} holds for all $u \in \dot C_c^\infty$, $k \in \N_0$ and $h \in (0, h_0)$, for some $h_0 > 0$ depending on $e$.
\end{prop}

\begin{remark}
The $\de^3$ in the definitions of $\Omega_\de^{\mathbb{S}^2}$ and $\Omega_\de^{\mathrm{par}}$ could have been replaced by $\de^k$ for any $k \geq 3$, and the $\de^{10}$ in Proposition~\ref{prop: between critical point estimates} could have been replaced by $\de^k$ for any $k \geq 10$, if the proof of Proposition~\ref{prop: between critical point estimates} was adjusted correspondingly.
It is just convenient to have explicit expressions that localize sufficiently, c.f.~Remark~\ref{rmk: de k would work}.
\end{remark}

Combining Proposition~\ref{prop: critical point estimates} and Proposition~\ref{prop: between critical point estimates}, we get conclusions at the `full $2$-sphere' at $r - r_0 = s = \abs \xi = 0$.

\begin{cor} \label{cor: sphere propagation}
There are $\rho_0, \de_0 > 0$ such that if $\rho \in [0, \rho_0]$ and $\de \in (0, \de_0)$, then there is an $\ep \in (0, 1)$ such that if $e \in \left(0, \ep \right)$, then the following holds:

If $\Bop_1, \Bop_2, \Sop \in \Psi_{\be, h}^{- \infty, 0, 0}(M)$ are such that
\begin{align}
	\{ s \in [\ep \de, \de] \} \cap \Omega_{\de}^{\mathbb{S}^2}
		&\subseteq \ell(\Bop_1), \label{eq: assumption region in sphere propagation} \\
	\WFb(\Bop_2)
		&\subseteq \Omega_{\ep \de}^{\mathbb{S}^2}, \nonumber \\
	\Omega_{\de}^{\mathbb{S}^2}
		&\subseteq \ell(\Sop), \nonumber
\end{align}
then~\eqref{eq: main low frequency propagation estimate} holds for all $u \in \dot C_c^\infty$, $k \in \N_0$ and $h \in (0, h_0)$, for some $h_0 > 0$ depending on $e$.
\end{cor}

Let us first show how Corollary~\ref{cor: sphere propagation} follows from Proposition~\ref{prop: between critical point estimates} and Proposition~\ref{prop: critical point estimates}.

\begin{proof}[Proof of Corollary~\ref{cor: sphere propagation}]
Let us first choose $\de_0 > 0$ to have the same value as in Proposition~\ref{prop: critical point estimates} and let $\de \in (0, \de_0)$.
If $\de_N \in (0, \de/2)$, then $\Omega_{(1 + \ep_N) \de_N} \subseteq \Omega^{\mathbb{S}^2}_{\de}$ and $\de_N \in (0, \de_0)$, so Proposition~\ref{prop: critical point estimates} implies that~\eqref{eq: main low frequency propagation estimate} holds if
\begin{equation} \label{eq: B2 in pN ball}
	\WFb(\Bop_2)
		\subseteq \Omega_{\de_N}^{p_N}
\end{equation}
for any $\de_N \in (0, \de_0)$ and $e \in \left(0, (\de_N)^3\right)$.

In the next step, we propagate this estimate further using Proposition~\ref{prop: between critical point estimates}, i.e., we let this $\Bop_2$ in~\eqref{eq: B2 in pN ball} play the role of $\Bop_1$ in Proposition~\ref{prop: between critical point estimates}.
For this, note first that if $\de_{\mathrm{par}} \in \left(0, \de_N/2\right)$ and $\ep_{\mathrm{par}} \in (0, 1)$, then
\begin{align*}
	& \left( \{ s \in [\de_{\mathrm{par}}, (1 + \ep_{\mathrm{par}})\de_{\mathrm{par}} ] \} \cup  \{ \theta \in [\de_{\mathrm{par}}, (1 + \ep_{\mathrm{par}})\de_{\mathrm{par}} ] \} \right) \cap \Omega_{\left(1 + \ep_{\mathrm{par}}\right)\de_{\mathrm{par}}}^{\mathbb{S}^2} \\
		& \qquad \subseteq \Omega_{\de_N}^{p_N} \cup \left( \{ s \in [\ep \de, \de] \} \cap \Omega_{\de}^{\mathbb{S}^2} \right),
\end{align*}
if $\ep > 0$ is small enough.
We do have control in this region by the conclusion~\eqref{eq: B2 in pN ball} and by the a priori assumption~\eqref{eq: assumption region in sphere propagation}. 
Thus, Proposition~\ref{prop: between critical point estimates} implies~\eqref{eq: main low frequency propagation estimate} if
\begin{equation} \label{eq: B2 in parabolic}
	\WFb(\Bop_2)
		\subseteq \Omega_{\de_{\mathrm{par}}}^{\mathrm{par}}
\end{equation}
for any $\de_{\mathrm{par}} \in \left(0, \de_N/2 \right)$ and $e \in \left(0, \left(\de_{\mathrm{par}} \right)^{10}\right)$.

It remains to propagate this control into a region of the form $\Omega_{\de_S}^{p_S}$, for a sufficiently small $\de_S$, using Proposition~\ref{prop: between critical point estimates} again, however, this time with the statement near $p_S$.
For this, we need to show that the south pole assumption region in Proposition~\ref{prop: between critical point estimates} is contained in the union of $\Omega_{\de_{\mathrm{par}}}^{\mathrm{par}}$ and the a priori assumption~\eqref{eq: assumption region in sphere propagation}, if $\de_{\mathrm{par}}$, $\de_S$ and $\ep$ are small enough.
The south pole assumption region in Proposition~\ref{prop: between critical point estimates}, near $p_S$, is given for some $\ep_S \in (0, 1)$ by
\begin{align*}
	&\{ \abs w \in [\de_S, (1 + \ep_S)\de_S] \} \cap \Omega^{p_S}_{(1 + \ep_S)\de_S} \\
		&\qquad \subseteq \left\{ \sqrt{s^2 + (\pi - \theta)^2} \in \left[ \frac{\de_S}2, 2 \de_S \right] \right\} \cap \Omega_{(1 + \ep_S) \de_S}^{p_S} \\
		&\qquad \subseteq \left\{ s \in \left[ \frac{\de_S}4, 2 \de_S \right] \right\} \cap \Omega_{(1 + \ep_S) \de_S}^{p_S} \\
		&\qquad \qquad \cup \left\{ \pi - \theta \in \left[ \frac{\de_S}4, 2 \de_S \right] \right\} \cap \Omega_{(1 + \ep_S) \de_S}^{p_S},
\end{align*}
if $\de_S$ is small enough.
Firstly, the set
\[
	\left\{ s \in \left[ \frac{\de_S}4, 2 \de_S \right] \right\} \cap \Omega_{(1 + \ep_S) \de_S}^{p_S}
\]
is covered by~\eqref{eq: assumption region in sphere propagation}, if $\ep$ is small enough.
Secondly, 
\begin{align*}
	&\left\{ \pi - \theta \in \left[ \frac{\de_S}4, 2 \de_S \right] \right\} \cap \Omega_{(1 + \ep_S) \de_S}^{p_S} \\
		&\subseteq \left\{ \pi - \theta \in \left[ \frac{\de_S}4, 2 \de_S \right] \right\} \cap \{s \in [0, 2 \de_S]\} \cap \{ \abs{r - r_0} \in [0, 2\de_S] \} \cap \{ \abs{\xi} \in [0, 2\de_S]\}.
\end{align*}
We claim that this set is contained in $\Omega_{\de_{\mathrm{par}}}^{\mathrm{par}}$ if $\de_{\mathrm{par}}$ and $\de_S$ are small enough.
For this we need to show that if $\pi - \theta \in \left[ \frac{\de_S}4, 2 \de_S \right]$ and $\abs{r - r_0} \in [0, 2\de_S]$, then 
\[
	\theta + \pi \left( \de_{\mathrm{par}} \right)^{-2} (r - r_0)^2
		< \pi - (\de_{\mathrm{par}})^3
\]
or equivalently
\[
	 \pi \left( \de_{\mathrm{par}} \right)^{-2} (r - r_0)^2 + (\de_{\mathrm{par}})^3
		< \pi - \theta
\]
for sufficiently small $\de_{\mathrm{par}}$ and $\de_S$.
Since $\pi - \theta \geq \frac{\de_S}4$ and $(r - r_0)^2 \leq (\de_S)^2$ by Remark~\ref{rmk: out versus sphere regions}, it suffices to show that
\[
	 \pi \left( \de_{\mathrm{par}} \right)^{-2} (\de_S)^2 + (\de_{\mathrm{par}})^3
		< \frac {\de_S}4 
\]
for sufficiently small $\de_{\mathrm{par}}$ and $\de_S$.
Choosing now $\de_{\mathrm{par}}^2 = 5 \pi \de_S$, this becomes
\[
	 \frac15 \de_S + (5 \pi \de_S)^{3/2}
		< \frac {\de_S}4,
\]
which is satisfied if $\de_S > 0$ is small enough.
We may therefore apply Proposition~\ref{prop: between critical point estimates}, the part near $p_S$, and conclude~\eqref{eq: main low frequency propagation estimate} if
\begin{equation} \label{eq: B2 in pS ball}
	\WFb(\Bop_2)
		\subseteq \Omega_{\de_S}^{p_S}
\end{equation}
for any sufficiently small $\de_S$ and if $e \in \left(0, (\de_S)^3\right)$.
In summary, if $\ep > 0$ is chosen such that
\[
	\ep \de 
		< \min \left(\de_S, \de_{\mathbb{S}^2}, \de_N \right),
\]
and such that
\[
	\ep 
		< \min \left((\de_S)^3, (\de_{\mathrm{par}})^{10}, (\de_N)^3 \right),
\]
then by~\eqref{eq: B2 in pN ball},~\eqref{eq: B2 in parabolic} and~\eqref{eq: B2 in pS ball}, we conclude that~\eqref{eq: main low frequency propagation estimate} holds under the assumption that
\[
	\WFb(\Bop_2)
		\subseteq \Omega_{\ep \de}^{\mathbb{S}^2} \subseteq \Omega_{\de_S}^{p_S} \cup \Omega_{\de_{\mathbb{S}^2}}^{\mathrm{par}} \cup \Omega_{\de_N}^{p_N}
\]
and $e \in \left(0, \ep \right)$, as claimed.
\end{proof}

We now continue towards the proof of Proposition~\ref{prop: between critical point estimates}.
Given any $\de > 0$ and $\ep \in (0, 1)$, fix a $\chi \in C_c^\infty(\R)$ satisfying
\begin{itemize}
	\item $\chi \geq 0$,
	\item $\supp(\chi) \subseteq [-(1 + \ep)\de, (1 + \ep)\de]$,
	\item $\chi(s) \equiv 1$ for $\abs s \leq \de$, 
	\item $\left(-s\chi'(s)\chi(s)\right)^{\frac12} \in C^\infty_c(\R)$.
\end{itemize}
In particular, $s \chi'(s) \leq 0$ for all $s \in \R$.
We will also use $\tilde \chi(s) := \chi \left( \frac s {\de^2} \right)$, which has similar properties.
For any fixed $\rho \geq 0$, define the symbol $\as \in \S^{ -\infty, \rho, 0}_{\be, h}(M; \bT^* M)$ by
\[
	\as(z, \xi)
		:= s^{-\rho}\chi(s) \left( 1 - \chi(\theta) \right) \Psi \left( \theta + \pi \de^{-2}(r - r_0)^2 \right)\tilde \chi(\abs{\xi}),
\]
for $s > 0$, where
\[
	\Psi(t)
		:= 
		\begin{cases}
			e^{\frac \digamma {t-\left(\pi - \de^3\right)}}, & t < \pi - \de^3 \\
			0, & t \geq \pi - \de^3,
		\end{cases}
\]
for any $\digamma > 0$, which we eventually will choose sufficiently large.
(The term $\de^3$ in the definition of $\Psi(t)$ could in fact be replaced by $\de^k$, for any $k \geq 3$; the point is just that we get sufficiently close to $p_S$, c.f.~the proof of Corollary~\ref{cor: sphere propagation}.)
Note that
\[
	\esssupp(\as)
		= \Omega_{\de}^{\mathrm{par}}.
\]
We define the commutator to be
\[
	\Aop
		:= \Re\left( \Op_h\left( \as \right) \right).
\]
\begin{lemma} \label{le: commutator equality between critical points}
There are $\rho_0, \de_0 > 0$ such that if $\rho \in [0, \rho_0]$, $\de \in (0, \de_0)$, $\ep \in (0, 1)$, $e \in \left(0, \de^{10}\right)$, then there is a $\digamma_0 > 0$ such that if $\digamma > \digamma_0$, then the following holds:

There are (endomorphism-valued) symbols $\be, \f, \k \in \S_{\be, h}^{-\infty, \rho, 0}$ such that
\begin{align*}
	\esssupp(\f)
		&\subseteq \left( \{ s \in [\de, (1 + \ep)\de ] \} \cup  \{ \theta \in [\de, (1 + \ep)\de ] \} \right) \cap \Omega_{(1 + \de)\ep}^{\mathbb{S}^2} \\
	\esssupp(\k)
		&\subseteq \{ \abs{\xi} \in [\de^3, (1 + \ep)^3\de^3] \} \cap \Omega_{(1 + \de)\ep}^{\mathbb{S}^2},
\end{align*}
such that
\begin{equation} \label{eq: symbol equality parabolic}
	\s \left( \frac i h \left[\tLop_e + \Jop_e, \Aop^2\right] + \frac i h \left( \tLop_e^{*\mB} - \tLop_e \right) \Aop^2 \right)
		= - \be^2 + \f \as + \k \as,
\end{equation}
and $\be$ is elliptic on the interior of $\esssupp(\as)$.
\end{lemma}
\begin{proof}
Analogous to~\eqref{eq: commutator term critical point}, we note that
\begin{align}
	\s \left( \frac i h [\tLop_e + \Jop_e, \Aop] \right) 
		&= \ \left( \left( \id + \frac{\c \otimes \c^\sharp}{G(\c, \c)} \right)\c^\sharp + \sqrt{2e} \cdot \r_e + \md_\xi \j_e \right)\left(\md_z \as^2 \right) + \as \k,
\end{align}
where 
\[
	\k
		:= - 2 \md_z \left( \tl_e + \j_e \right) \cdot \md_\xi a,
\]
so that indeed $\supp(\k) \subset \{ \abs \xi \in [\de^3, (1 + \ep)^3\de^3] \}$.
We first compute
\begin{align*}
	&\md_z \as \\*
		&\quad = \left( \rho \as - s^{-\rho} s \chi'(s) \left( 1 - \chi(\theta) \right) \Psi \left( \theta + \pi \de^{-2}(r - r_0)^2 \right) \tilde \chi(\abs \xi) \right) \left( - \frac{\md s}s \right) \\
		&\quad \qquad - s^{-\rho} \chi(s) \chi'(\theta) \Psi \left( \theta + \pi \de^{-2}(r - r_0)^2 \right) \tilde \chi(\abs \xi) \md \theta \\
		&\quad \qquad + s^{-\rho} \chi(s) \left( 1 - \chi(\theta) \right) \Psi' \left( \theta + \pi \de^{-2}(r - r_0)^2 \right) \tilde \chi(\abs \xi) \left( 2 \pi \de^{-2}(r - r_0) \md r + \md \theta \right).
\end{align*}
By Equation~\eqref{eq: c sharp in new coordinates}, we have
\begin{align*}
	\c^\sharp \left( \md_z \as \right) 
		&= \left( - s \d_{s} + a \frac{r^2 - r_0^2}{(r^2 + a^2)(r_0^2 + a^2)} \d_{\psi} + \a f(r) \d_r + \b \sin(\theta) \d_\theta \right) \as \\
		&\quad = \rho \as - s^{-\rho} s \chi'(s) \left( 1 - \chi(\theta) \right) \Psi \tilde \chi(\abs \xi) \\
		&\quad \qquad - s^{-\rho} \chi(s) \chi'(\theta) \Psi \tilde \chi(\abs \xi) \b \sin(\theta) \\
		&\quad \qquad + s^{-\rho} \chi(s) \left( 1 - \chi(\theta) \right) \Psi' \tilde \chi(\abs \xi) \left(2\a\pi \de^{-2} (r - r_0)f(r) + \b \sin(\theta) \right).
\end{align*}
Using that $\Psi'(t) = - \frac \digamma {(t - \left(\pi - \de^3 \right))^2} \Psi(t)$, we conclude that
\begin{align*}
	\left( \left( \id + \frac{\c \otimes \c^\sharp}{G(\c, \c)} \right)\c^\sharp + \sqrt{2e} \cdot \r_e + \md_\xi \j_e \right)\left(\md_z \as^2 \right)
		&= - \be^2 + \as \f,
\end{align*}
where
\begin{align*}
	\f
		&= - 2 s^{-\rho} s \chi'(s) \left( 1 - \chi(\theta) \right) \Psi \tilde \chi(\abs \xi) \\*
		&\qquad \quad \cdot \left( \id + \frac{\c \otimes \c^\sharp}{G(\c, \c)} + \left( \sqrt{2e} \cdot \r_e + \md_\xi \j_e \right) \left( - \frac{\md s}s \right) \cdot \id \right) \\
		&\qquad - s^{-\rho} \chi(s) \chi'(\theta) \tilde \Psi \chi(\abs \xi) \\*
		&\qquad \quad \cdot \left( \b \sin(\theta) \left( \id + \frac{\c \otimes \c^\sharp}{G(\c, \c)} \right) + \left( \sqrt{2 e} \cdot \r_e + \md_\xi \j_e \right)(\md \theta) \right), \\
	\be^2
		&= \as^2 \Bigg( - 2 \rho \left( \id + \frac{\c \otimes \c^\sharp}{G(\c, \c)} + \left( \sqrt{2e} \cdot \r_e + \md_\xi \j_e \right) \left( - \frac{\md s}s \right) \cdot \id \right) \\
		&\qquad \qquad + \frac \digamma {\left( \theta + \pi \de^{-2}(r - r_0)^2 - \left(\pi - \de^3\right) \right)^2} \\
		&\qquad \qquad \qquad \cdot \bigg( \left( 2\a\pi \de^{-2}(r - r_0)f(r) + \b \sin(\theta) \right) \left( \id + \frac{\c \otimes \c^\sharp}{G(\c, \c)} \right) \\*
		&\qquad \qquad \qquad + \left( \sqrt{2e} \cdot \r_e + \md_\xi \j_e \right) \left( 2\pi \de^{-2}(r - r_0) \md r + \md \theta \right) \cdot \id \bigg) \Bigg).
\end{align*}
Note that $\f$ clearly has the desired support. 
We now claim that $\be^2$ really is the square of a well-defined symbol $\be \in \S^{-\infty, \rho, 0}_{\be, h}$, if $\de_0$ is sufficiently small.
Since the factor $\as^2$ is the square of the symbol $\as$,  it remains to show that the endomorphism is positive definite, if $\digamma > 0$ is large enough, so that we can take the square root.
Similarly as in the proof of Lemma~\ref{le: prop at critical points}, we may estimate
\begin{align*}
	& \abs{\left( \sqrt{2e} \cdot \r_e + \md_\xi \j_e \right) \left( 2\pi \de^{-2}(r - r_0) \md r + \md \theta \right) \cdot \id} \\
		&\qquad \leq C \left( \sqrt{2e} + \abs{\xi}^2 \right) \abs{ \left( 2\pi \de^{-2}(r - r_0) \md r + \md \theta \right)} \id \\
		&\qquad \leq C \de^{-1} \left( \de^5 + \de^6 \right) \id \\
		&\qquad \leq C \de^4 \cdot \id,
\end{align*}
in $\esssupp(\as)$.
On the other hand, we conclude that
\begin{align*}
	\left( 2\a \pi \de^{-2}(r - r_0)f(r) + \b \sin(\theta) \right) \left( \id + \frac{\c \otimes \c^\sharp}{G(\c, \c)} \right)
		&\geq \b \sin(\theta) \left( \id + \frac{\c \otimes \c^\sharp}{G(\c, \c)} \right) \\
		&\geq c \de^3 \cdot \id,
\end{align*}
for some $c > 0$, since $\theta > \de$ and $\pi - \theta < \de^3$ in the support of $\as$.
Therefore, if $\de_0 > 0$ is sufficiently small,
\begin{align*}
	&\left( 2\a \pi \de^{-2} (r - r_0)f(r) + \b \sin(\theta) \right) \left( \id + \frac{\c \otimes \c^\sharp}{G(\c, \c)} \right) \\*
		&\qquad + \left( \sqrt{2e} \cdot \r_e + \md_\xi \j_e \right) \left( 2\pi \de^{-2}(r - r_0) \md r + \md \theta \right) \cdot \id \\
		&\qquad \geq c \de^3 \cdot \id,
\end{align*}
for some $c > 0$.
Since 
\[
	 \abs{- 2 \rho \left( \id + \frac{\c \otimes \c^\sharp}{G(\c, \c)} + \left( \sqrt{2e} \cdot \r_e + \md_\xi \j_e \right) \left( - \frac{\md s}s \right) \cdot \id \right) }
\]
is uniformly bounded for $e \in (0, 1]$, it follows that $\be^2$ is indeed positive in $K$, if $\de_0 > 0$ is sufficiently small and $\digamma > 0$ is large enough (depending on $\de$).
We conclude that $\be \in \S^{-\infty, \rho, 0}_{\be, h}$ if $\de_0 > 0$ is sufficiently small and $\digamma > 0$ is large enough. 
\end{proof}

\begin{proof}[Proof of Proposition~\ref{prop: between critical point estimates}]

Quantizing the symbol equality~\eqref{eq: symbol equality parabolic}, we get
\begin{align*}
	\frac i h 
		\left[\tLop_e + \Jop_e, \Aop^2\right] + \frac i h \left( \tLop_e^{*\mB} - \tLop_e \right) \Aop^2
		= - \Bop^* \Bop + \Aop \Fop + \Aop \Kop + h \Rop,
\end{align*}
where $\Bop \in \Psi^{-\infty, \rho, 0}$  is elliptic on $\WFb(\Aop)$, $\Fop, \Kop \in \Psi^{-\infty, \rho, 0}_{\be, h}$ and $\Rop \in \Psi^{-\infty, 2 \rho, 0}_{\be, h}$ and
\begin{align*}
	\WFb(\Fop)
		&\subseteq \left( \{ s \in [\de, (1 + \ep)\de ] \} \cup  \{ \theta \in [\de, (1 + \ep)\de ] \} \right) \cap \Omega_{(1 + \de)\ep}^{\mathbb{S}^2}, \\
	\WFb(\Kop)
		&\subseteq \{ \abs{\xi} \in [\de^2, (1 + \ep)^2\de^2] \} \cap \Omega_{(1 + \ep)\de}^{\mathbb{S}^2}, \\
	\WFb(\Rop)
		&\subseteq \Omega_{(1 + \ep)\de}^{\mathbb{S}^2}.
\end{align*}
It follows that
\begin{align*}
	&\ldr{\frac i {2h} \left( \left[\tLop_e + \Jop_e, \Aop_N^2\right] + \left( \tLop_e^{*\mB} - \tLop_e \right) \Aop_N^2\right)u, u} \\
		&\quad = - \norm{\Bop u}_{L^2}^2 + \ldr{\Fop u, \Aop u} + \ldr{\Kop u, \Aop u} + h \ldr{\Rop u, u} \\
		&\quad \leq - c \norm{\Aop u}_{L^2}^2 + \left( \norm{\Fop u}_{L^2} + \norm{\Kop u}_{L^2} \right) \norm{\Aop u}_{L^2} + h \ldr{\Rop u, u} + h^{2k} \norm{u}_{H^{0, \rho}_{\be, h}}^2,
\end{align*}
for some $c > 0$, since $\Bop$ is elliptic on $\WFb(\Aop)$.
Since $\Sop \hat \P_e$ is elliptic on $\WFb(\Kop)$, by Proposition~\ref{prop: ellipticity finite frequencies}, and since $\Bop_1$ is elliptic on $\WFb(\Fop)$, by assumption, microlocal elliptic regularity theory implies that
\begin{align*}
	\norm{\Kop u}_{L^2}
		&\lesssim_{k, e} \norm{\Sop \hat \P_e u}_{H^{0, \rho}_{\be, h}} + h^k \norm{u}_{H^{0, \rho}_{\be, h}}, \\
	\norm{\Fop u}_{L^2}
		&\lesssim_{k, e} \norm{\Bop_1 u}_{H^{0, \rho}_{\be, h}} + h^k \norm{u}_{H^{0, \rho}_{\be, h}},
\end{align*}
for all $k$.
Moreover,
\[
	\ldr{\Rop u, u}_{L^2}
		= \ldr{s^{\rho} \Rop_N u, s^{-\rho}u}_{L^2}
		\lesssim \norm{\Sop u}_{H^{0, \rho}_{\be, h}}^2 + h^{2k} \norm{u}_{H^{0, \rho}_{\be, h}}^2.
\]
Combining~\eqref{eq: low order commutator} with~\eqref{eq: commutator equality} and applying Lemma~\ref{le: the right-hand-side low}, we conclude the estimate
\begin{align*}
	\norm{\Aop u}_{L^2}^2
		&\lesssim_{k, e} \left( \norm{\Bop_1 u}_{H^{0, \rho}_{\be, h}} + h^{-1}\norm{\Sop \hat \P_e u}_{H^{0, \rho}_{\be, h}} + h^k \norm{u}_{H^{0, \rho}_{\be, h}} \right) \norm{\Aop u}_{L^2} \\
		&\quad \qquad + \left(\de + h^{1/2} \right) \norm{\Aop u}_{L^2}^2 + h \norm{\Sop u}^2_{H^{0, \rho}_{\be, h}} + h^k \norm{u}_{H^{0, \rho}_{\be, h}}^2,
\end{align*}
which implies the estimate
\begin{align*}
	\norm{\Aop u}_{L^2}
		&\lesssim_{k, e} \norm{\Bop_1 u}_{H^{0, \rho}_{\be, h}} + h^{-1}\norm{\Sop \hat \P_e u}_{H^{0, \rho}_{\be, h}} + h^{1/2} \norm{\Sop u}_{H^{0, \rho}_{\be, h}} + h^k \norm{u}_{H^{0, \rho}_{\be, h}},
\end{align*}
if $\de$ and $h$ are small enough.
Since $\Aop$ is elliptic on $\WFb(\Bop_2)$,
\[
	\norm{\Bop_2 u}_{H^{0, \rho}_{\be, h}}
		\lesssim_{k, e} \norm{\Aop u}_{L^2} + h^k \norm{u}_{H^{0, \rho}_{\be, h}},
\]
and we therefore conclude
\begin{equation} \label{eq: non-optimal parabolic estimate}
	\norm{\Bop_2 u}_{H^{0, \rho}_{\be, h}}
		\lesssim_{k, e} \norm{\Bop_1 u}_{H^{0, \rho}_{\be, h}} + h^{-1}\norm{\Sop \hat \P_e u}_{H^{0, \rho}_{\be, h}} + h^{1/2} \norm{\Sop u}_{H^{0, \rho}_{\be, h}} + h^k \norm{u}_{H^{0, \rho}_{\be, h}}.
\end{equation}
The factor $h^{1/2}$ can be improved to a factor $h^k$, for an arbitrary $k \in \N_0$, by iterating this estimate similar to the end of the proof of Proposition~\ref{prop: critical point estimates}.
Therefore, the second last term in~\eqref{eq: non-optimal parabolic estimate} can be absorbed by the last term in~\eqref{eq: non-optimal parabolic estimate}, yielding the desired estimate.
\end{proof}

\begin{remark}[A corresponding adjoint estimate] \label{rmk: adjoint between critical points}
Again, obtains an analogous estimate for the adjoint operator propagated in the direction $-\c^\sharp$ instead of $\c^\sharp$, using the adjoint estimate at the critical points, described in Remark~\ref{rmk: adjoint near critical points}.
Concretely, the symbol $\as$ in the commutator estimate, Lemma~\ref{le: commutator equality between critical points}, would be replaced by
\[
	\as'(z, \xi)
		:= s^{\rho} \chi(s) (1 - \chi(\pi - \theta)) \Psi\left(\pi - \theta + \pi \de^{-2}(r - r_0)^2 \right)
\]
propagating the control at $p_N$ to a parabolic region of the type
\begin{align*}
	{\Omega_\de^{\mathrm{par}}}'
		:= & \ \{s \in [0, \de]\} \cap \{ \pi - \theta + \pi \de^{-2} (r - r_0)^2 \in (0, \pi-\de^3) \} \\
		&\quad\cap \{ \theta \in (0, \pi - \de] \} \cap \{ \abs \xi \in [0, \de^3] \},
\end{align*}
analogous to Proposition~\ref{prop: between critical point estimates}, and then continue by propagating the control to $\Omega_{\ep \de}^{\mathbb{S}^2}$ analogous to Corollary~\ref{cor: sphere propagation}.
Again, since the propagation is in the reversed direction, the weights $\rho$ get replaced by $-\rho$.
\end{remark}

\subsubsection{Global estimates near the zero section}

Before finally proving Proposition~\ref{prop: main low frequency estimate}, it is natural to first prove the we have~\eqref{eq: main low frequency propagation estimate} close to timelike infinity, where $s = 0$.

\begin{prop} \label{prop: at timelike infinity}

There are $\rho_0, \de_0 > 0$ such that if $\rho \in [0, \rho_0]$, $\de \in (0, \de_0)$ and $e \in \left(0, \de^3\right)$, there is a $\ep \in (0, 1)$ such that the following holds:

If $\Bop_1, \Bop_2, S \in \Psi_{\be, h}^{- \infty, 0, 0}(M)$ are such that
\begin{align*}
	\{ s \in [\ep \de, \de] \} \cap \{ \abs \xi \in [0, \de^2] \}
		&\subseteq \ell(\Bop_1), \\
	\WFb(\Bop_2)
		&\subseteq \{ s \in [0, \ep \de] \} \cap \{ \abs \xi \in [0, \de^2] \}, \\
	\{ s \in [0, \de] \} \cap \{ \abs \xi \in [0, \de^2] \}
		&\subseteq \ell(\Sop),
\end{align*}
then~\eqref{eq: main low frequency propagation estimate} holds for all $u \in \dot C_c^\infty$ and all $k \in \N_0$.
\end{prop}

\begin{proof}
Let $\rho_0, \de_0 > 0$, $\rho \in [0, \rho_0]$, $\de \in (0, \de_0)$, $e \in \left(0, \de^3\right)$ and $\ep \in (0, 1)$ be as in Corollary~\ref{cor: sphere propagation}.
Then, by Corollary~\ref{cor: sphere propagation},~\eqref{eq: main low frequency propagation estimate} holds if
\[
	\WFb(\Bop_2) 
		\subseteq \Omega_{\ep \de}^{\mathbb{S}^2}.
\]
The key point is that any integral curve of $\c^\sharp$ goes to $\Omega_{\ep \de}^{\mathbb{S}^2} \cup \{ s \in [\ep \de, \de] \} \cap \{ \abs \xi \in [0, \de^2] \}$ to the past and transversely through either
\[
	\{r = r_e - \ep_M\}
\]
or 
\[
	\{r = r_c + \ep_M\}
\]
to the future.
The proposition is therefore proven using symbols of the form
\begin{align*}
	\as(z, \xi)
		:= & \ s^{-\rho}\chi_\de(s) \Phi \left( r \right)\chi_\de(\abs{\xi}),
\end{align*}
where of $\Phi'$ is very negative, and beats $\Phi$, towards the hypersurfaces $\{r = r_e - \ep_M\}$ and $\{r = r_c + \ep_M\}$.
These steps are very similar, but significantly simpler, than the corresponding steps in the proof of Proposition~\ref{prop: between critical point estimates}.
\end{proof}

We finally turn to the proof of Proposition~\ref{prop: main low frequency estimate}.

\begin{proof}[Proof of Proposition~\ref{prop: main low frequency estimate}]
Similarly to the proof of Proposition~\ref{prop: at timelike infinity}, improving the statement of Proposition~\ref{prop: at timelike infinity} to the statement of Proposition~\ref{prop: main low frequency estimate} is an even simpler propagation of semiclassical singularities along the integral curves of $\c^\sharp$ away from timelike infinity, where $s = 0$.
Indeed, noting that each integral curve of $\c^\sharp$ passes through the hypersurface $\{s = s_0\}$ to the past and passes either through $\{s = \ep \de\}$, $\{r = r_e - \ep_M\}$ or $\{r = r_c + \ep_M\}$, this is a rather immediate consequence of Proposition~\ref{prop: at timelike infinity}. 
\end{proof}

\begin{remark}[A corresponding adjoint estimate] \label{rmk: adjoint global low freq estimate}
Yet again, c.f.~Remark~\ref{rmk: adjoint near critical points} and Remark~\ref{rmk: adjoint between critical points}, the propagation in the opposite direction follows analogous also near timelike infinity and away from timelike infinity with the weights $\rho$ replaced by $-\rho$, proving the adjoint low frequency estimate, Proposition~\ref{prop: main low frequency adjoint estimate}.
\end{remark}

\subsection{High frequency propagation estimates} \label{sec: high frequency estiamtes}

We now consider the high frequencies for the operator $\tP_e$ in the Kerr--de~Sitter spacetimes.
Note that there is no need to modify $\tP_e$ to $\hat \P_e$, as we did in the previous section.
Since we already have estimates in compact sets of finite non-zero frequencies, by Proposition~\ref{prop: ellipticity finite frequencies}, and near zero frequency, by Proposition~\ref{prop: main low frequency estimate}, we now focus on the semiclassical characteristic set at infinite frequencies, given by
\begin{align*}
	\Sigma^\pm
		:= & \ \{ \xi \in \bS M \mid G(\xi, \xi) = 0\ \text{and}\ {\mp}G(\c, \xi) > 0 \} \\
		= & \  \{ \xi \in \bS M \mid G(\xi, \xi) = 0\ \text{and}\ {\pm}\xi\ \text{is future directed} \}.
\end{align*}
The key in the high frequency propagation estimates is that by Lemma~\ref{le: skew adjoint part} and~\eqref{eq: tilde l 0}, there is a $e_0 > 0$ such that
\begin{equation} \label{eq: skew-adjoint part}
	\pm \s_{\be, h} \left( \frac i2 \left(\tP_e^{*\mB} - \tP_e \right) \right)|_{\Sigma^\pm}
		= \pm \tl_e|_{\Sigma^\pm}
		< 0
\end{equation}
for all $e \in (0, e_0]$, as quadratic forms.
The propagation estimates will be based on the commutator equality
\begin{equation} \label{eq: commutator equality 2}
	\Im \ldr{\frac1 h \Aop \tP_e u, \Aop u}
		= \ldr{ \frac i {2 h} \left( [\tP_e, \Aop^2] + \left(\tP_e^{*\mB} - \tP_e\right) \right) \Aop^2 u, u},
\end{equation}
where we note that our skew-adjoint term~\eqref{eq: skew-adjoint part} will appear with an extra factor $h^{-1}$.
The skew-adjoint term~\eqref{eq: skew-adjoint part} has the correct sign for our \emph{future} propagation setup (regularity in $\Sigma^+$ will be propagated forward and regularity in $\Sigma^-$ will be propagated backward), and the largeness of $h^{-1}$ will be used in order to remove various regularity thresholds, which would normally be present in a propagation estimate in $\Sigma^\pm$.
The main goal for this section is to prove the following:

\begin{prop}[The high frequency propagation estimate] \label{eq: global high frequency estimates}
Let $s_0 > 0$ and $\ep \in (0, 1)$ be given, let $K \subset \boT M \backslash o$ be a compact subset and let $\Bop_1, \Bop_2, \Sop \in \Psi_{\be, h}^{0, 0, 0}(M)$.
Let $e_0 > 0$ be such that~\eqref{eq: skew-adjoint part} holds for all $e \in (0, e_0]$.
Assume that
\begin{align}
	\bS M \cap \{ s \in [\ep s_0, s_0] \}
		&\subseteq \ell(\Bop_1), \label{eq: B_1 assumption high} \\
	\WFb(\Bop_2)
		&\subseteq \{ s \in [0, s_0] \} \cap  K, \nonumber \\
	K
		&\subseteq \ell(\Sop). \nonumber
\end{align}
Then, for all $m, \rho \in \R$, $k \in \N_0$ and $e \in (0, e_0]$, there is an $h_0 > 0$ such that
\begin{equation} \label{eq: main high frequency propagation estimate}
	\norm{\Bop_2 u}_{H^{m, \rho}_{\be, h}}
		\lesssim_{k, e} \norm{\Bop_1 u}_{H^{m, \rho}_{\be, h}} + h^{-1} \norm{\Sop \tP_e u}_{H^{m-1, \rho}_{\be, h}} + h^k \norm{u}_{H^{-k, \rho}_{\be, h}},
\end{equation}
for all $h \in (0, h_0)$ in the strong sense that if the right hand side in~\eqref{eq: main high frequency propagation estimate} is finite, then so is the norm on the left, and the estimate holds.
For the adjoint operator $\tP_e^{*\mB}$, the same holds if we replace~\eqref{eq: B_1 assumption high} by the assumption
\[
	\bS M \cap \left( \{ r \in [r_e - \ep_M, r_e - \ep \cdot \ep_M] \cup [r_c + \ep \cdot \ep_M, r_c + \ep_M] \} \right)
		\subseteq \ell(\Bop_1).
\]
\end{prop}

\subsubsection{Propagation away from critical points}

The first proposition gives the propagation estimates away from any critical points.

\begin{prop} \label{prop: high away from critical points}
Let $\Bop_1, \Bop_2, \Sop \in \Psi_{\be, h}^{0, 0, 0}(M)$ and let $e_0 > 0$ be such that~\eqref{eq: skew-adjoint part} holds for all $e \in (0, e_0]$.
Assume that
\[
	\WFb(\Bop_2)
		\subseteq \ell(\Sop) \backslash o,
\]
and that all backward (resp.~forward) bicharacteristics of $\sfH_G$ in $\Sigma^+$ (resp.~$\Sigma^-$) through $\WFb(\Bop_2)$ pass through $\ell(\Bop_1)$ while staying in $\ell(\Sop)$.

Then, for all $m, \rho \in \R$, $k \in \N_0$ and $e \in (0, e_0]$, there is an $h_0 > 0$ such that~\eqref{eq: main high frequency propagation estimate} holds for all $h \in (0, h_0)$ in the strong sense that if the right hand side in~\eqref{eq: main high frequency propagation estimate} is finite, then so is the norm on the left, and the estimate holds.

Moreover, the same estimate holds for the adjoint operator $\tP_e^{*\mB}$ interchanging `backward' and `forward'. 
\end{prop}

\begin{proof}
The proof follows from a standard positive commutator argument.
Let 
\[
	\Aop = \Aop^{*\mB} \in \Psi_{\be, h}^{m - \frac12, \rho, 0}(M),
\]
and consider~\eqref{eq: commutator equality 2}.
The left-hand side is simply estimated as
\[
	\Im \ldr{\frac1 h \Aop \tP_e u, \Aop u}
		\leq \frac {C_1} {h^2} \norm{\Aop \tP_e u}_{L^2}^2 + \frac1{C_1} \norm{\Aop u}_{L^2}^2,
\]
for some large constant $C_1 > 0$, yet to be chosen.
If $\Aop = \Op_h(\as)$, where the principal symbol $\as \in \S^{m-\frac12,\rho, 0}_{\be, h}$ is real and scalar, Proposition~\ref{prop: KdS skew adjoint term} implies that
\[
	\s_h \left( \frac i{2h} \left( [\tP_e, \Aop^2] + \left(\tP_e^{*\mB} - \tP_e\right) \Aop^2 \right) \right) 
		= \sfH_G \as^2 + h^{-1} \tl_e \as^2.
\]
By the usual construction of the symbol $\as$, we may arrange that
\[
	\sfH_G \as^2
		= - \be_2^2 + \be_1^2
\]
in $\Sigma^+$, where 
\begin{itemize}
	\item $\esssupp(\be_1) \subset \ell(\Bop_1)$,
	\item $\esssupp(\as) \subset \ell(\Sop)$,
	\item $\be_2 \neq 0$ on
\[
	\esssupp(\as) \backslash \ell(\Bop_1) \supset \WFb(\Bop_2).
\]
\end{itemize}
Since $h^{-1} \tl_e \as^2$ is negative for $\xi \in \Sigma^+$ by~\eqref{eq: skew-adjoint part}, and thereby has the same sign as $-\be_2^2$, we can disregard it from the estimate and the quantization gives 
\[
	\norm{\Bop_2 u}_{H^{m, \rho}_{\be, h}}
		\lesssim_{k, e} \norm{\Bop_1 u}_{H^{m, \rho}_{\be, h}} + h^{-1} \norm{\Sop \tP_e u}_{H^{m-1, \rho}_{\be, h}} + h \norm{u}_{H^{m-1, \rho}_{\be, h}}.
\]
Iterating this similar to the proof of Proposition~\ref{prop: critical point estimates}, we may improve this to the desired estimate~\eqref{eq: main high frequency propagation estimate}.
For the estimate to hold in the strong sense, one applies a standard regularization argument.
\end{proof}

\subsubsection{Propagation at radial points}
We now turn to the estimate for the first type of critical points of the bicharacteristic flow in $\Sigma^\pm$, namely the radial points $\mR^\pm \subset \Sigma^\pm$.

\begin{prop} \label{prop: high radial points}
Let $\Bop_1, \Bop_2, \Sop \in \Psi_{\be, h}^{0, 0, 0}(M)$ with wave front sets disjoint form $\Sigma^-$ (resp.~$\Sigma^+$). 
Let $e_0 > 0$ be such that~\eqref{eq: skew-adjoint part} holds for all $e \in (0, e_0]$.
Assume that
\[
	\WFb(\Bop_2)
		\subseteq \ell(\Sop) \backslash o,
\]
that $B_2$ is elliptic at $\mR^+$ (resp.~$\mR^-$) and that all backward (resp.~forward) bicharacteristics through $\WFb(\Bop_2) \cap \bS M$ either tend to $\mR^+$ (resp.~$\mR^-$) or enter $\ell(\Bop_2)$ in finite time, while remaining in $\ell(\Sop)$.

Then, for all $m, \rho \in \R$, $k \in \N_0$ and $e \in (0, e_0]$, there is an $h_0 > 0$ such that~\eqref{eq: main high frequency propagation estimate} holds for all $h \in (0, h_0)$ in the strong sense that if the right hand side in~\eqref{eq: main high frequency propagation estimate} is finite, then so is the norm on the left, and the estimate holds.

Moreover, the same estimate holds for the adjoint operator $\tP_e^{*\mB}$ interchanging `backward' and `forward'. 
\end{prop}

\begin{proof}
Again, the proof is based on choosing $\Aop$ properly and estimate both sides in equality~\eqref{eq: commutator equality 2}.
The bicharacteristic flow has \emph{saddle points} $\mR^\pm$, and we refer the reader to \cite{HV2015}*{Prop.~2.1} for the details of such estimates, which are very similar in structure to the proof of Proposition~\ref{prop: critical point estimates}. 
Just like in the proof of Proposition~\ref{prop: critical point estimates}, the key is that $\Aop$ is elliptic at the critical points, so that the skew-adjoint part is useful at those points.
By the non-positivity (resp.~non.negativity) of the skew-adjoint part, see~\eqref{eq: skew-adjoint part}, it follows that~\eqref{eq: main high frequency propagation estimate} holds under the regularity threshold
\[
	m - \b \rho - \frac12 
		> 0,
\]
where $\b := \max(\b_e, \b_c)$, where
\[
	\b_{e/c} 
		:= \pm 2 \left(1 + \frac{\Lambda a^2}3\right) \frac{r^2 + a^2}{\mu'(r)}|_{r = r_{e/c}},
\]
c.f.~\cite{V2013}*{Eq.~(6.10)} and \cite{PV2023}*{Thm.~2.1}.
However, the extra negativity of~\eqref{eq: skew-adjoint part} at the critical points, which comes with a factor $h^{-1}$, implies that the regularity only needs to satisfy
\[
	m - \b \rho - \frac12 + h^{-1} \de 
		> 0,
\]
for some $\de > 0$.
Given any $m$ and $\rho$, this inequality is satisfied for sufficiently small $h$.
This is the reason there is no regularity threshold in this statement, as opposed to \cite{HV2015}*{Prop.~2.1}.
\end{proof}

\subsubsection{Propagation near the trapped set}

The trapping in subextremal Kerr--de~Sitter spacetimes is normally hyperbolic, as shown in \cite{PV2023}*{Thm.~3.2} (see also \cite{V2013}*{Sec.~6}).
It consists of two components $\Gamma^\pm$, which are given by transversal intersections
\[
	\Gamma^\pm
		= \Gamma^\pm_{fw} \cap \Gamma^\pm_{bw} \subset \Sigma^\pm,
\]
of the \emph{forward} trapped set $\Gamma^\pm_{fw}$ and \emph{backward} trapped set $\Gamma^\pm_{bw}$.
The propagation estimate we need near the trapping is the following:

\begin{prop} \label{prop: high trapped set}
There are $\Bop_1, \Bop_2, \Sop \in \Psi_{b,h}^{0,0,0}(M)$ such that
\begin{itemize}
	\item $\Gamma^+ \subset \ell(\Bop_2)$ (resp.~$\Gamma^- \subset \ell(\Bop_2)$),
	\item $\WFb(\Bop_1) \cap \Gamma^+_{bw} = \emptyset$ (resp.~$\WFb(\Bop_1) \cap \Gamma^-_{fw} = \emptyset$),
\end{itemize}
and such that for all $m, \rho \in \R$, $k \in \N_0$ and $e \in (0, e_0]$, there is an $h_0 > 0$ such that~\eqref{eq: main high frequency propagation estimate} holds for all $h \in (0, h_0)$, in the strong sense that if the right hand side in~\eqref{eq: main high frequency propagation estimate} is finite, then so is the norm on the left, and the estimate holds.
Moreover, the same estimate holds for the adjoint operator $\tP_e^{*\mB}$, now with 
\begin{itemize}
	\item $\WFb(\Bop_1) \cap \Gamma^+_{fw} = \emptyset$ (resp.~$\WFb(\Bop_1) \cap \Gamma^-_{bw} = \emptyset$),
\end{itemize}
instead. 
\end{prop}
\begin{proof}
This statement relies on the positive commutator estimate in the proof of \cite{HV2013}*{Thm.~3.2}, with one important difference.
In \cite{HV2013}*{Thm.~3.2}, it was assumed that the skew-adjoint part of the operator $\P$ was in $\Psi_b^{m-2, 0}(M)$.
Note here that $\P$ was a non-semiclassical operator in the reference.
In our case, the (semiclassical) skew-adjoint part in the positive commutator estimate is
\[
	\frac1{2i h} \left( \tP_e - \tP_e^{*\mB} \right)
\] 
and Lemma~\ref{le: skew adjoint part} in particular implies that
\[
	\mp \s_{\be, h} \left( \frac1{2i} \left( \tP_e - \tP_e^{*\mB} \right) \right)|_{\Gamma^\pm}
		> 0.
\]
Hence, at the trapping, we get a term with the `good sign' in both $\Sigma^+$ and $\Sigma^-$, shifting the threshold weight $r = 0$ to any fixed $C > 0$, provided $h_0 > 0$ is sufficiently small.
Consequently, estimate~\eqref{eq: main high frequency propagation estimate} holds for any $m, \rho \in \R$ (with the appropriately small $h_0 > 0$), unlike in \cite{HV2013}*{Thm.~3.2}.
\end{proof}

\subsubsection{Proving the global high frequency estimates}

We are finally ready to prove the high frequency estimates.

\begin{proof}[Proof of Proposition~\ref{eq: global high frequency estimates}]
First recall that $\tP_e$ and $\tP_e^{*\mB}$ are elliptic in
\[
	\bS M \backslash \left( \Sigma^+ \cup \Sigma^- \right).
\]
In $\Sigma^+ \cup \Sigma^-$, we combine Proposition~\ref{prop: high away from critical points}, Proposition~\ref{prop: high radial points} and Proposition~\ref{prop: high trapped set} and the structure of the null-geodesic flow (see \cite{V2013}*{§6}, \cite{HV2017}*{§2} and \cite{PV2023}*{§3}), which together prove the statement.
\end{proof}

\subsection{Energy estimates near the initial and final hypersurfaces}
We would now like to improve upon the propagation estimates Proposition~\ref{prop: main low frequency estimate} (resp.~Proposition~\ref{prop: main low frequency adjoint estimate}) and Proposition~\ref{eq: global high frequency estimates}, in which there was an assumption region (i.e., where the $\Bop_1$ was supposed to be elliptic), to instead prove the estimates on Sobolev spaces with support conditions.
Concretely, we would like to prove estimates on a spacetime with corners of the form
\[
	\Omega_{s_0}
		:= \{ s \in [0, s_0] \} \cap \{r \in [r_e - \ep_M, r_c + \ep_M]\} \subset M,
\]
for some $s_0 > 0$. 
Eventually we will set $s_0 := - \ln(t_0)$, in which case $\Omega_{s_0} = \Omega$.
We will work with the function spaces
\[
	H^{m, \rho}_{\be, h}(\Omega_{s_0})^{\bullet, -},
\]
which consists of restrictions to $\Omega_{s_0}$ of distributions in $H^{m, \rho}_{\be, h}(M)$ which are \emph{supported} (in the sense of \cite{H2009}*{App.~B.2}) with respect to the boundary hypersurface
\[
	\Sigma_0
		:= \ \{ s = s_0\} \cap \{r \in [r_e - \ep_M, r_c + \ep_M]\},
\]
and \emph{extendible} (in the sense of \cite{H2009}*{App.~B.2}) at the boundary hypersurfaces
\begin{align*}
	\Sigma_e
		:= & \ \{ s \in [0, s_0] \} \cap \{r = r_e - \ep_M\}, \\ 
	\Sigma_c
		:= & \ \{ s \in [0, s_0] \} \cap \{r = r_e - \ep_M\}.
\end{align*}
The dual space to $H^{m, \rho}_{\be, h}(\Omega_{s_0})^{\bullet, -}$ is given by
\[
	H^{-m, -\rho}_{\be, h}(\Omega_{s_0})^{-, \bullet}
\]
and consists of distributions which are now extendible at $\Sigma_0$ and supported with respect to $\Sigma_e$ and $\Sigma_c$.

In order to obtain estimates on these spaces, we combine the propagation estimates in Proposition~\ref{prop: main low frequency estimate} (resp.~Proposition~\ref{prop: main low frequency adjoint estimate}) and Proposition~\ref{eq: global high frequency estimates} with energy estimates in region near the hypersurfaces $\Sigma_0, \Sigma_e$ and $\Sigma_c$.
The energy estimate near $\Sigma_0$ will take place in a region of the form
\[
	\Omega_{0, \ep}
		:= \{ s \in [\ep s_0, s_0] \} \cap \{r \in [r_e - \ep_M, r_c + \ep_M]\} \subset \Omega_{s_0}
\]
for some $\ep \in (0, 1)$.
\begin{prop}[The energy estimate near the initial hypersurface] \label{prop: energy initial hypersurface}
Let $s_0 > 0$ and $\ep \in (0, 1)$ be given.
Then there is a $e_0 > 0$ such that
\[
	\norm{u}_{H^1_h(\Omega_{0, \ep})^{\bullet, -}}
		\lesssim_e \frac1h \norm{\tP_e u}_{L^2(\Omega_{0, \ep})},
\]
and 
\[
	\norm{u}_{H^1_h(\Omega_{0, \ep})^{-, \bullet}}
		\lesssim_e \frac1h \norm{\tP_e^{*\mB} u}_{L^2(\Omega_{0, \ep})},
\]
for all $e \in (0, e_0]$, $u \in \dot C^\infty(\Omega_{0, \ep})$ and $h \in (0, h_0)$, for some $h_0$ depending on $e$.
\end{prop}

\noindent
The energy estimate near $\Sigma_e$ and $\Sigma_c$ will take place in regions of the form
\begin{align*}
	\Omega_{e, \ep}
		:= & \ \{ s \in [0, s_0] \} \cap \{r \in [r_e - \ep_M, r_e - \ep \cdot \ep_M]\} \subset \Omega_{s_0} \\
	\Omega_{c, \ep}
		:= & \ \{ s \in [0, s_0] \} \cap \{r \in [r_c + \ep \cdot \ep_M, r_c + \ep_M]\} \subset \Omega_{s_0}
\end{align*}
for some $\ep \in (0, 1)$.

\begin{prop}[The energy estimate near the final hypersurfaces] \label{prop: energy final hypersurface}
Let $\ep \in (0, 1)$ and $\rho \in \R$ be given.
Then there is a $e_0 > 0$ such that
\[
	\norm{u}_{H^{1, \rho}_{\be, h}(\Omega_{e/c, \ep})^{\bullet, -}}
		\lesssim_e \frac1h \norm{\tP_e u}_{H^{0, \rho}_{\be, h}(\Omega_{e/c, \ep})^{\bullet, -}},
\]
and 
\[
	\norm{u}_{H^{1, \rho}_h(\Omega_{e/c, \ep})^{-, \bullet}}
		\lesssim_e \frac1h \norm{\tP_e^{*\mB} u}_{H^{0, \rho}_{\be, h}(\Omega_{e/c, \ep})^{\bullet, -}},
\]
for all $e \in (0, e_0]$, $u \in \dot C^\infty(\Omega_{e/c, \ep})$ and $h \in (0, h_0)$, for some $h_0$ depending on $e$.
\end{prop}

\noindent
As usual, we will only prove the estimate for $\tP_e$, as the estimate for $\tP_e^{*\mB}$ will follow similarly.
Further, we will only discuss the proof of Proposition~\ref{prop: energy initial hypersurface}, as the proof of Proposition~\ref{prop: energy final hypersurface} is analogous.
We choose
\[
	\chi \tLop_e
\]
as a multiplier, where $\chi: \Omega \to [0, \infty)$, chosen later in~\eqref{eq: next chi}, and note that
\begin{equation} \label{eq: commutator equality energy estimate}
\begin{split}
	\frac 1h \Im \ldr{\tP_e u, \chi \tLop_e u}
		&= \frac1{2 i h} \ldr{ \left( \tLop_e ^{*\mB} \chi \tBox_e - \tBox_e^{*\mB} \chi \tLop_e \right) u,  u} + \frac1h \norm{\sqrt{\chi} \tLop_e u}^2_{L^2} \\
		&\qquad + h \Im \ldr{\tilde \Ric_g^\sharp u, \chi \tLop_e u }.
\end{split}
\end{equation}
The last term will be irrelevant in the estimate, since it has an extra factor of $h^2$ compared to the other terms.
The main complication is indeed to analyse the commutator term on the right-hand-side.
For $z \in M$ and $\a \in \bT_z M$, we define the \emph{energy-momentum-tensor}
\[
	\T_{\a}(\xi, \eta)
		:= G(\a, \xi) G(\c, \eta) + G(\a, \eta) G(\c, \xi) - G(\a, \c) G(\xi, \eta),
\]
for any $\xi, \eta \in \bT_z M$.

\begin{lemma}[The commutator term]
Fix a b-connection $\n$ on $\bT M$ such that $\n \mB = 0$.
There are smooth bundle maps $\Rop_{j, e}$, for $j = 1, \hdots, 8$, smoothly depending on $e > 0$, such that
\begin{equation} \label{eq: commutator energy estimate}
\begin{split}
	\frac1{2 i h} 
		& \left( \tLop_e ^{*\mB} \chi \tBox_e - \tBox_e^{*\mB} \chi \tLop_e \right) \\
		& = \left( h \n \right)^{*\mB} \left( \frac12 \left( \id + \frac{\c \otimes \c^\sharp }{G(\c, \c)} \right) \T_{\md \chi} + \sqrt{2 e} \cdot \Rop_{1, e}(\md \chi) + \chi \Rop_{2,e} \right)\left( h \n \right) \\
		&\qquad + h \left( \Rop_{3,e}(\md \chi) + \chi \Rop_{4,e} \right) (h \n) + h (h \n)^{*\mB} \left( \Rop_{5,e}(\md \chi) + \chi \Rop_{6,e} \right) \\
		&\qquad + h^2 \left( \Rop_{7,e}(\md \chi) + \chi \Rop_{8,e} \right)
\end{split}
\end{equation}
for any $h > 0$.
Moreover, $\Rop_{1,e}$ is uniformly bounded in $e > 0$.
Here, $\n^{*\mB}$ is the formal dual of $\n$ with respect to $\mB$, given by
\[
	\n^{*\mB} u(\cdot)
		= - G^{\a \b}\n_{e_\a} u(e_\b, \cdot),
\]
with respect to a frame $e_0, \hdots, e_n$.
\end{lemma}
\begin{proof}
We begin by comparing the semiclassical principal symbol of both sides in~\eqref{eq: commutator energy estimate}.
The principal symbol of the right-hand-side is given by
\[
	\left( \id + \frac{\c \otimes \c^\sharp }{G(\c, \c)} \right) \T_{\md \chi} + \sqrt{2 e} \cdot \Rop_{1, e}(\md \chi) + \chi \Rop_{2,e}.
\]
To compute the principal symbol of the left-hand-side, we first note the identity
\begin{align*}
	\tLop_e^{*\mB} \chi \tBox_e - \tBox_e^{*\mB} \chi \tLop_e
		&= \left( \tLop_e^{*\mB} - \tLop_e \right) \chi \tBox_e - \left( \tBox_e^{*\mB} - \tBox_e \right) \chi \tLop_e \\
		&\qquad + [\tLop_e, \chi] \tBox_e + \chi[\tLop_e, \tBox_e] + [\chi, \tBox_e] \tLop_e.
\end{align*}
The semiclassical principal symbol of the commutator term is therefore given by
\begin{align*}
	\s 
		\left( \frac1{hi} \left( \tLop_e^{*\mB} \chi \tBox_e - \tBox_e^{*\mB} \chi \tLop_e \right) \right)
		&= \chi \be_e + \left( \H_G(\chi) \tl_e - \tl_e \left( \md \chi \right) G \right),
\end{align*}
where 
\[
	\be_e
		:= \s \left( \frac1{hi} \left( \left( \tLop_e^{*\mB} - \tLop_e \right) \tP_e - \left( \tBox_e^{*\mB} - \tBox_e \right) \tLop_e + [\tLop_e, \tBox_e] \right) \right).
\]
Furthermore, we compute
\begin{align*}
	\H_G 
		& (\chi)(\xi) \tl_e(\xi) - \tl_e \left( \md \chi \right) G(\xi, \xi) \\
		&= \sum_{\a = 0}^n G^{\a \b} (\d_\a \chi) \xi_\b \tl_e(\xi) + \sum_{\b = 0}^n G^{\a \b} \xi_\a (\d_\b \chi) \tl_e(\xi) - \tl_e(\md \chi) G(\xi, \xi) \\
		&= G(\md \chi, \xi) \tl_e(\xi) + G(\xi, \md \chi) \tl_e(\xi) - \tl_e(\md \chi) G(\xi, \xi).
\end{align*}
By Proposition~\ref{prop: tilde L e properties}, we recall that
\[
	\tl_e(\xi)
		:= \tl_0(\xi) + \sqrt{2e} \cdot \r_e(\xi),
\]
where
\begin{align*}
	\tl_0(\xi)
		&= G(\c, \xi) \left( \id + \frac{\c \otimes \c^\sharp}{G(\c, \c)} \right), \\
	\r_e(\xi) \o
		&= G(\xi, \o) \c - G(\o, \c) \xi - \sqrt{2 e} \frac{G(\o, \c) G(\xi, \c)}{G(\c, \c)} \c.
\end{align*}
In conclusion, we have 
\begin{align*}
	\H_G 
		& (\chi)(\xi) \tl_e(\xi) - \tl_e \left( \md \chi \right) G(\xi, \xi) \\*
		& = \left( 2 G(\md \chi, \xi) G(\c, \xi) - G(\md \chi, \c) G(\xi, \xi) \right) \left( \id + \frac{\c \otimes \c^\sharp}{G(\c, \c)} \right) \\*
		&\qquad + \sqrt{2 e} \left( \H_G (\chi)(\xi) \r_e(\xi) - \r_e \left( \md \chi \right) G(\xi, \xi) \right)  \\
		&= \left( \id + \frac{\c \otimes \c^\sharp}{G(\c, \c)} \right) \T_{\md \chi}(\xi, \xi) + \sqrt{2 e} \cdot \q_e(\md \chi, \xi, \xi),
\end{align*}
where $\q_e$ is a polynomial of degree $3$ in $\xi$.
This completes the computation of the principal symbol.

It remains to check that all coefficients on the right-hand-side of~\eqref{eq: commutator energy estimate} indeed do not depend on second derivatives of $\chi$ (higher derivatives will clearly not show up).
For this, note that any semiclassical differential operator $Q$ of order $k$ can be written as
\[
	\Qop 
		= \sum_{j = 0}^k A_{j,h}((h \n)^j)
\]
where $A_{j,h}$ are appropriate homomorphism fields which are smooth in $h \geq 0$.
It follows that 
\[
	\Qop^{*\mB}
		= \sum_{j = 0}^k (h \n^{*\mB})^j A_{j,h}^{*\mB}
		= (h \n)^{*\mB} \Sop_1 + \Sop_2,
\]
where $\Sop_1$ and $\Sop_2$ are semiclassical differential operators of order $k-1$.
It therefore follows that the left-hand-side of~\eqref{eq: commutator energy estimate} is of the form
\[
	\left( A_0 + (h \n)^{*\mB} A_1 \right) \mathcal C(\chi) \left( B_0 + B_1 (h \n) \right),
\]
where $A_0, A_1, B_0, B_1$ are homomorphism fields and $\mathcal C(\chi)$ is some expression given by $\chi, h \n$ and commutators $[\chi, h \n]$.
Hence at most first derivatives of $\chi$ will appear in the coefficients.
\end{proof}

The next lemma is the standard coercivity of the energy-momentum tensor in this setting.

\begin{lemma} \label{le: positivity energy-momentum tensor}
If $z \in M$ and $\a \in \bT_z M$ is future directed timelike, then 
\[
	\T_{\a}(\xi, \eta)
		:= G(\a, \xi) G(\c, \eta) + G(\a, \eta) G(\c, \xi) - G(\a, \c) G(\xi, \eta)
\]
is a positive definite quadratic form on $\bT_z M$.
If, moreover, $E_z$ is a complex vector space with definite Hermitian form $\ldr{-,-}_{E_z}$, then
\[
	\T_{\a} \otimes \ldr{-,-}_{E_z} \left( \xi \otimes e, \eta \otimes f \right)
		:= \T_{\a}(\xi, \eta) \ldr{e,f}_{E_z}
\]
and extended bilinearly, defines a positive definite Hermitian form on $\bT_z M \otimes E_z$.
\end{lemma}

\begin{remark} \label{rmk: one form causality test}
A causal b-one-form $\o$ is future directed by definition precisely when $\o^\sharp$ is future directed. 
Moreover, $\c^\sharp$ and thereby $\c$ are future directed by construction.
A causal b-one-form $\o$ is therefore future directed precisely when
\[
	\o(\c^\sharp)
		= G(\o, \c)
		< 0.
\]
\end{remark}

\begin{proof}[Proof of Lemma~\ref{le: positivity energy-momentum tensor}]
Note first that $\T_{\a}(\xi, \eta)$ is symmetric in $\xi$ and $\eta$.
We write $\xi = \xi_0 \c + \xi^\perp$, where $\xi_0 \in \R$ and $\xi^\perp \perp \c$.
It follows that
\begin{align}
	\T_{\a}(\xi, \xi)
		&= 2 G(\a, \xi) G(\c, \xi) - G(\a, \c) G(\xi, \xi) \nonumber \\
		&= 2 \xi_0^2 G(\a, \c) G(\c, \c) + 2\xi_0 G(\a, \xi^\perp) G(\c, \c) \nonumber \\
		&\qquad - G(\a, \c) \left( \xi_0^2 G(\c, \c) + G(\xi^\perp, \xi^\perp) \right) \nonumber \\
		&= - G(\a, \c) \left( - \xi_0^2 G(\c, \c) + G(\xi^\perp, \xi^\perp) \right) + 2 \xi_0 G(\a, \xi^\perp) G(\c, \c). \label{eq: energy-momentum equality}
\end{align}
If $\xi^\perp = 0$, the statement is now immediate. 
We may therefore assume that $\xi^\perp \neq 0$ from now on.
Since $\a$ is future directed timelike and
\[
	\frac{\c}{\sqrt{-G(\c, \c)}} \pm \frac{\xi^\perp}{\sqrt{G(\xi^\perp, \xi^\perp)}}
\]
is future directed lightlike, it follows that 
\begin{align*}
	0
		& > G\left( \a, \frac{\c}{\sqrt{-G(\c, \c)}} \pm \frac{\xi^\perp}{\sqrt{G(\xi^\perp, \xi^\perp)}} \right) 
		= \frac{G( \a, \c)}{\sqrt{-G(\c, \c)}} \pm \frac{G(\a, \xi^\perp)}{\sqrt{G(\xi^\perp, \xi^\perp)}},
\end{align*}
i.e., 
\begin{align*}
	\abs{G(\a, \xi^\perp)}
		& < \sqrt{G(\xi^\perp, \xi^\perp)}\frac{-G( \a, \c)}{\sqrt{-G(\c, \c)}}.
\end{align*}
We therefore conclude that
\begin{align*}
	\T_{\a}(\xi, \xi)
		&\geq - G(\a, \c) \left( - \xi_0^2 G(\c, \c) + G(\xi^\perp, \xi^\perp) \right) - 2 \abs{\xi_0} \abs{G(\a, \xi^\perp)} \abs{G(\c, \c)} \\
		&> - G(\a, \c) \left( - \xi_0^2 G(\c, \c) + G(\xi^\perp, \xi^\perp) \right) \\
		&\qquad - 2 \abs{\xi_0} \sqrt{G(\xi^\perp, \xi^\perp)}\frac{-G( \a, \c)}{\sqrt{-G(\c, \c)}} \left( - G(\c, \c) \right) \\
		&= - G(\a, \c) \left( \abs{\xi_0} \sqrt{-G(\c, \c)} -  \sqrt{G(\xi^\perp, \xi^\perp)} \right)^2 \\
		& > 0,
\end{align*}
since $G(\a, \c) < 0$, unless $\abs{\xi_0} \sqrt{-G(\c, \c)} =  \sqrt{G(\xi^\perp, \xi^\perp)}$.
But in this case, $\xi$ is lightlike, and we get
\[
	\T_{\a}(\xi, \xi)
		= 2 G(\a, \xi) G(\c, \xi) 
		> 0,
\]
since $G(\a, \xi)$ and $G(\c, \xi)$ have the same sign.
Since the tensor product of a positive definite quadratic form and a positive definite Hermitian form is positive definite, the second statement follows immediately from the first one.
\end{proof}

We now turn to the proof of Proposition~\ref{prop: energy initial hypersurface}.
We use the notation
\[
	H(x)
		:= 
		\begin{cases}
			1, & x \geq 0, \\
			0, & x < 0.
		\end{cases}
\]
For a fixed $\ep \in (0, 1)$, define
\begin{equation} \label{eq: next chi}
	\chi
		:= \ e^{\digamma s} H(s - \ep s_0) H(s_0 - s) H(r_c + \ep_M - r) H(r - (r_e - \ep_M)),
\end{equation}
for any $\digamma > 0$.
Note that $\chi \geq 0$ everywhere and that $\chi$ has a positive bound from below on $\Omega_{0, \ep}$.
The differential is given by the distribution
\begin{equation} \label{eq: differential chi}
	\md \chi
		= \ \digamma \chi \md s + \chi \left( \de_{s - \ep s_0} - \de_{s - s_0} \right) \md s + \chi \left( \de_{r - r_e + \ep_M} - \de_{r - r_c - \ep_M} \right) \md r.
\end{equation}
In the following estimates, we use the notation $\abs{\cdot} := \sqrt{\mB(\cdot, \cdot)}$.

\begin{lemma} \label{le: applied T coercivity}
There is a $c > 0$ such that
\begin{align*}
	& \int_{\Omega_{s_0}} \T_{\md \chi}(h \n u, h \n u) \md \Vol_g \\
		&\quad \geq c \digamma \norm{ \sqrt{\chi} \abs {h \n u}}^2_{L^2}  + c \norm{\sqrt{\chi} \abs {h \n u}|_{s = s_0}}^2_{L^2} \\
		&\quad \qquad + c \norm{\sqrt{\chi} \abs {h \n u}|_{r = r_e - \ep_M}}^2_{L^2} + c \norm{\sqrt{\chi} \abs {h \n u}|_{r = r_c + \ep_M}}^2_{L^2},
\end{align*}
for any $u \in C^\infty(\Omega)$ with $u|_{s > s_0} = 0$, and for any $\digamma > 0$ and $h > 0$.
\end{lemma}

\begin{proof}
By~\eqref{eq: differential chi} and the support condition on $u$, we have
\begin{align*}
	\int_{\Omega_{s_0}} 	\T_{\md \chi}(h \n u, h \n u) \md \Vol_g
		&= \digamma \int_{\Omega_{s_0}} \chi \T_{\md s}(h \n u, h \n u) \md \Vol_g \\
		&\qquad + \int_{[r_e - \ep_M, r_c + \ep_M]} \chi \T_{\md s}(h \n u, h \n u)|_{s = \ep s_0} \md S_g \\
		&\qquad + \int_{[\ep s_0, s_0]} \chi \T_{\md r}(h \n u, h \n u)|_{r = r_e - \ep_M} \md S_g \\
		&\qquad + \int_{[\ep s_0, s_0]} \chi \T_{-\md r}(h \n u, h \n u)|_{r = r_c + \ep_M} \md S_g 
\end{align*}
where $\md S_g$ is induced volume form on the respective hypersurface by $\md \Vol_g$.
Note that $\md s$, $\md r|_{r = r_e - \ep_M}$ and $\md r|_{r = r_c + \ep_M}$ are all timelike (recall the discussion in the introduction where it was arranged that $\md t_* = - \frac{\md s}s$ is timelike).
Since
\begin{align*}
	\md s(\c^\sharp)
		&= - s 
		< 0, \\
	\md r(\c^\sharp) |_{r = r_e - \ep_M}
		&= \g f(r_e - \ep_M)
		< 0, \\
	\md r(\c^\sharp) |_{r = r_c + \ep_M}
		&= \g f(r_c + \ep_M)
		> 0,
\end{align*}
the argument in Remark~\ref{rmk: one form causality test} shows that $\md s$, $\md r|_{r = r_e - \ep_M}$ and $- \md r|_{r = r_c + \ep_M}$ are future directed timelike.
Lemma~\ref{le: positivity energy-momentum tensor} therefore implies the desired estimate.
\end{proof}

Using this lemma, we now prove the following.

\begin{lemma} \label{le: first order commutator energy estimate}
There are $c, \digamma_0, e_0 > 0$ such that for any $\digamma > \digamma_0$, $e \in (0, e_0]$,
\begin{align*}
	\frac1{2 i h} 
		&\ldr{ \left( \tLop_e ^{*\mB} (\chi) \tBox_e + \tBox_e^{*\mB} (\chi) \tLop_e \right) u,  u} \\
		&\geq c \left( \digamma - C_e \right) \norm{ \sqrt{\chi} \abs {h \n u}}^2_{L^2}  - C_e \digamma h^2 \norm{\sqrt{\chi} u}_{L^2}^2 - C_e h^2 \norm{\sqrt{\chi} u|_{s = s_0}}_{L^2}^2 \\
		&\qquad - C_e h^2 \norm{\sqrt{\chi} u|_{r = r_e - \ep_M}}_{L^2}^2 - C_e h^2 \norm{\sqrt{\chi} u|_{r = r_c + \ep_M}}_{L^2}^2,
\end{align*}
for any $u \in C^\infty(M)$ with $u|_{s > s_0} = 0$, where $C_e > 0$ depends on $e$, and any $h > 0$.
\end{lemma}

\begin{proof}
We estimate the terms in~\eqref{eq: commutator energy estimate}.
By Lemma~\ref{le: positivity energy-momentum tensor}, 
\begin{align*}
	& \ldr{\left( \frac12 \left( \id + \frac{\c \otimes \c^\sharp }{G(\c, \c)} \right) \T_{\md \chi} + \sqrt{2 e} \cdot \Rop_{1, e}(\md \chi) + \chi \Rop_{2,e} \right)\left( h \n \right) u, \left( h \n \right) u} \\
	&\quad \geq \frac12 \int_M \T_{\md \chi}(h \n u, h \n u) \md \Vol_g - C \sqrt{e} \digamma \norm{\sqrt{\chi} \abs {h \n u}}^2_{L^2} \\
	&\quad \qquad - C \sqrt{e} \norm{\sqrt{\chi} \abs {h \n u}|_{s = \ep s_0}}^2_{L^2} - C \sqrt{e} \norm{\sqrt{\chi} \abs {h \n u}|_{r = r_e - \ep_M}}^2_{L^2} \\
	&\quad \qquad - C \sqrt{e} \norm{\sqrt{\chi} \abs {h \n u}|_{r = r_c + \ep_M}}^2_{L^2} - C_e \norm{\sqrt \chi \abs{h \n u}}_{L^2}^2,
\end{align*}
for some constants $C, C_e > 0$ independent of $\digamma$ (where $C$ is independent of $e > 0$).
We continue with the remaining terms of~\eqref{eq: commutator energy estimate}.
\begin{align*}
	& h \ldr{\left( \Rop_{3,e}(\md \chi) + \chi \Rop_{4,e} \right) (h \n)u, \o} + h \ldr{\left( \Rop_{5,e}(\md \chi) + \chi \Rop_{6,e} \right) u, (h \n) u } \\
		&\qquad + h^2 \ldr{\left( \Rop_{7,e}(\md \chi) + \chi \Rop_{8,e} \right)u, u} \\
		&\geq - \de \digamma \norm{\sqrt{\chi} \abs {h \n u}}_{L^2}^2 - \digamma \frac {C_e} \de h^2 \norm{\sqrt{\chi} u}_{L^2}^2 \\
	&\quad \qquad - \de  \norm{\sqrt{\chi} \abs {h \n u}|_{s = \ep s_0}}_{L^2}^2 - \frac {C_e} \de h^2 \norm{\sqrt{\chi} u|_{s = \ep s_0}}_{L^2}^2 \\
	&\quad \qquad - \de \norm{\sqrt{\chi} \abs {h \n u}|_{r = r_e - \ep_M}}_{L^2}^2 - \frac {C_e} \de h^2 \norm{\sqrt{\chi} u|_{r = r_e - \ep_M}}_{L^2}^2 \\
	&\quad \qquad - \de \norm{\sqrt{\chi} \abs {h \n u}|_{r = r_c + \ep_M}}_{L^2}^2 - \frac {C_e} \de h^2 \norm{\sqrt{\chi} u|_{r = r_c + \ep_M}}_{L^2}^2,
\end{align*}
for a constant $C_e > 0$ and any $\de > 0$.
Applying Lemma~\ref{le: applied T coercivity} with $e_0 > 0$ and $\de > 0$ small enough and $\digamma_0$ large enough, the statement follows.
\end{proof}

By choosing $\digamma$ large enough in Lemma~\ref{le: first order commutator energy estimate}, depending on the choice of $e$, we can control any first derivatives of $u$.
The next statement will give the necessary control of the zeroth order terms.

\begin{lemma}\label{le: zeroth order commutator energy estimate}
Let $\ep \in (0, 1)$.
There are $c, \digamma_0, e_0, h_0 > 0$ such that
\begin{align*}
	\frac1h \norm{\sqrt{\chi}\tLop_e u}_{L^2}^2
		& \geq c \left( \digamma - C_e \right) \norm{\sqrt{\chi} u}_{L^2}^2 + c \norm{\sqrt{\chi} u|_{s = \ep s_0}}_{L^2}^2 \\*
		& \qquad + c \norm{\sqrt{\chi} u|_{r = r_e - \ep_M}}_{L^2}^2 \\*
		& \qquad + c \norm{\sqrt{\chi} u|_{r = r_c + \ep_M}}_{L^2}^2,
\end{align*}
for any $u \in C^\infty(M)$ with $u|_{s > s_0} = 0$, and for any $\digamma > \digamma_0$, $e \in (0, e_0]$ and $h \in (0, h_0)$.
\end{lemma}
\begin{proof}
We compute
\begin{align*}
	2 \Im \ldr{\tLop_e u, \chi u}
		&= - i \ldr{  \left( \chi \tLop_e - \tLop_e^{*\mB} \chi\right) u, u} \\
		&= \ldr{ i [\tLop_e, \chi] u, u} - \ldr{i \left( \tLop_e - \tLop_e^{*\mB} \right) \chi u, u} \\
		&= \ldr{ \tl_e(\md \chi) u, u} - i \ldr{ \left( \tLop_e - \tLop_e^{*\mB} \right) \chi u, u} \\
		&= \ldr{ \md \chi(\c^\sharp) \left( \id + \frac{\c \otimes \c^\sharp}{G(\c, \c)} \right) u, u} + \sqrt{2e}\ldr{ \r_e(\md \chi) u, u} \\
		&\qquad - \ldr{ i \left( \tLop_e - \tLop_e^{*\mB} \right) \chi u, u},
\end{align*}
where we recall that $\r_e$ is bounded uniformly in $e \in [0, 1]$.
Moreover, where $u \neq 0$, we have
\begin{align*}
	\md \chi(\c^\sharp)
		= & \ - \digamma \chi s - \chi s \de_{s - \ep s_0} \\
		& \ + \a \chi f(r - r_e + \ep_M) \de_{r - r_e + \ep_M} - \a \chi f(r - r_c - \ep_M) \de_{r - r_c - \ep_M}.
\end{align*}
We thus get the estimate
\begin{align*}
	\frac1h 
		& \norm{\sqrt{\chi}\tLop_e u}_{L^2}^2 + h \norm{\sqrt{\chi} u}_{L^2}^2 \\
		& \geq 2 \Im \ldr{\tLop_e u, \chi u} \\
		& \geq \ep s_0 \digamma \norm{\sqrt{\chi} u}_{L^2}^2 + \ep s_0 \norm{\sqrt{\chi} u|_{s = \ep s_0}}_{L^2}^2 \\*
		& \qquad + \a \abs{f(r - r_e + \ep_M)} \norm{\sqrt{\chi} u|_{r = r_e - \ep_M}}_{L^2}^2 \\*
		& \qquad + \a \abs{f(r - r_c - \ep_M)} \norm{\sqrt{\chi} u|_{r = r_c + \ep_M}}_{L^2}^2 \\
		& \qquad - \sqrt{e} C \digamma \norm{\sqrt{\chi} u}_{L^2}^2 - \sqrt{e} C\norm{\sqrt{\chi} u|_{s = s_0}}_{L^2}^2 \\*
		& \qquad - \sqrt{e} C \norm{\sqrt{\chi} u|_{r = r_e - \ep_M}}_{L^2}^2 - \sqrt{e} C \norm{\sqrt{\chi} u|_{r = r_c + \ep_M}}_{L^2}^2 \\
		& \qquad - C_e \norm{\sqrt{\chi} u}_{L^2}^2,
\end{align*}
for some constants $C, C_e > 0$ (where $C$ is independent of $e$).
This proves the statement if $\digamma_0 > 0$ is large enough and $e_0, h_0 > 0$ are small enough.
\end{proof}

We may now finally turn to the proof of the energy estimate near the initial hypersurface.

\begin{proof}[Proof of Proposition~\ref{prop: energy initial hypersurface}]

By combining~\eqref{eq: commutator equality energy estimate} with Lemma~\ref{le: first order commutator energy estimate} and Lemma~\ref{le: zeroth order commutator energy estimate}, we get the estimate
\begin{align*}
	\frac2h \norm{\sqrt{\chi} \tP_e u}^2
		& \geq \frac1 h \Im \ldr{\tP_e u, \chi \tLop_e u} - \frac1{2h} \norm{\sqrt \chi \tLop_e u}^2 \\
		& = \frac1{2 i h} \ldr{ \left( \tLop_e ^{*\mB} \chi \tBox_e - \tBox_e^{*\mB} \chi \tLop_e \right) u,  u} + \frac1{2h} \norm{\sqrt{\chi} \tLop_e u}^2_{L^2} \\*
		& \qquad + h \Im \ldr{\tilde \Ric_g^\sharp u, \chi \tLop_e u } \\*
		& \geq \frac1{2 i h} \ldr{ \left( \tLop_e ^{*\mB} \chi \tBox_e - \tBox_e^{*\mB} \chi \tLop_e \right) u,  u} + \left( \frac1{2h} - 1 \right) \norm{\sqrt{\chi} \tLop_e u}^2_{L^2} \\*
		& \qquad - h C \norm{\sqrt{\chi}u}_{L^2}^2 \\
		& \geq c \left( \digamma - C_e \right) \norm{ \sqrt{\chi} \abs {h \n u}}^2_{L^2} - C_e \digamma h^2 \norm{\sqrt{\chi} \abs {u}}_{L^2}^2 \\*
		& \qquad - C_e h^2 \norm{\sqrt{\chi} \abs {u}|_{s = s_0}}_{L^2}^2 - C_e h^2 \norm{\sqrt{\chi} \abs {u}|_{r = r_e - \ep_M}}_{L^2}^2 \\*
		& \qquad - C_e h^2 \norm{\sqrt{\chi} \abs {u}|_{r = r_c + \ep_M}}_{L^2}^2 \\*
		& \qquad + c \left( \digamma - C_e \right) \norm{\sqrt{\chi} u}_{L^2}^2 + c \norm{\sqrt{\chi} u|_{s = \ep s_0}}_{L^2}^2 \\*
		& \qquad + c \norm{\sqrt{\chi} u|_{r = r_e - \ep_M}}_{L^2}^2 + c \norm{\sqrt{\chi} u|_{r = r_c + \ep_M}}_{L^2}^2 \\*
		& \qquad - h C \norm{\sqrt{\chi}u}_{L^2}^2 \\
		& \geq c \left( \digamma - C_e \right) \left( \norm{ \sqrt{\chi} \abs {h \n u}}^2_{L^2} + \norm{\sqrt{\chi} u}_{L^2}^2 \right)
\end{align*}
if $e_0 > 0$ is small enough and $h_0$ is small enough depending on the choice of $e$.
Fixing $e > 0$, we may finally choose $\digamma$ large enough to get the desired estimate.
\end{proof}

\subsection{Global estimates and conclusion}

We now turn to the proof of the constraint damping.

\begin{proof}[Proof of Theorem~\ref{thm: stable constraint propagation}]

Combining the low frequency estimate, Proposition~\ref{prop: main low frequency estimate}, the high frequency estimate, Proposition~\ref{eq: global high frequency estimates}, the ellipticity at finite frequency, Proposition~\ref{prop: ellipticity finite frequencies}, and the energy estimate near the initial and final hypersurfaces, Proposition~\ref{prop: energy initial hypersurface} and Proposition~\ref{prop: energy final hypersurface}, we conclude that there are $\rho_0, e_0 > 0$ such that
\[
	\norm{u}_{H^{1, \rho}_{\be, h} (\Omega_{s_0})^{\bullet, -}}
		\lesssim_e \frac 1h \norm{\tP_e u}_{L^2(\Omega_{s_0})},
\]
for all $\rho \in [0, \rho_0]$, $e \in (0, e_0]$, $u \in C^\infty(\Omega)$ with $u|_{s > s_0} = 0$ and $h \in (0, h_0)$, where $h_0$ depends on $e$.
Recalling that 
\[
	\tP_e
		= b_{(2e)^{1/2}} \P_e b_{(2e)^{-1/2}},
\]
we conclude that there are $\rho_0, e_0$ such that
\[
	\norm{u}_{H^{1, \rho}_{\be, h} (M_{s_0})^{\bullet, -}}
		\lesssim_e \frac 1h \norm{\P_e u}_{H^{0, \rho}_{\be, h} (M_{s_0})^{\bullet, -}},
\]
for all $\rho \in [0, \rho_0]$, $e \in (0, e_0]$, $u \in C^\infty(\Omega)$ with $u|_{s > s_0} = 0$ and $h \in (0, h_0)$, where $h_0$ depends on $e$.
The same propositions, with the exception of Proposition~\ref{prop: main low frequency estimate}, which we replace with Proposition~\ref{prop: main low frequency adjoint estimate}, similarly imply that there are $\rho_0, e_0 > 0$ such that
\[
	\norm{u}_{H^{1, - \rho}_{\be, h} (\Omega_{s_0})^{-, \bullet}}
		\lesssim_e \frac 1h \norm{\P_e^{*\mB} u}_{H^{0, - \rho}_{\be, h} (\Omega_{s_0})^{-, \bullet}},
\]
for all $\rho \in [0, \rho_0]$, $e \in (0, e_0]$, $u \in C^\infty(\Omega)$ with $u|_{r \notin [r_e - \ep_M, r_c + \ep_M]} = 0$ and $h \in (0, h_0)$, where $h_0$ depends on $e$.

This implies the solvability of the adjoint operator to $\P_e^{*\mB}$, i.e., we obtain a forward solution operator 
\[
	\P_e^{-1}
		: H^{-1, \rho}_{\be, h} (\Omega_{s_0})^{\bullet, -} \to H^{0, \rho}_{\be, h} (\Omega_{s_0})^{\bullet, -}.
\]
for any $\rho \in [0, \rho_0]$, $e \in (0, e_0]$, $h \in (0, h_0)$.
Let us now assume that there is a resonant state 
\[
	u(s, r, \psi, \theta)
		= s^{i \s} v(r, \psi, \theta)
\]
such that $\P_e u = 0$ with $\Im(\s) > -\rho$.
Let $\phi \in C^\infty(\R)$ be such that $\phi(s) = 1$ for $s < \frac{s_0}4$ and $\phi(s) = 0$ for $s > \frac{s_0}2$, and consider
\[
	\tilde u(s, r, \psi, \theta)
		:= \phi(s) u(s, r, \psi, \theta).
\]
It follows that $\P_e \tilde u$ vanishes for $s \notin [s_0/4, s_0/2]$.
In particular, 
\[
	\P_e \tilde u \in H^{-1, \rho}_{\be, h} (\Omega_{s_0})^{\bullet, -}.
\]
Applying the solution operator $\P_e^{-1}$ to $\P_e \tilde u$ implies, since the Cauchy problem for linear wave equations has a unique solution in any fixed real Sobolev regularity (this is essentially due to an energy estimate like Proposition~\ref{prop: energy initial hypersurface} improved to arbitrary regularity by, e.g., standard propagation of singularities and considering the estimate for  adjoint estimate), that 
\[
	\tilde u \in H^{0, \rho}_{\be, h} (\Omega_{s_0})^{\bullet, -}.
\]
But this is a contradiction, since
\[
	\tilde u(s, r, \psi, \theta)|_{s < s_0/4}
		= u(s, r, \psi, \theta)
		= s^{i \s} v(r, \psi, \theta),
\]
which would imply that
\[
	\tilde u \notin H^{0, \rho}_{\be, h} (\Omega_{s_0})^{\bullet, -},
\]
since $\Im(\s) > -\rho$.
This finishes the proof.
\end{proof}

\end{sloppypar}

\begin{bibdiv}
  \begin{biblist}

\bib{AHW24}{article}{
  author={Andersson, Lars},
  author={H{\"a}fner, Dietrich},
  author={Whiting, Bernard},
title={Mode analysis for the linearized {E}instein equations on the {K}err metric: the large $a$ case},
journal={ J. Eur. Math. Soc., DOI 10.4171/JEMS/1544},
year={2024}
}

\bib{B2007}{book}{
	author={Besse, A.\ L.},
	title={Einstein manifolds},
	series={Classics in Mathematics},
	note={Reprint of the 1987 edition},
	publisher={Springer-Verlag, Berlin},
	date={2008},
	pages={xii+516},
}

\bib{BZ09}{book}{
    author={Bieri, Lydia},
    author={Zipser, Nina},
title={Extensions of the stability theorem of the {M}inkowski space in
  general relativity},
series={AMS/IP Studies in Advanced Mathematics},
volume={45},
publisher={American Mathematical Society},
year={2009}}

\bib{BFHR1999}{article}{
   author={Brodbeck, Othmar},
   author={Frittelli, Simonetta},
   author={H\"ubner, Peter},
   author={Reula, Oscar A.},
   title={Einstein's equations with asymptotically stable constraint
   propagation},
   journal={J. Math. Phys.},
   volume={40},
   date={1999},
   number={2},
   pages={909--923},
}


\bib{CK93}{book}{
  AUTHOR = {Christodoulou, Demetrios},
  author={Klainerman, Sergiu},
     TITLE = {The global nonlinear stability of the {M}inkowski space},
    SERIES = {Princeton Mathematical Series},
    VOLUME = {41},
 PUBLISHER = {Princeton University Press},
   ADDRESS = {Princeton, NJ},
      YEAR = {1993},
     PAGES = {x+514},
      ISBN = {0-691-08777-6}
}

\bib{DHR19T}{article}{
   author={Dafermos, Mihalis},
   author={Holzegel, Gustav},
   author={Rodnianski, Igor},
   title={Boundedness and decay for the Teukolsky equation on Kerr
   spacetimes I: The case $|a|\ll M$},
   journal={Ann. PDE},
   volume={5},
   date={2019},
   number={1},
   pages={Paper No. 2, 118},
}

\bib{DHR19S}{article}{
   author={Dafermos, Mihalis},
   author={Holzegel, Gustav},
   author={Rodnianski, Igor},
   title={The linear stability of the Schwarzschild solution to
   gravitational perturbations},
   journal={Acta Math.},
   volume={222},
   date={2019},
   number={1},
   pages={1--214},
}


\bib{DHRT21}{article}{
         AUTHOR = {Dafermos, Mihalis},
  author={Holzegel, Gustav},
  author={Rodnianski, Igor},
  author={Taylor, Martin},
title={The non-linear stability of the {S}chwarzschild family of black holes},
journal={Preprint, arXiv:2104.08222},
year={2021}
}


\bib{F-B1952}{article}{
   author={Four\`es-Bruhat, Y.},
   title={Th\'eor\`eme d'existence pour certains syst\`emes d'\'equations aux
   d\'eriv\'ees partielles non lin\'eaires},
   language={French},
   journal={Acta Math.},
   volume={88},
   date={1952},
   pages={141--225},
}


\bib{F1996}{article}{
   author={Friedrich, Helmut},
   title={Hyperbolic reductions for Einstein's equations},
   journal={Classical Quantum Gravity},
   volume={13},
   date={1996},
   number={6},
   pages={1451--1469},
}

\bib{Dya2015}{article}{
	author = {Dyatlov, Semyon},
    title = {Asymptotics of linear waves and resonances with applications to black holes},
	journal = {Comm. Math. Phys.},
	volume = {335},
	year = {2015},
	number = {3},
	pages = {1445--1485},
}


\bib{Dya2015b}{article}{
	author = {Dyatlov, Semyon},
	title = {Resonance projectors and asymptotics for {$r$}-normally hyperbolic trapped sets},
	journal = {J. Amer. Math. Soc.},
	volume = {28},
	year = {2015},
	number = {2},
	pages = {311--381},
}

\bib{Dya2016}{article}{
	author = {Dyatlov, Semyon},
	title = {Spectral gaps for normally hyperbolic trapping},
	journal = {Ann. Inst. Fourier (Grenoble)},
	volume = {66},
	year = {2016},
	number = {1},
	pages = {55--82},
}

\bib{FS24}{article}{
  author={Fournodavlos, G.},
  author={Schlue, V.},
title={Stability of the expanding region of {K}err de {S}itter spacetimes},
journal={Preprint, arXiv:2408.02596},
year={2024}
}

\bib{Fr86}{article}{
author={Friedrich, Helmut},
title={On the existence of {$n$}-geodesically complete or future complete
  solutions of {E}instein's field equations with smooth asymptotic structure},
journal={Comm. Math. Phys.},
volume={107},
number={4},
pages={587--609},
year={1986}
}


\bib{GKS24}{article}{
   author={Giorgi, Elena},
   author={Klainerman, Sergiu},
   author={Szeftel, J\'er\'emie},
   title={Wave equations estimates and the nonlinear stability of slowly
   rotating Kerr black holes},
   journal={Pure Appl. Math. Q.},
   volume={20},
   date={2024},
   number={7},
   pages={2865--3849},
}


\bib{GL91}{article}{
   author={Graham, C. Robin},
   author={Lee, John M.},
   title={Einstein metrics with prescribed conformal infinity on the ball},
   journal={Adv. Math.},
   volume={87},
   date={1991},
   number={2},
   pages={186--225},
}

\bib{Gund05}{article}{
  author={Gundlach, Carsten},
  author={Calabrese, Gioel},
  author={Hinder, Ian},
  author={Martin-Garcia, Jose~M.},
title={{C}onstraint damping in the ${Z}4$ formulation and harmonic gauge},
journal={Classical and Quantum Gravity},
volume={22},
number={17},
pages={3767},
year={2005}
}

\bib{HHV25}{article}{
 AUTHOR = {H\"{a}fner, Dietrich},
 author={Hintz, Peter},
 author={Vasy, Andr\'{a}s},
 title={Linear stability of Kerr black holes in the full subextremal range },
 journal={Preprint, arXiv:2506.21183},
 year={2025}
}

\bib{HV2013}{article}{
   author={Hintz, Peter},
   author={Vasy, Andr{\'a}s},
   title={Non-trapping estimates near normally hyperbolic trapping},
   journal={Math. Res. Lett.},
   volume={21},
   date={2014},
   number={6},
   pages={1277--1304},
}



\bib{HV2015}{article}{
	author = {Hintz, Peter},
	author={Vasy, Andr{\'a}s},
	title = {Semilinear wave equations on asymptotically de {S}itter,
              {K}err--de {S}itter and {M}inkowski spacetimes},
	journal = {Anal. PDE},
	volume = {8},
	year = {2015},
	number = {8},
	pages = {1807--1890},
}

\bib{HV2017_IMRN}{article}{
   author={Hintz, Peter},
   author={Vasy, Andr{\'a}s},
   title={Global analysis of quasilinear wave equations on asymptotically Kerr–de~Sitter spaces},
   journal={Int. Math. Res. Not.},
   volume={17},
   date={2016},
   pages={5355–5426},
}

\bib{HV2017}{article}{
   author={Hintz, Peter},
   author={Vasy, Andr{\'a}s},
   title={Analysis of linear waves near the Cauchy horizon of cosmological
   black holes},
   journal={J. Math. Phys.},
   volume={58},
   date={2017},
   number={8},
   pages={081509, 45},
}


\bib{HV2018}{article}{
	author = {Hintz, Peter},
   	author={Vasy, Andr{\'a}s},
	title = {The global non-linear stability of the {K}err--de {S}itter family of black holes},
    journal = {Acta Math.},
	volume = {220},
    year = {2018},
	number = {1},
	pages = {1--206},
}

\bib{HV20}{article}{
   author={Hintz, Peter},
   author={Vasy, Andr\'as},
   title={Stability of Minkowski space and polyhomogeneity of the metric},
   journal={Ann. PDE},
   volume={6},
   date={2020},
   number={1},
   pages={Paper No. 2, 146},
}

\bib{H2015}{article}{
	author = {Hintz, Peter},
	title = {Resonance expansions for tensor-valued waves on asymptotically Kerr--de~Sitter spaces},
	journal = {Journal of Spectral Theory},
	volume = {7},
	year = {2017},
	pages = {519–557},
}

\bib{H2021}{article}{
	author = {Hintz, Peter},
	title ={Normally hyperbolic trapping on asymptotically stationary spacetimes},
	journal={Probability and Mathematical Physics},
	volume={2},
	number={1},
	pages={71–126},
	year={2021},
}

\bib{H2024}{article}{
	author = {Hintz, Peter},
	title ={Gluing small black holes along timelike geodesics II: uniform analysis on glued spacetimes},
	journal={Preprint, arxiv: 2408.06712},
    year={2024}
}

\bib{Hin2024M}{article}{
author={Hintz, Peter},
title={Mode stability and shallow quasinormal modes of {K}err-de {S}itter black holes away from extremality},
journal={J. Eur. Math. Soc. (2024), published online first},
volume={DOI 10.4171/JEMS/1463},
year={2024}
}

\bib{Hin23E}{article}{
   author={Hintz, Peter},
   title={Exterior stability of Minkowski space in generalized harmonic
   gauge},
   journal={Arch. Ration. Mech. Anal.},
   volume={247},
   date={2023},
   number={5},
   pages={Paper No. 99, 45},
}

\bib{HV24}{article}{
  author={Hintz, Peter},
  author={Vasy, Andras},
  title={Stability of the expanding region of Kerr--de~Sitter
    spacetimes and smoothness at the conformal boundary},
  journal={Preprint, arXiv:2409.1546},
  year={2024}
  }

\bib{H2009}{book}{
   author={H\"ormander, Lars},
   title={The analysis of linear partial differential operators},
   volume = {1-4},
   publisher={Springer-Verlag},
   date={1983},
}

\bib{J2025}{article}{
	author={Jia, Qiuye},
	title={Propagation of Singularities with Normally Hyperbolic Trapping},
    journal = {Ann.~Henri Poincar\'e},
    year = {2025},
  }


\bib{KS20P}{book}{
   author={Klainerman, Sergiu},
   author={Szeftel, J\'er\'emie},
   title={Global nonlinear stability of Schwarzschild spacetime under
   polarized perturbations},
   series={Annals of Mathematics Studies},
   volume={210},
   publisher={Princeton University Press, Princeton, NJ},
   date={2020},
   pages={xi+840},
}


\bib{KS22C}{article}{
   author={Klainerman, Sergiu},
   author={Szeftel, J\'er\'emie},
   title={Construction of GCM spheres in perturbations of Kerr},
   journal={Ann. PDE},
   volume={8},
   date={2022},
   number={2},
   pages={Paper No. 17, 153},
}


\bib{KS22E}{article}{
   author={Klainerman, Sergiu},
   author={Szeftel, J\'er\'emie},
   title={Effective results on uniformization and intrinsic GCM spheres in
   perturbations of Kerr},
   note={With an appendix by Camillo De Lellis},
   journal={Ann. PDE},
   volume={8},
   date={2022},
   number={2},
   pages={Paper No. 18, 89},
}

\bib{KS23K}{article}{
   author={Klainerman, Sergiu},
   author={Szeftel, J\'er\'emie},
   title={Kerr stability for small angular momentum},
   journal={Pure Appl. Math. Q.},
   volume={19},
   date={2023},
   number={3},
   pages={791--1678},
}


\bib{LIG2016}{article}{
    author={LIGO Scientific Collaboration},
    author={Virgo Collaboration},
title={Observation of Gravitational Waves from a Binary Black Hole Merger},
journal={Phys. Rev. Lett.},
volume={116:061102},
year={2016}
}


\bib{LR10}{article}{
   author={Lindblad, Hans},
   author={Rodnianski, Igor},
   title={The global stability of Minkowski space-time in harmonic gauge},
   journal={Ann. of Math. (2)},
   volume={171},
   date={2010},
   number={3},
   pages={1401--1477},
}



\bib{M1983}{article}{
	author = {Marck, Jean-Alain},
	title = {Parallel-tetrad on null geodesics in Kerr–Newman space-time},
    journal = {Physics Letters A},
	volume = {97},
    year = {1983},
	number = {4},
	pages = {140-142},
}


\bib{PV2021}{article}{
   author={Petersen, Oliver},
   author={Vasy, Andr\'as},
   title={Analyticity of quasinormal modes in the Kerr and Kerr--de~Sitter
   spacetimes},
   journal={Comm. Math. Phys.},
   volume={402},
   date={2023},
   number={3},
   pages={2547--2575},
}



\bib{PV2023}{article}{
   author={Petersen, Oliver},
   author={Vasy, Andr\'as},
   title={Wave equations in the Kerr--de~Sitter spacetime: the full
   subextremal range},
   journal={J. Eur. Math. Soc. (JEMS)},
   volume={27},
   date={2025},
   number={8},
   pages={3497--3526},
}


\bib{PV2024}{article}{
   author={Petersen, Oliver},
   author={Vasy, Andr\'as},
   title={Stationarity and Fredholm theory in subextremal Kerr--de~Sitter
   spacetimes},
   journal={SIGMA Symmetry Integrability Geom. Methods Appl.},
   volume={20},
   date={2024},
   pages={Paper No. 052, 11},
 }

\bib{Pre2005}{article}{
author={Pretorius,Frans},
title={Evolution of Binary Black-Hole Spacetimes},
journal={Phys. Rev. Lett.},
volume={95},
pages={121101},
year={2005}
}


\bib{Rin2008}{article}{
   author={Ringstr\"om, Hans},
   title={Future stability of the Einstein-non-linear scalar field system},
   journal={Invent. Math.},
   volume={173},
   date={2008},
   number={1},
   pages={123--208},
}


\bib{SR1989}{article}{
   author={Saint Raymond, Xavier},
   title={A simple Nash–Moser implicit function theorem},
   journal={Enseign. Math.},
   volume={35},
   date={1989},
   pages={217–226},
 }

\bib{Sh2023}{article}{
   author={Shen, Dawei},
   title={Construction of GCM hypersurfaces in perturbations of Kerr},
   journal={Ann. PDE},
   volume={9},
   date={2023},
   number={1},
   pages={Paper No. 11, 112},
}


\bib{Shl2015}{article}{
   author={Shlapentokh-Rothman, Yakov},
   title={Quantitative mode stability for the wave equation on the Kerr
   spacetime},
   journal={Ann. Henri Poincar\'e},
   volume={16},
   date={2015},
   number={1},
   pages={289--345},
}

\bib{D1981}{article}{
   author={DeTurck, Dennis M.},
   title={Existence of metrics with prescribed Ricci curvature: local
   theory},
   journal={Invent. Math.},
   volume={65},
   date={1981/82},
   number={1},
   pages={179--207},
}



\bib{V2013}{article}{
	author={Vasy, Andr\'{a}s},
	title={Microlocal analysis of asymptotically hyperbolic and Kerr--de~Sitter spaces (with an appendix by Semyon Dyatlov)},
	journal={Invent. Math.},
	volume={194},
	date={2013},
	number={2},
	pages={381--513},
      }


\bib{W1989}{article}{
   author={Whiting, Bernard F.},
   title={Mode stability of the Kerr black hole},
   journal={J. Math. Phys.},
   volume={30},
   date={1989},
   number={6},
   pages={1301--1305},
}

\bib{WZ2011}{article}{
	author = {Wunsch, Jared},
	author={Zworski, Maciej},
	title = {Resolvent estimates for normally hyperbolic trapped sets},
	journal = {Ann. Henri Poincar\'e},
	volume = {12},
	year = {2011},
	number = {7},
	pages = {1349--1385},
}


\bib{Z2012}{book}{
   author={Zworski, Maciej},
   title={Semiclassical analysis},
   series={Graduate Studies in Mathematics},
   volume={138},
   publisher={American Mathematical Society, Providence, RI},
   date={2012},
   pages={xii+431},
}


\end{biblist}
\end{bibdiv}
\end{document}